\pdfoutput=1
\documentclass[journal]{IEEEtran} 
\IEEEoverridecommandlockouts       

\usepackage{amsthm}
\usepackage{amsmath,amssymb}
\usepackage{graphicx}
\usepackage{latexsym,url}
\usepackage{sty/algorithm}
\usepackage[noend]{sty/algorithmic}
\usepackage{cite}
\usepackage{comment}
\usepackage{stfloats}
\usepackage{color}
\usepackage{xcolor}
\usepackage[hidelinks]{hyperref}
\usepackage[caption=false,font=footnotesize]{subfig}

\usepackage{fancyhdr}

\newtheorem{theorem}{Theorem}
\newtheorem{assumption}{Assumption}
\newtheorem{lemma}{Lemma}
\newtheorem{remark}{Remark}

\newtheorem{definition}{Definition}

\newcommand{\N}{\mathcal{N}}

\newcommand{\mc}{\mathcal}
\newcommand{\mrm}{\mathrm}

\newcommand{\bN}{{\mathbb N}}
\newcommand{\R}{{\mathbb R}}

\newcommand{\diag}{{\rm diag}}
\newcommand{\be}{{\bf e}}
\newcommand{\lam}{\lambda}
\newcommand{\bfp}{\mathbf{p}}

\newcommand{\pers}{\mathrm{pers}}
\newcommand{\res}{\mathrm{res}}
\newcommand{\rd}{\mathrm{d}}
\newcommand{\bd}{\mathrm{bd}}
\newcommand{\sgn}{\mathrm{sgn}}
\newcommand{\co}{\mathrm{co}}

\newcommand{\argmin}{\mathop{\rm arg~min}\limits}

\newcommand{\blue}[1]{\textcolor{blue}{#1}}

\begin{document}

%\title{Distributed Coverage Maintenance for Quadcopters via Nonsmooth Control Barrier Functions}
\title{Distributed Coverage Hole Prevention for \\Visual Environmental Monitoring with Quadcopters\\ via Nonsmooth Control Barrier Functions}

\author{Riku Funada$^{1}$, Mar\'ia Santos$^{2}$, Ryuichi Maniwa$^{1}$, Junya Yamauchi$^{3}$,\\
Masayuki Fujita$^{3}$, Mitsuji Sampei$^{1}$, and Magnus Egerstedt$^{4}$

%\begin{tikzpicture}[remember picture, overlay]
%      \node[minimum width=4in,font=\bfseries] at ([yshift=-1cm]current page.north)  {This paper has been accepted for publication in the IEEE Transactions on Robotics.
%Please cite the paper as: R. Funada, M. Santos, R. Maniwa, J. Yamauchi, M. Fujita, M. Sampei, and M. Egerstedt,\newline
%``Distributed Coverage Hole Prevention for Visual Environmental Monitoring with Quadcopters via Nonsmooth Control Barrier Functions," 2023.};
%    \end{tikzpicture}

\thanks{\copyright 2023 IEEE.  Personal use of this material is permitted.  Permission from IEEE must be obtained for all other uses, in any current or future media, including reprinting/republishing this material for advertising or promotional purposes, creating new collective works, for resale or redistribution to servers or lists, or reuse of any copyrighted component of this work in other works.}
\thanks{*This work was supported by JSPS KAKENHI Grant Number 22K14275. The work by Magnus Egerstedt was supported by grant CNS-2233783 from the US National Science Foundation.}
\thanks{$^{1}$R. Funada, R. Maniwa and M. Sampei are with the Department of Systems and Control Engineering, Tokyo Institute of Technology, Tokyo 152-8550, Japan
{\tt\footnotesize \{\href{mailto:funada@sc.e.titech.ac.jp}{funada},\href{mailto:sampei@sc.e.titech.ac.jp}{sampei}\}@sc.e.titech.ac.jp},
\href{mailto:maniwa@sl.sc.e.titech.ac.jp} {\tt\footnotesize {maniwa}@sl.sc.e.titech.ac.jp}}
\thanks{$^2$M. Santos is with the Department of Mechanical and Aerospace Engineering, Princeton University, Princeton, NJ 08544, USA 
\href{mailto:maria.santos@princeton.edu}{\tt\footnotesize {maria.santos}@princeton.edu}}
\thanks{$^{3}$J. Yamauchi and M. Fujita are with the Department of Information Physics and Computing, The University of Tokyo, Tokyo 113-8656, Japan 
{\tt\footnotesize \{\href{mailto:junya_yamauchi@ipc.i.u-tokyo.ac.jp}{junya\_yamauchi},\href{mailto:masayuki_fujita@ipc.i.u-tokyo.ac.jp}{masayuki\_fujita}\}@ipc.i.u-tokyo.ac.jp}}
\thanks{$^{4}$M. Egerstedt is with the Department of Electrical Engineering and Computer Science, University of California Irvine,
Irvine, CA 92697, USA
\href{mailto:magnus@uci.edu}{\tt\footnotesize {magnus}@uci.edu}}
}

\maketitle
\thispagestyle{fancy}
\fancyhead{}
\fancyhead[C]{This paper has been accepted for publication in the IEEE Transactions on Robotics.\\
%Please cite the paper as: 
%R. Funada, M. Santos, R. Maniwa, J. Yamauchi, M. Fujita, M. Sampei, and M. Egerstedt, ``Distributed Coverage Hole Prevention for Visual Environmental Monitoring with Quadcopters via Nonsmooth Control Barrier Functions," IEEE Transactions on Robotics, 2023 (Early Access). 
Citation information: DOI 10.1109/TRO.2023.3347132 \hspace{5mm}
URL: \url{https://ieeexplore.ieee.org/document/10374239}}
\fancyhead[RO]{\footnotesize \thepage}
\IEEEpeerreviewmaketitle

\begin{abstract}
This paper proposes a distributed coverage control strategy for quadcopters equipped with downward-facing cameras that prevents the appearance of unmonitored areas in between the quadcopters' fields of view (FOVs). We derive a necessary and sufficient condition for eliminating any unsurveilled area that may arise in between the FOVs among a trio of quadcopters by utilizing a power diagram, i.e. a weighted Voronoi diagram defined by radii of FOVs. Because this condition can be described as logically combined constraints, we leverage nonsmooth control barrier functions (NCBFs) to prevent the appearance of unmonitored areas among a team's FOV. We then investigate the symmetric properties of the proposed NCBFs to develop a distributed algorithm. The proposed algorithm can support the switching of the NCBFs caused by changes of the quadcopters composing trios. The existence of the control input satisfying NCBF conditions is analyzed by employing the characteristics of the power diagram. The proposed framework is synthesized with a coverage control law that maximizes the monitoring quality while reducing overlaps of FOVs. The proposed method is demonstrated in simulation and experiment.
\end{abstract}

\section{Introduction}

\begin{figure}
    \centering
    \includegraphics[width=0.75\linewidth]{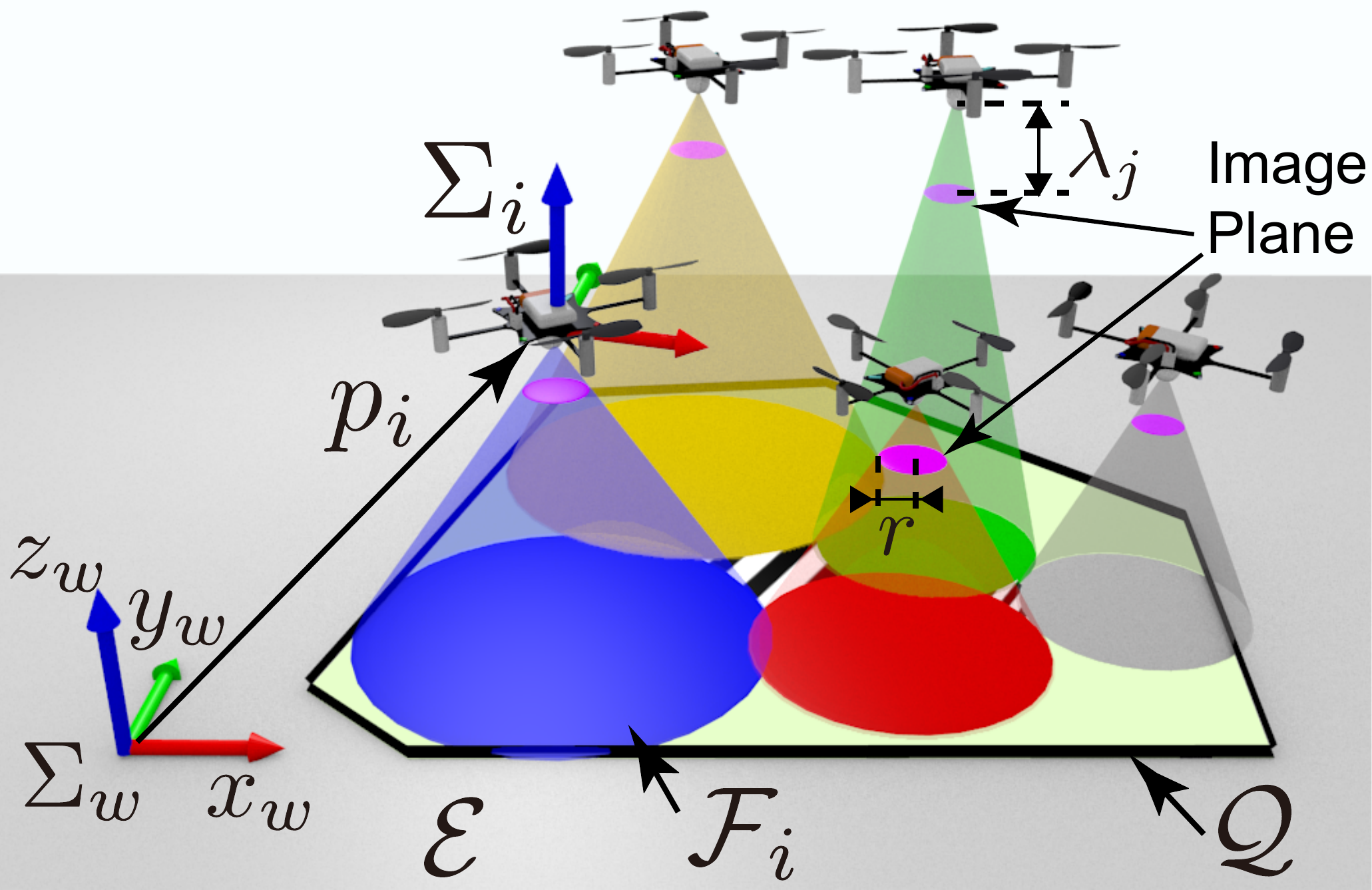}
    \caption{Proposed scenario. The team of quadcopters monitors a mission space $\mc Q$ specified within the environment $\mc E$.
    Each quadcopter mounts a downward-facing camera with a circular field of view (FOV) that is adjustable via its position $p_i$ and focal length $\lambda_i$. 
    The hatched regions in between the FOVs of quadcopters represent unmonitored areas to be prevented.}
    \label{fig:scenario}
\end{figure}

Visual environmental monitoring has been utilized to collect information in a broad range of applications such as urban traffic \cite{Zhou15}, home/facility security \cite{Fleck08}, terrain data \cite{Cole09}, and natural phenomena \cite{Zhou2005,Erdelj17}.
Among the various types of robots utilized in these scenarios, unmanned aerial vehicles with visual sensors afford capturing vast environments from the sky. 
In addition, some unmanned aerial vehicles (e.g., quadcopters) are authorized to fly at relatively lower altitudes \cite{Shakhatreh19}, which makes them suitable for capturing information at a closer distance to a surveilled area, which results in more detailed information. However, flying at a lower altitude typically leads to a smaller field of view (FOV) of the quadcopter, which in turn can make difficult to capture large-scale phenomena in vast environments. 

The restrictions associated with the FOV of a single quadcopter can be addressed by employing a team of quadcopters to cooperatively capture extensive environments. 
For instance, the research in \cite{Meng15} and \cite{Ruiz18} presents a surveillance system that can produce a single high-resolution panoramic video of a large target area in real-time by stitching together multiple videos from quadcopters. 
Furthermore, several empirical studies suggest that wider FOVs result in higher performance for human operators in human-robot interaction \cite{PREWETT10}. Thus, such collaborative monitoring systems could contribute to better situation awareness compared with receiving multiple disjoint videos provided individually by the quadcopters. 
However, the advantages provided by the joint monitoring efforts of a group of drones can be hindered by the potential appearance of unmonitored areas within the target environment, as illustrated in Fig.~\ref{fig:scenario}. 
Even though a quadcopter flies at a relatively lower altitude than conventional aerial vehicles, the altitude of quadcopters in environmental monitoring could exceed $100\,{\rm m}$ in some applications, e.g., \cite{Boon17}. In such cases, the size of the quadcopters' FOVs becomes large, resulting in the appearance of a non-negligible unmonitored area between FOVs.
Unmonitored areas are undesirable for situational awareness and search missions since they can lead to missing important events or failures when stitching videos together. 
Besides, such appearances of unmonitored areas among a team could result in overlooking potential hazards in safety-critical applications requiring real-time information, such as the monitoring of construction sites to prevent accidents \cite{Ham16}.

This paper presents a coverage strategy for a team of quadcopters that prevents the appearance of unmonitored areas in-between FOVs of the quadcopters. We achieve this goal by modifying the nominal input maximizing the monitoring performance minimally invasive way so that the appearance of unmonitored areas is restricted. Specifically, we develop a distributed algorithm that employs nonsmooth control barrier functions (NCBFs), logically combined control barrier functions (CBFs), to prevent the appearance of a hole. Then, as the nominal input utilized with the algorithm, we propose the coverage control law that maximizes the monitoring performance of the team according to the perspective quality of the quadcopters' cameras while expanding the surveyed area of a team by avoiding redundant observation of the multiple quadcopters.

\subsection{Related Work}
Cooperative environmental monitoring has been studied extensively in the context of coverage control, which yields an optimal deployment of mobile sensors with respect to a spatial field in a distributed fashion \cite{Cortes2004,Cassandras2005,Martinez2007}. 
The coverage control problem has been considered in a variety of scenarios, 
including networks of Pan-Tilt-Zoom (PTZ) cameras \cite{Ding2012,Forstenhaeusler2015,Arslan2018,Hatanaka20} as well as teams of ground \cite{Laventall2008,Gusrialdi2008,Dimitra15,Chen21} and aerial robots \cite{Schwager2011,Papatheodorou2017,Bousias19,Renzaglia20}. 

In the context of aerial robots, the work in \cite{Schwager2011} presents a coverage strategy for a heterogeneous robot team where the cameras present various degrees of freedom depending on the robot type, e.g., PTZ cameras, quadcopters, etc. 
The proposed control law captures a 2D environment with a maximal resolution, where a better coverage quality is allocated to an area monitored by multiple cameras as opposed to the same area covered by a single camera.
The work \cite{Papatheodorou2017} proposes a solution that reduces FOV overlaps while guaranteeing the altitude of quadcopters within the predefined range.
A more challenging scenario is considered in \cite{Bousias19}, where each quadcopter is mounted with a PTZ camera, making their FOVs become ellipses on the 2D environment. 
The work \cite{Renzaglia20} presents the unified framework for both exploration and coverage problems in an unknown 3D environment with quadcopters. Both problems are formulated to maximize the portion of observed environments while minimizing the overlap of FOVs. 
Despite the variety of advances presented in these works, the proposed controllers do not address the appearance of unmonitored areas in between FOVs, which may result in overlooking critical events occurring in the environment.

The work in \cite{Mahboubi2014,Mahboubi2017} reduces the unmonitored area in the whole environment, rather than eliminating unobserved regions in-between FOVs, by maximizing the area of the covered region. Both works consider the agents with heterogeneous and fixed sensing regions. The uncovered areas are detected by utilizing a multiplicatively weighted Voronoi diagram having curved boundaries, where the weight of each sensor is equal to its sensing radius. The algorithms guarantee increased area coverage, but the elimination of unmonitored areas is not assured. The rejection of unmonitored regions is also considered in \cite{So-In19}, where the total moving distance of mobile sensors is minimized. Still, this work assumes the sensing radius of all agents is fixed and identical. These methods, which assume a fixed or same measurement range, might not be readily extendable to cooperative monitoring with a team of quadcopters, where the size of FOVs dynamically evolves according to the changes of the zoom level or altitude.
The work \cite{Cortes2006} proposes a solution based on a nonsmooth gradient flow for a disk-covering problem, where each agent can control its sensing radius. Because a disk-covering problem corresponds to completely covering the entire environment with disks of minimum radius, the proposed control law eliminates unobserved areas. The resulting final deployment renders the radius of sensing fields of all agents identical. However, it is difficult to cover the entire environment with high-resolution images in case the whole mission space is large compared with FOVs of quadcopters, or the number of quadcopters is limited. 
For instance, the monitoring area of \cite{Shah20} reaches $2\,\textrm{km}^2$.
In such cases, rather than overextending FOVs for covering a whole area, preventing unmonitored areas around targeted regions under a team's surveillance could yield more detailed information for contributing to the monitoring purposes.
Furthermore, if the environment has non-uniform importance, allowing each quadcopter to decide on its sensing radius could lead to better monitoring performance because a team can adjust the resolution of the image according to the importance of the area.
To achieve such a visual monitoring strategy, the team should maximize the objective function for visual monitoring quality, as in \cite{Schwager2011,Papatheodorou2017,Bousias19,Renzaglia20}. However, in case the reduction or elimination of unmonitored areas is already encoded as the objective of the team as in \cite{Mahboubi2014,Mahboubi2017,So-In19,Cortes2006}, the problem becomes a multi-objective optimization, with the associated challenges of balancing two different requirements.

In order to address the aforementioned issue, this paper prevents the appearance of a hole by means of a constraint-based approach, where the visual monitoring quality is maximized as long as unmonitored areas in-between a team are prevented. More specifically, the behavior of the quadcopters is constrained to prevent unsurveilled areas in-between quadcopters by CBFs. 
CBFs provide the forward invariance property to the admissible set of the robot state space \cite{Egerstedt21,Ames2017,Ames2019_CBF_thapp}.
CBFs have been employed in various applications ranging from bipedal robots \cite{Nguyen16,Ahmadi22} to swarm robots \cite{Wang17_swarm,Funada20_CSM,Srinivasan21,Nishimoto22} and seeing success also in the context of environmental monitoring \cite{Notomista21,Santos2019mrs,Dan21}. 
The work \cite{Notomista21} proposes a control framework that guarantees the robot's battery is never depleted during a monitoring task. Several charging models, such as solar lights and charging stations, are interpreted as a time-varying function defined in the mission space. Then, a time-varying CBF is developed to ensure energy constraints of the robots. 
The work \cite{Santos2019mrs} also utilizes the time-varying CBF to consider time-varying coverage control, where a density function specifying the importance of field changes dynamically. 
The proposed methods do not impose any requirements on the rate of change of the density functions and can be executed in a distributed fashion. 
In \cite{Dan21}, a notion of information reliability is encoded through a time-varying density function as in \cite{Hubel2008}, which decays if a region is not surveilled for a long period. Then, a patrolling motion is achieved by a time-varying CBF together with other safety constraints. 
In \cite{Wu16,Wang17}, an extended version of CBFs are utilized to a quadcopter having nonlinear dynamics.

\subsection{Proposed Approach and Contributions}

In this paper, we propose a control method that prevents unmonitored areas to arise in-between FOVs of quadcopters equipped with downward facing cameras while monitoring a mission area. First, CBFs preventing the appearance of unmonitored areas in between a team is designed by utilizing a power diagram. The power diagram \cite{Aurenhammer1987} is a weighted Voronoi diagram defined by the different radii of FOVs, where the boundaries of the diagram are composed of lines known as the radical axis. 
We leverage this favorable property of the power diagram, compared with the multiplicatively weighted Voronoi partitions in \cite{Mahboubi2014,Mahboubi2017} that generate more complicated Voronoi cells composed of Appollonian circles, to analyze its feasibility and develop distributed algorithm.
Furthermore, we provide the coverage control law that intends to monitor important regions while preventing too much overlap between the fields of view of quadcopters. 
The presented coverage controller is synthesized with the aforementioned CBFs to achieve the maximum visual monitoring quality while preventing the unmonitored areas in-between a team.

The contributions of this paper are as follows.
%
% Contribution
\begin{itemize}
% First Contribution
\item[(a)]  
We develop a distributed algorithm preventing an unmonitored area in between a team of quadcopters based on CBFs. 
For this goal, the appearance of unsurveilled areas is encoded by the vertices of a power diagram. 
Because the derived constraints contain a Boolean logic, we leverage a nonsmooth control barrier function (NCBF) \cite{Glotfelter18} to synthesize the controller that prevents an unmonitored area. 
We then formulate an algorithm that allows the jumps in the value of NCBF caused by the changes in the Delaunay graph induced by the power diagram.
%
% Second Contribution
\item[(b)] 
We verify the feasibility of the proposed CBF-based methods. 
The existence of a control input that satisfies the constraints is closely related to how the CBF evolves together with the system dynamics. For example, in \cite{Notomista18,Nishimoto22}, feasibility was established by showing that the gradient of the CBF does not vanish. In our case, this computation is not feasible and instead, we rely on the special structure of the power diagram to derive a simplified form of the designed CBFs' derivatives. Then, we show that the set yielding the derivatives of CBF zero vectors can be neglected in practice. This result guarantees that the proposed CBF-based controller can generically generate the control input to prevent an unmonitored area in between a team during monitoring mission.
\item[(c)]
The coverage control law in our conference paper \cite{Funada19} is combined with the above CBF-based methods to showcase its utility.
This visual coverage control serves two purposes: (i) improving the performance of the monitoring, considering the spatial sensing quality of cameras as in \cite{Arslan2018}, and (ii) reducing the overlap between FOVs of the quadcopters to cover a large-scale area by a team.
\item[(d)]
The practical usefulness of the proposed monitoring strategy is demonstrated by simulations and experiments. 
\end{itemize}

In our previous work \cite{Funada20}, we proposed the preliminary results of (a), which does not allow a distributed computation without a strong assumption. 
More specifically, when the algorithm calculates the derivatives of CBFs for Quadcopter~$i$, it presumes the other quadcopters are fixed in the environment. 
In this paper, the algorithm is improved to achieve a distributed computation without the assumption by utilizing the symmetric properties of the proposed CBFs. 
We newly analyze the validity of the proposed NCBFs in (b), which was not discussed in our previous conference papers \cite{Funada20,Funada19}.
The coverage control law discussed in (c) is proposed in our previous work \cite{Funada19}, but we only show the utility of the control law with the preliminary algorithm that overly restricts a teams's movement for preventing the appearance of unmonitored area in between quadcopters. 
Finally, this paper provides comprehensive simulation and experimental results demonstrating the effectiveness of the proposed methods.

\section{Preliminaries} \label{sec:pre}

In this section, we introduce several variations of CBFs, which can synthesize logically combined constraints in a system exhibiting graph changes. The methods will be needed to formulate the control algorithms preventing coverage holes in-between agents in Section~\ref{sec:hole_prev}.

Let us consider a control-affine system
\begin{align} \label{eq:control_affine}
    \dot x(t) = f(x(t)) + g(x(t)) u(x(t), t),~x(t_0),~t_0 \in \R, 
\end{align}
where $f: \R^n\to\R^n$ and $g: \R^n\to \R^m$ are locally Lipschitz continuous while $u: \R^n \times \R \to \R^m$ is measurable and locally bounded in both arguments.
In this formulation, $u$ is not necessarily continuous. Therefore, the solution might not exist for the differential equation \eqref{eq:control_affine}.
To yield a system to which solutions exist, a system \eqref{eq:control_affine} can be turned into a differential inclusion via Filippov's operator $K[f+gu]:\R^n \times \R \to 2^{\R^n}$ 
\begin{align} \label{eq:diff_inc}
\dot x(t) \in K[f+gu](x(t),t), x(t_0) = x_0, t_0 \in \R,
\end{align}
with 
\begin{align}
&K[f+gu](x', t') \nonumber\\
&=\!\co\! \left\{ \lim_{i\to \infty} f(x_i)\!+\! g(x_i) u(x_i, t'): x_i \!\to\! x', x_i \!\notin\! S_f, S \right\},
\end{align}
where $S_f$ is a particular zero-Lebesgue-measure set and $S$ is an arbitrary zero-Lebesgue-measure set.
As detailed in \cite{Cortes2008,Glotfelter19}, a Carath\'{e}odory solution, an absolutely continuous function $x:[t_0, t_1] \to \R^n$ such that $\dot x(t) \in K[f+gu](x(t),t)$ almost everywhere on $t\in [t_0, t_1]$, always exists in the above system.

Hereafter, we introduce the CBF-based methods that confine the state of the system \eqref{eq:diff_inc} in the admissible set. 
While the original papers \cite{Glotfelter18,Glotfelter19} assume a more general model than \eqref{eq:diff_inc}, we proceed with \eqref{eq:diff_inc} as it is sufficient for our objective. 
In addition, we sometimes modify the notations in \cite{Glotfelter18,Glotfelter19} to suit our paper.
Note that, in practice, the method introduced in this section does not require the computation of Filippov's operator; instead, more convenient algorithm is formulated.

CBFs have been employed for guaranteeing the forward invariance property of the set $\mc C$, which represents constraints to be satisfied during the execution of the task.
%in which the state $x$ should be confined during the execution of the task. 
We assume that a set $\mc C$ can be defined as the superlevel set of a barrier function $h(x,t): \mc D \subset \R^n \times \R_{\geq t_0} \to \R$, namely $\mc C = \{ (x', t') \in \R^n \times \R_{\geq t_0} \mid h(x', t') \geq 0 \}$. 
Then, the forward invariance of the set $\mc C$ can be accomplished by ensuring
\begin{align} \label{eq:CBF_cond_default}
    \dot h(x(t), t) \geq -\alpha(h(x(t), t)),~a.e.~t \in [t_0, t_1],
\end{align}
for every Carath\'{e}odry solution, 
for some locally Lipschitz extended class-$\mc K$ function $\alpha : \R \to \R$ \cite{Glotfelter19}.

Let us formulate the time-varying nature of \eqref{eq:CBF_cond_default} to better suit the problem we consider, where the constraints are added or subtracted according to the changes in the network graph of a team.
First, the following two definitions are necessitated.
%Switching sequence
\begin{definition} \cite[Definition 1]{Glotfelter19} \label{def:sw_sq}
A sequence $\{\tau^{\rm k}\}_{\rm k=1}^{\infty}$ is a switching sequence for \eqref{eq:diff_inc}
%\eqref{eq:si_dyn} 
if and only if it is strictly increasing, unbounded, and $\tau^1 = t_0$. With respect to $\{\tau^{\rm k}\}_{\mrm k=1}^{\infty}$, let 
\begin{align}
    {\rm K}_{t_1} = \inf \left\{ {\rm K} \in \bN \;\middle|\; [t_0, t_1] \subset \bigcup_{\mrm k=1}^{\mrm K} [\tau^{\mrm k}, \tau^{\mrm k+1} )\right\}.
\end{align}
\end{definition}
\begin{definition}  \cite[Definition 2]{Glotfelter19} \label{def:hyb_fw_inv}
A set $\mc C \subset \R^n \times \R$ is hybrid forward invariant with respect to \eqref{eq:diff_inc} and a switching sequence $\{\tau^{\mrm k}\}_{\mrm k=1}^\infty$ for \eqref{eq:diff_inc} if and only if for every Carath\'{e}odry solution starting from $x_0$ at $t_0$,
\begin{align}
    (x(\tau^{\mrm k}), \tau^{\mrm k}) \in \mc C, \forall {\mrm k} \leq \mrm K_{t_1} \Rightarrow (x(t), t) \in C, \forall t \in [t_0, t_1]. \nonumber
\end{align}
\end{definition}
Definition~\ref{def:hyb_fw_inv} signifies the desirable behavior of the system. 
Namely, if the state is inside of the admissible set $\mc C$ at the initial time and the state does not leave the set at any jump, then the state remains in the set.
In the scenario introduced later on, the admissible set for Quadcopter $i$ switches every time its neighboring agents change.
Therefore, by achieving the hybrid forward invariance of the set prohibiting the unsurveilled area in-between a team, the unmonitored area does not appear if the changes in the graph make no holes instantaneously.
Furthermore, even if a change in the graph results in the appearance of a hole, such a hole can be eliminated by the general property of CBFs that yields the safe set attractive \cite{Glotfelter19,Ames2019_CBF_thapp}. This will be demonstrated in simulation in Section~\ref{ssec:sim_nine}.

%To better resemble
To capture the above requirements, we presume $h$ is not time-varying during each switching interval $t' \in [\tau^{\mrm k}, \tau^{\mrm k+1})$, which can be modeled with $h^{\mrm k} : \mc D^{\mrm k} \subset \R^n \to \R$ as 
\begin{align} \label{eq:HNCBF_NCBF}
    h(x', t') = h^{\mrm k}(x'),~\forall {\mrm k} \in \bN, \forall t' \in [\tau^{\mrm k}, \tau^{\mrm k+1}), \forall x'\in \R^n,
\end{align}
where the superlevel set of $h^{\mrm k}$ is denoted as $\mc C^{\mrm k} = \{ x' \in \R^n \mid h^{\mrm k}(x') \geq 0 \}$.

In case $\mc C$ is a zero-superlevel set of $h$ in \eqref{eq:HNCBF_NCBF}, the properties of $h^{\mrm k}$, especially its time derivative $\dot h^{\mrm k}$ as appeared in \eqref{eq:CBF_cond_default} during $t' \in [\tau^{\mrm k}, \tau^{\mrm k+1})$, need to be investigated to guarantee the hybrid forward invariance of the set $\mc C$.
Furthermore, if there exists nonsmoothness in the barrier function $h^{\mrm k}$, e.g., as a result of logically combined constraints as in this paper, calculating $\dot h^{\mrm k}$ can no longer be possible via the usual chain rule.

The barrier function embracing logically combined requirements is proposed in \cite{Glotfelter18} as nonsmooth control barrier functions (NCBFs). The result of \cite{Glotfelter18} revealed that the Boolean operations $\land$, $\lor$, and $\lnot$ are formulated as 
\begin{equation} \label{eq:Bool_logic}
\begin{split}
    &h^{\mrm k}(x') = \min \{ h_1^{\mrm k}(x'), h_2^{\mrm k}(x') \} := h_1^{\mrm k} \land h_2^{\mrm k} \\
    &h^{\mrm k}(x') = \max \{ h_1^{\mrm k}(x'), h_2^{\mrm k}(x') \} := h_1^{\mrm k} \lor h_2^{\mrm k} \\
    &h^{\mrm k}(x') = -h_1^{\mrm k}(x') := \lnot h_1^{\mrm k}.
\end{split}
\end{equation}
at each $x' \in \mc D$. % with denoting $h(x')$ as a Boolean NCBF (BNCBF).
Let us assume that the NCBF $h^{\mrm k}$ described in \eqref{eq:Bool_logic} has the following favorable properties. Note that the definition from \cite{Glotfelter18} has been modified to suit the terminology of this paper.
\begin{definition}\cite[Definition~2]{Glotfelter18} \label{def:candidate_NCBF}
A locally Lipschitz function $h^{\mrm k}: \mc D^{\mrm k} \subset \R^n \to \R$, where $\mc D^{\mrm k}$ is an open and connected set, is a candidate NCBF if the set $\mc C^{\mrm k}$ is nonempty.
\end{definition}
\begin{definition}\cite[Definition~6]{Glotfelter18} \label{def:sm_comp}
    A candidate NCBF $h^{\mrm k}: \mc D^{\mrm k} \subset \R^n \to \R$ composed by the Boolean operator as in \eqref{eq:Bool_logic} is smoothly composed if each component candidate NCBF $h_i^{\mrm k} : \mc D^{\mrm k} \subset \R^n \to \R$ is continuously differentiable. 
\end{definition}
For ensuring the forward invariance of the set $\mc C^{\mrm k}$ described by a smoothly composed candidate NCBFs, the work \cite{Glotfelter18} introduces the almost-active set and the almost-active gradient. 
%
%
% D Definition
\begin{definition} \cite[Definition~7]{Glotfelter18} \label{def:alm_act_grad}
Given $\epsilon > 0$ and two candidate NCBFs $h_1^{\mrm k}, h_2^{\mrm k}: \mc D^{\mrm k} \subset \R^n \to \R$, the almost-active set of functions for a candidate NCBF given by $h^{\mrm k}=h_1^{\mrm k} \land h_2^{\mrm k}$ or $h^{\mrm k} = h_1^{\mrm k} \lor h_2^{\mrm k}$ is defined at each $x' \in \mc D^{\mrm k}$ as
\begin{align} \label{eq:index_set}
    I_{\epsilon}^{\mrm k}(x') = \left\{ i \mid \| h_i^{\mrm k}(x') - h^{\mrm k}(x') \| \leq \epsilon \right\}.
\end{align}
The almost-active gradient of an NCBF, denoted by $\partial_\epsilon h^{\mrm k}:\mc D^{\mrm k} \subset \R^n \to 2^{\R^n}$, at a point $x' \in \mc D^{\mrm k}$ is
\begin{align}
    \partial_{\epsilon} h^{\mrm k}(x') = \co\bigcup_{i \in I_\epsilon^{\mrm k}(x')} \partial h_i^{\mrm k}(x').
\end{align}
\end{definition}

The requirement \eqref{eq:CBF_cond_default} expressed with the almost-active gradient is integrated as a constraint in a quadratic programming (QP) problem to generate the control input guaranteeing the forward invariance of the set $\mc C^{\mrm k}$.
If the formulated QP can always generate the control input, the designed NCBF is denoted as a valid NCBF as formally defined in the following theorem, where $\langle\cdot,\cdot\rangle$ denotes a set-valued inner product.

\begin{theorem} \cite[Theorem~3]{Glotfelter18} \label{thm:Valid_NCBF}
Let $h^{\mrm k}:\mc D^{\mrm k} \subset \R^n \to \R$ be a smoothly composed candidate NCBF, as in Definition~\ref{def:alm_act_grad}. If there exists $\epsilon > 0$ and a locally Lipschitz extended class-$\mc K$ function $\alpha: \R\to \R$ such that the QP
\begin{subequations}
\begin{align}
    u^{\mrm k} &\in \argmin_{u\in \R^m} u^\top A(x')u + b(x')^\top u  \\
    &~\mathrm{s.t.}~ \left\langle \partial_{\epsilon}h^{\mrm k}(x'), f(x')+g(x')u \right\rangle \geq - \alpha(h^{\mrm k}(x')),
\end{align}
\end{subequations}
with $A: \mc D^{\mrm k} \subset \R^n \to \R^{m\times m}$, $b: \mc D^{\mrm k} \subset \R^n\to \R^m$ continuous, has a solution for every $x'\in \mc D^{\mrm k}$ and $u^{\mrm k}$ is measurable and locally bounded, then $h^{\mrm k}$ is a valid NCBF for \eqref{eq:diff_inc}.
\end{theorem}

\begin{remark} \label{rem:QP_feas}
In practice, both the almost-active gradient and a set-valued inner product do not need to be evaluated in the QP. 
Instead, we can formulate constraints by including the gradients of all functions that constitute the index set $I_{\epsilon}^{\mrm k}(x')$ defined in \eqref{eq:index_set}, eliminating the convex-hull operation as discussed in \cite{Glotfelter19}. 
To see how this works, let us consider the system described as a single integrator, namely $f = {\bf 0}_{n\times 1}$ and $g = {\bf 1}_{m\times 1}$, the same model we assume for quadcopters as will be introduced in \eqref{eq:si_dyn}. 
Then, instead of the QP in Theorem~\ref{thm:Valid_NCBF}, we can employ 
%\red{Then, the QP in Theorem~\ref{thm:Valid_NCBF} can be described as}
\begin{subequations} \label{eq:QP_NCBF_valid}
\begin{align} 
    %u^*(x) &= \argmin_{u_i\in \R^4} (u_i - u_{i, {\rm nom}})^\top W (u_i - u_{i, {\rm nom}}), \nonumber \\
    u^{\mrm k} &= \argmin_{u\in \R^m} u^\top A(x')u + b(x')^\top u  \\
    &~{\rm s.t.}~ \left.\frac{\partial h_i^{\mrm k}}{\partial x}\right|_{x=x'}^\top u \geq - \alpha(h^{\mrm k}(x')),~\forall i\in I_{\epsilon}^{\mrm k}(x'). 
\end{align}
\end{subequations}
Similar to \cite{Glotfelter18}, let us assume that the resultant controller is measurable and locally bounded. 
Then, to guarantee the QP has a solution for every $x' \in \mc D^{\mrm k}$ in the considered problem setting, namely the validness of the NCBF, we can investigate the following sufficient conditions: 1) $\mc C^{\mrm k}$ is not an empty set, 2) $\partial h_i^{\mrm k} / \partial x |_{x=x'} \neq {\bf 0}_{m\times 1}$, $\forall i\in I_{\epsilon}^{\mrm k}(x')$, $\forall x' \in \mc D^{\mrm k}$. 
To consider the second condition, we analyze the gradients of the component CBFs of the designed NCBF in Section~\ref{ssec:Valid_CBF}.
\end{remark}

Having defined the validity of the NCBF $h^{\mrm k}$ necessitated to assure the forward invariance of $\mc C^{\mrm k}$, we are now ready to introduce the requirement for $h$ in \eqref{eq:HNCBF_NCBF}.
The following theorem formulates a valid hybrid NCBF (HNCBF), which is required to keep the safe set $\mc C$ to be hybrid forward invariant.

\begin{theorem} \cite[Theorem~2]{Glotfelter19} \label{thm:paul_HNCBF}
Let $h^{\mrm k}:\R^n \to \R$, $\mrm k\in \bN$, be a collection of valid NCBFs %, in the sense of Definition~, 
for \eqref{eq:control_affine} under the control laws and extended class-$\mc K$ functions, $u^{\mrm k}:\R^n \to \R$, $\alpha^{\mrm k}: \R \to \R, \mrm k \in \bN$; and let $\{\tau^{\mrm k}\}_{\mrm k=1}^\infty$ be a switching sequence for \eqref{eq:control_affine}. Then, the function $h:\R^n \times \R_{\geq t_0} \to \R$ defined in \eqref{eq:HNCBF_NCBF} is a candidate HNBF. Moreover, $h$ is a valid HNCBF for \eqref{eq:control_affine} with the input $u: \R^n \times \R_{\geq t_0} \to \R^m$ defined as
\begin{align}
    u(x', t') = u^{\mrm k}(x'), \forall k \in \bN, \forall t' \in [\tau^{\mrm k}, \tau^{\mrm k+1}), \forall x' \in \R^n.
\end{align}
\end{theorem}
%\blue{The above theorem implies that we need to show the validness of NCBFs, e.g., by the discussion in Remark~\ref{rem:QP_feas}, to exhibit the validity of HCNBF.}
Theorem~\ref{thm:paul_HNCBF} implies that we can render the set $\mc C$ hybrid forward invariance by solving \eqref{eq:QP_NCBF_valid} at each switching sequence~$\mrm k$ with valid NCBFs $h^{\mrm k}, \mrm k \in \bN$ composing $h$.
In Section~\ref{ssec:Valid_CBF}, following the above discussion, we analyze the validity of the proposed NCBF, namely the feasibility of the developed distributed algorithm. Please note that although the proposed algorithm utilizes NCBFs, a similar category of CBFs to \cite{Glotfelter18,Glotfelter19}, the designed NCBF and its objectives are completely different from \cite{Glotfelter18,Glotfelter19}. Hence, analyzing the validity of the designed NCBFs necessitates new theoretical results, where we employ the properties of a weighted Voronoi partition. Furthermore, neither \cite{Glotfelter18} nor \cite{Glotfelter19} does not consider a distributed algorithm.

\section{Problem Statement} \label{sec:state}

\subsection{Quadcopters and Environment}
In this paper, we consider the scenario illustrated in Fig.~1, where $n$ quadcopters $\mc N = \{1,\cdots n \}$ equipped with a downward facing camera are distributed in the 3-D space to monitor a 2-D region $\mc Q$. 
We denote the region $\mc Q$ to be surveilled as the mission space hereafter, and assumed to be a closed and bounded convex subspace of the environment $\mc E \subset \R^3$.
The relative importance of points on $\mc E$ is specified by a density function $\phi: \mc E \to \R_{\geq0}$, where the higher the importance of a point $q\in \mc E$ is, the higher the value of $\phi(q)$, with $\phi(q) = 0,~q \in \mc E \backslash \mc Q$.
The goal of this paper is to monitor the important region while preventing the appearance of unmonitored areas in-between a team of quadcopters, shown as hatched areas in Fig.~\ref{fig:scenario}.

The world coordinate frame $\Sigma_w$ is placed so that its $x_wy_w$-plane is coplanar with the environment $\mc E$, where $\{\be_x,\be_y,\be_z\}$ describes the standard basis of $\Sigma_w$.
Then, the set of coordinates of all points in the environment is defined as $\mc E := \{ q\in \R^3 \mid \be_z^\top q = 0 \}$.
We suppose that all quadcopters stay in the common half space $\{ q \in \R^3 \mid \be_z^\top q > 0 \}$ during the operation.

Let $\Sigma_i$ be the coordinate frame of Quadcopter $i \in \mc N$, where its position vector with respect to $\Sigma_w$ is $p_i = [x_i~y_i~z_i]^\top$ with $z_i \in \R_{>0}$, and its attitude is fixed to be consistent with that of $\Sigma_w$ regardless of the attitude of Quadcopter~$i$.
%We assume each quadcopter mounts a camera with a gimbal that keeps an optical axis of the camera corresponds with the $Z_i$-axis of $\Sigma_i$.
We assume each quadcopter equips with a gimbal that can keep an optical axis of the camera corresponds with the axis of the $z_i$-axis of $\Sigma_i$.
Namely, Quadcopter~$i$ points its camera to the downward direction.
In addition, each quadcopter is assumed to be equipped with localization devices, such as the GNSS, gyro sensors, and IMU, to estimate its own pose accurately, as in \cite{Marantos16}.
Let us express the image projection of the camera with a perspective projection model \cite{Ma2003}.
Then, the coordinate frame $\Sigma_i$ is located above the image plane with a distance the same as the focal length $\lambda_i \in \R_{>0}$.
%Then, the coordinate frame $\Sigma_i$ is located at the distance equal to the focal length $\lambda_i \in \R_{>0}$ above the image plane.
We assume that all quadcopters have a circular image plane with radius $r\in \R_{>0}$.
Then, Quadcopter~$i$ can monitor all points inside of its field of view (FOV) $\mc F_i$ described as
\begin{equation*} %\label{eq:}
\mc F_i = \left\{ q\in \mc E \;\middle|\; \left\| q - [x_i~y_i~0]^\top \right\| \leq r\frac{z_i}{\lam_i} \right\},
%\mc F_i = \left\{ q\in \mc E ~\middle|~ \| q - [x_i, y_i, 0]^T \| \leq \blue{R_i} \right\},
\end{equation*}
where its side view is illustrated in Fig.~\ref{fig:FOV_graph}(a) with $I = [x_i~y_i~0]^\top$ denoting Quadcopter~$i$'s horizontal position projected on the environment $\mc E$.
%denoting the horizontal position of Quadcopter $i$ projected on the environment $\mc E$.
%Hereafter, we utilize the symbol $\bd(\cdot)$ to mean boundary of a set.

% Dynamics or Robot and its sensing region
The state of Quadcopter $i$ is denoted as ${\bf p}_i = [p_i^\top~\lambda_i]^\top$, where the focal length $\lambda_i$ specifies the zoom level of the camera.
%We assume that the dynamics of the quadcopter can be modeled as the single integrator dynamics
We assume that the dynamics of each quadcopter can be modeled as
\begin{equation}\label{eq:si_dyn}
\dot \bfp_i(t) = u_i(t),~ \bfp_i(t_0) = \bfp_{i0},~ t_0\in\R, ~ t\in\R_{>t_0},
\end{equation}
which can be converted to reflect the dynamics of a quadcopter \cite{Mellinger2011}.
Note that higher-order dynamics could represent more aggressive motions, as in \cite{Wu16}. However, good quality images are typically paramount in environmental monitoring, which means quadcopters rarely fly with large accelerations, e.g., to prevent image blurs. Hence, we opt for first-order dynamics.
%\blue{Note that higher-order dynamics could represent more aggressive motions as presented in \cite{Wu16}. However, because quadcopters in environmental monitoring rarely require flying with a large acceleration to obtain good-quality images, e.g., by preventing an image blur, we opt for first-order dynamics.}
Furthermore, this paper does not consider the disturbance, such as wind, assuming a controller implemented in each quadcopter can successfully cancel the effect as in \cite{Xiao17}.

\begin{figure}[t!]%\label{fig:exp_snap}
    \centering
    \subfloat[]
    {\makebox[0.48\hsize][c]{\includegraphics[width=0.46\linewidth]{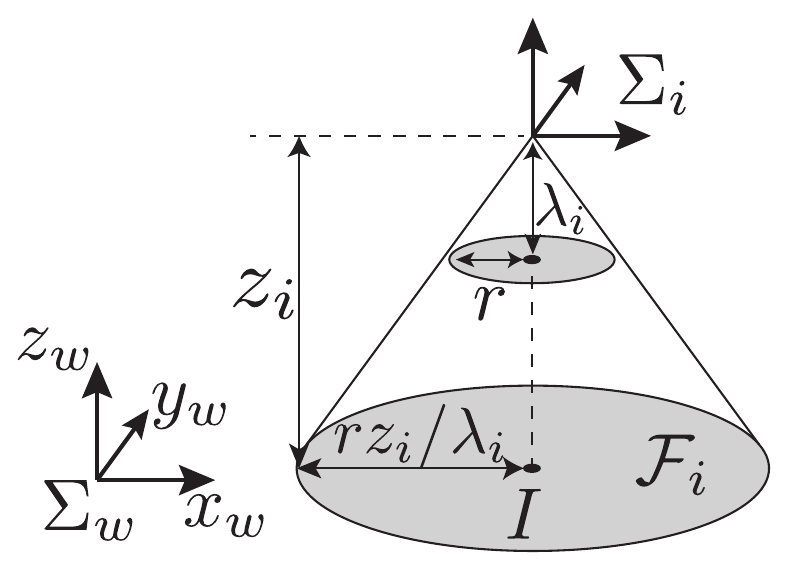}
    \label{fig:test1}}} \quad 
    \subfloat[]
    {\makebox[0.48\hsize][c]{\includegraphics[width=0.46\linewidth]{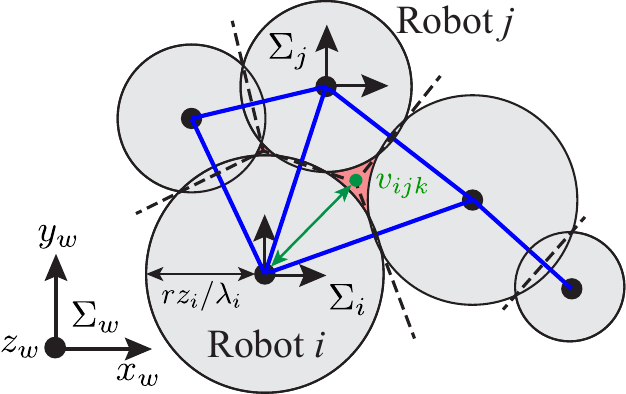}
    \label{fig:test2}}}
    \caption{An illustration of (a) Quadcopter $i$'s FOV $\mc F_i$ and (b) communication graph $\mc G$. In (b), a radical center $v_{ijk}$ is depicted as a green dot.}
    \label{fig:FOV_graph}
\end{figure}

\subsection{Communication Graph and Power Diagram}
% Network, Power diagram 
Every quadcopter is capable of sending its state to the neighboring quadcopters, where the inter-robot communication graph $\mc G(t)$ is modeled by the two undirected graphs $\mc G_{\mc F}(t)$ and $\mc G_{\mc P}(t)$.
First, $\mc G_{\mc F}(t)$ is described as the graph which has an edge between the agents $i$ and $j$ if their FOVs are overlapped, namely $\mc F_i\cap\mc F_j\neq\emptyset$.
%$\mc E_{\mc F} = \{(i,j) \in \mc V \times \mc V \mid \mc F_i\cap\mc F_j\neq\emptyset\}$. (\blue{Cite this? \cite{Montijano2016}})
Second, $\mc G_{\mc P}(t)$ is the graph generated by the power diagram, which is a form of the weighted Voronoi diagram, specified by the weighted distance %characterized by the weighted distance
\begin{align}\label{eq:power_dist}
 d_{\mathcal{P}}(\bfp_i, q) = \left\|[q_x~q_y]^\top-[x_i~y_i]^\top \right\|^2- \left(r\frac{z_i}{\lam_i}\right)^2,
\end{align}
%which we call, following \cite{Aurenhammer1987}, the power distance hereafter.
where $rz_i/ \lam_i$ in the second term signifies the radius of $\mc F_i$ as shown in Fig.~\ref{fig:FOV_graph}\subref{fig:test1}.
Hereafter, we call the distance \eqref{eq:power_dist} the power distance following \cite{Aurenhammer1987}. 
Then, the communication graph is expressed as $\mc G(t) = \mc G_{\mc F}(t) \cap \mc G_{\mc P}(t)$, which preserves the edges of power diagram only between those neighbors whose FOVs overlap, depicted by blue lines in Fig.~\ref{fig:FOV_graph}\subref{fig:test2}. % This operator is called "intersection"
%We denote Quadcopter $i$'s neighbors in graph $\mc G(t)$ as $\mc N_i(t)$.
The neighbors of Quadcopter~$i$ in graph $\mc G(t)$ is denoted as $\mc N_i(t)$.

An example of a power diagram is shown in Fig.~\ref{fig:power_diagram} with the graph $\mc G(t)$. 
The bisector of two agents $i$ and $j$ is composed of a straight line known as the radical axis, along which both circles $\bd(\mc F_i)$ and $\bd(\mc F_j)$, with the symbol $\bd(\cdot)$ to mean boundary of a set, have the same power distance.
Namely, if we denote the radical axis between agent $i$ and $j$ as $L_{ij}$, $d_{\mc P}(\bfp_i, q) = d_{\mc P}(\bfp_j, q),~q \in L_{ij}$ holds.
%In addition, , when the two circles $\mc F_i$ and $\mc F_j$ have intersection points, $L_{ij}$ passes through these intersection points.
In addition, if $\bd(\mc F_i) \cap \bd(\mc F_j) \neq \emptyset$ with $\bfp_i\neq \bfp_j$, then $\bd(\mc F_i) \cap \bd(\mc F_j) \in L_{ij}$ is satisfied.
%Every vertix 
%In addition, when the two circles $\mc F_i$ and $\mc F_j$ have intersection points, the radical axis between $i$ and $j$ passes through two intersection points of their FOVs.
It is known that the radical axes defined by three circles %$\bd(\mc F_i), \bd(\mc F_j), \bd(\mc F_k)$ 
whose centers are not collinear intersect in a common point termed the radical center, %$v_{ijk} = [v_{ijk}^x~v_{ijk}^y]^\top$
depicted with red dots in Fig.~\ref{fig:power_diagram}.
%Let us denote the radical center generated by $\bd(\mc F_i)$, $\bd(\mc F_j), \bd(\mc F_k)$ 
The radical center is generated by three agents and forms a vertex of the weighted Voronoi cell with the weighted distance \eqref{eq:power_dist}.
%The radical center is generated by three agents and 
Let us denote the radical center generated by $\bd(\mc F_i)$, $\bd(\mc F_j)$, and $\bd(\mc F_k)$ as $v_{ijk} = [v_{ijk}^x~v_{ijk}^y~0]^\top\in \R^3$.  %$v_{ijk}$.
Then, from the property of the radical axis, the following condition is satisfied.
\begin{align}\label{eq:power_dist_eq}
 d_{\mathcal{P}}(\bfp_i, v_{ijk}) = d_{\mathcal{P}}(\bfp_j, v_{ijk}) = d_{\mathcal{P}}(\bfp_k, v_{ijk})
\end{align}
%Note that in the special case of $rz_i/\lambda_i = rz_j/\lambda_j$, namely the radius of $\bd(\mc F_i)$

\begin{figure}
    \centering
    \includegraphics[width=0.45\linewidth]{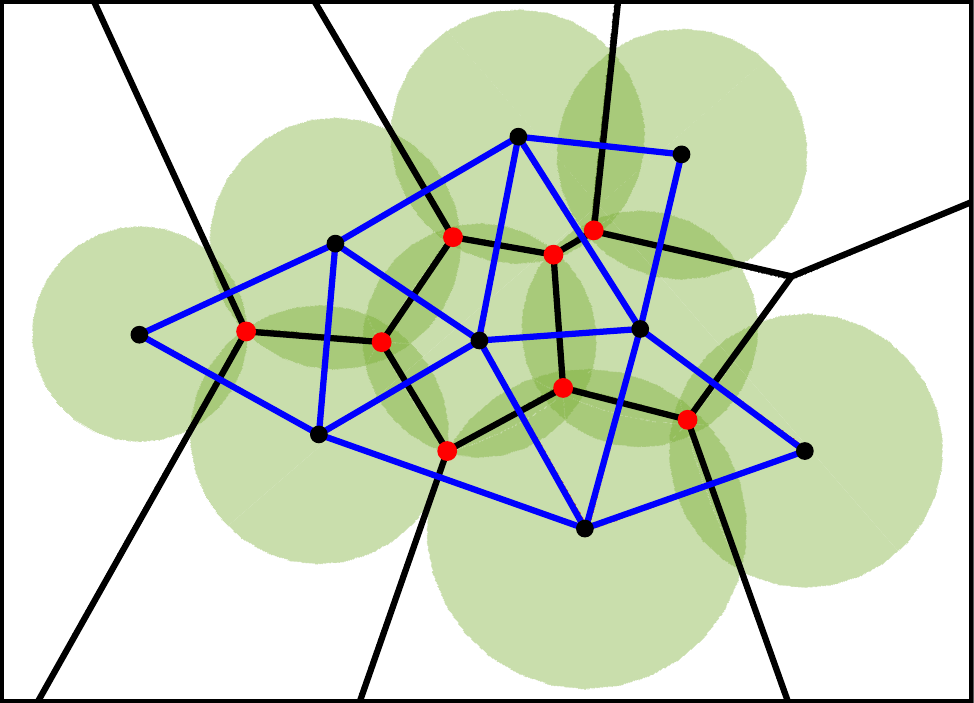}
    \caption{Power diagram generated by nine quadcopters, illustrated by black lines. The vertices of the power diagram shown in red dots, called radical centers, are generated by groups of three quadcopters depicted in black dots. The FOVs of quadcopters and the graph $\mc G(t)$ are depicted with the green circles and the blue lines, respectively. On the right, the boundary of $\cup_{i=1}^n \mc F_i$ has a dent with a radical center, but this is not regarded as a hole, which is formally defined in Definition~\ref{def:hole}.
    }
    \label{fig:power_diagram}
\end{figure}

% Triangulation
Since the graph $\mc G_{\mc P}(t)$ is generated in a similar fashion to how the ordinary Voronoi diagram yields the Delaunay triangulation, the graph $\mc G(t)$ considered throughout this paper also provides triangular subgraphs as observed in Fig.~\ref{fig:power_diagram}.
Let us denote the set of triangular subgraphs containing Quadcopter~$i$ as its vertex as $\mc T_i(t) = \{\Delta_{1},\cdots,\Delta_{m_i} \}$, where $\Delta_{\iota}, \iota\in \mc I_i(t)=\{ 1,\cdots, m_i(t)\}$ stands for $m_i(t)$ triangular subgraphs\footnote{Although we should define a triangular subgraph of Quadcopter~$i$ with $\Delta_{i,\iota}$, we use $\Delta_{\iota}$ for notational simplicity.}. % composing $\mc T_i(t)$.
Note that $\mc T_i(t)$ can be regarded as a subset of the Delaunay triangulation generated by the power diagram.
%Without loss of generality, we make the following assumption for a trio determined by a triangular subgraph $\Delta_{\iota} \in \mc T_i(t)$.
Without loss of generality, we introduce the following assumption for a trio constituting a triangular subgraph $\Delta_{\iota} \in \mc T_i(t)$.
\begin{assumption} \label{asm:non_parallel}
%Consider a triangular subgraph $\{i,j,k\} = \Delta_{\iota} \in \mc T_i(t)$, $t\geq t_0, i\in \mc N$.
No pair in $\{\bd(\mc F_i), \bd(\mc F_j), \bd(\mc F_k)\}$, $\Delta_{\iota} = \{i,j,k\} \in \mc T_i(t)$ %$\{i,j,k\} = \Delta_{\iota} \in \mc T_i(t)$ 
%, $t\geq t_0, \iota\in \mc I_i, i\in \mc N$ 
contains sensing regions that are concentric. In addition, no two axes among $\{L_{ij}, L_{jk}, L_{ki}\}$ are parallel.
\end{assumption}
Under Assumption~\ref{asm:non_parallel}, the existence of the radical axes $\{L_{ij}, L_{jk}, L_{ki}\}$ is assured, and all three axes intersect at a single point. Moreover, $\triangle{IJK}$ does not become a degenerate triangle, where $\triangle{IJK}$ denotes the triangle with vertices $I$, $J$, and $K$.
Note that $\triangle{IJK}$ becomes a degenerate triangle only when the pathological situation, and even if it happens, the proposed method can easily manage such an ill condition with a heuristic method, as shown in the simulation studies in Section~\ref{ssec:sim_comp_trio}.

Notice that the graph $\mc G(t)$ dynamically changes according to the quadcopters' motion. 
We suppose that the set of triangular subgraphs $\mc T_{i}(t)$ does not change over intervals specified by a switching sequence $\{\tau_i^{\mrm k}\}_{\mrm k=1}^\infty$ as introduced in Definition~\ref{def:sw_sq}.

The proposed scenario requires each quadcopter to exchange its state with the neighboring quadcopters. While one could consider employing a state estimation technique to obtain the neighbors' state \cite{Vetrella19}, the zoom level of the mounted camera, included in the state vector, is deemed challenging to estimate externally. Therefore, this paper presumes that each quadcopter is equipped with a communication device to send and receive its state. The communication devices and the challenges involved with outdoor communication, such as guaranteeing reliability by dedicating redundant networks and fail-safety, are detailed in \cite{Hayat16}.

\subsection{Unmonitored Areas to Be Prevented} \label{ssec:hole_def}

Noticing that the power diagram generates convex cells, this paper prevents the formation of unmonitored areas in-between FOVs by including the Voronoi vertices $v_{ijk}$, which exist in-between FOVs of teams, in each FOV, as shown in Fig.~\ref{fig:power_diagram}.
Because each Voronoi vertex is generated by three agents except for degenerate Voronoi diagrams \cite{Okabe2000}, we can separate the problem into triangular subgraphs $\Delta_{\iota} \in \mc T_i(t)$. 
Note that this degeneracy does not cause serious problems because degenerate Voronoi diagrams can be divided into triangular subgraphs.

For the above reasons, we restrict our attention to the unmonitored regions that may appear among a trio determined by triangular subgraphs $\mc T_i(t), \forall i \in \mc N$, illustrated in Fig.~\ref{fig:trio_hole}(a).
This unsurveilled area is denoted as a {\it hole} hereafter and formally defined as follows.
\begin{definition} \label{def:hole}
%Let all pairs in $\Delta_{\iota} = \{i, j, k\}$ have an edge determined by the graph $\mc G$.
Consider a triangular subgraph $\{i,j,k\} = \Delta_{\iota} \in \mc T_i(t)$, $\iota\in \mc I_i(t), t\geq t_0, i\in \mc N$.
A closed set $E \subset \mc Q$ is said to be a {\it hole}, if and only if it satisfies 
the conditions
\begin{align}
\begin{split}
(\bd(E) \cap \bd(\mc Q) = \emptyset) &~\land~
({\rm Int} (E) \cap \mc F_i = \emptyset, ~\forall i \in \mc N)\\ 
~\land~ (E &\subset \triangle{IJK}),
\end{split}
\end{align}
where ${\rm Int} (E)$ denotes the interior of $E$. 
\end{definition}

Since a hole in Definition~\ref{def:hole} depends on a triangular subgraph $\Delta_{\iota}\in\mc T_i(t)$ that dynamically changes according to the switching sequence $\{\tau_i^{\mrm k}\}_{\mrm k=1}^\infty$, the condition for preventing a hole also switches.
%Let us suppose that Quadcopter $i$'s condition for eliminating a hole in $m_i \in T_i$ during a time sequence $\tau_i^{\mrm k}$ is described as $h_{i,m_i}^{\mrm k}$, which we will derive in Section~\ref{sec:hole_prev}.
Let us suppose that the zero superlevel set of a function $h_{i,\Delta_{\iota}}^{\mrm k}(\bfp_{l \in \Delta_{\iota}})$ with $\bfp_{l \in \Delta_{\iota}}:= [\bfp_i^\top \bfp_j^\top \bfp_k^\top]^\top$, $\Delta_{\iota} = \{i, j, k\}$, which we will derive in Section~\ref{sec:hole_prev}, describes the condition for Quadcopter~$i$ to eliminate a hole in $\Delta_{\iota} \in \mc T_i(\tau_i^{\mrm k})$.
By incorporating $h_{i,\Delta_{\iota}}^{\mrm k}(\bfp_{l \in \Delta_{\iota}}), \forall \Delta_{\iota} \in \mc T_i(\tau_i^{\mrm k})$ together, we obtain
\begin{align} \label{eq:h_i_land}
    h_i^{\mrm k}(\bfp_i, \bfp_{l \in {\mc N}_i}) = \bigwedge_{\iota = 1}^{m_i(\tau_i^{\mrm k})} h_{i, \Delta_{\iota}}^{\mrm k},
\end{align}
the zero-superlevel set of which signifies the constraint Quadcopter~$i$ must satisfy to prevent the formation in-between all triangular subgraphs $\mc T_i(\tau_i^{\mrm k})$.
Then, the set
\begin{align} \label{eq:set_CHNBF}
\mc C_i = \{ (\bfp_i,\! t) \in \R^4 \times \R_{\geq t_0} \mid h_i(\bfp_i, \bfp_{l \in {\mc N}_i}, t) \geq 0\},
%\mc C_i \!\!=\!\! \{ (\bfp_i,\! t)\! \in \R^4 \!\times\! \R_{\geq t_0} \!\mid h_{i,m_i}(\bfp_i, \bfp_{l \in {\mc N}_i}, t)\! \geq\! 0\, \forall m_i \!\in \mc T_i(t)\},
\end{align}
with HNCBF
\begin{align} \label{eq:can_CHNBF}
\begin{split}
%&h_{i, m_i}(\bfp_i, \bfp_{l \in {\mc N}_i}, t) = h_{i, m_i}^{\mrm k}(\bfp_i, \bfp_{l \in {\mc N}_i}), \\
&h_i(\bfp_i, \bfp_{l \in {\mc N}_i}, t) = h_i^{\mrm k}(\bfp_i, \bfp_{l \in {\mc N}_i}), \\
&\hspace{1cm} \forall {\mrm k} \in \bN, \forall t \in [\tau_i^{\mrm k}, \tau_i^{\mrm k+1}), \forall \bfp_{i} \in \R^{4},
\end{split}
\end{align}
signifies the prevention of a hole is satisfied for Quadcopter~$i$ during the entire mission.

The goal of this paper is to present a visual coverage framework that prevents the appearance of a hole by guaranteeing the hybrid forward invariance of the set $\mc C_i, \forall i \in \mc N$. 
More specifically, during a monitoring task, we ensure no hole appears among FOVs of quadcopters if the changes of the graph do not make holes instantaneously in-between $\mc T_i(t), \forall i \in \mc N$. 
Note that, in the proposed scenario, there is a possibility that changes of a graph $\mc G(t)$ yield a sudden appearance of a hole. 
For instance, as illustrated in the right part of Fig.~\ref{fig:power_diagram}, $\bd(\bigcup_{i=1}^n \mc F_i)$ can have a dent.
If the right-upper and right-lower quadcopters' FOVs make overlap, a new subtrianglar graph is created. 
Although this graph change could alter the dent into a hole, the robustness of the forward invariance of the set $\mc C_i$ can eliminate this suddenly appeared hole in most cases as will be shown in the simulation.

\section{Prevention of Coverage Holes} \label{sec:hole_prev}

In this section, we first derive the NCBF $h_{i,\Delta_{\iota}}^{\mrm k}(\bfp_{l\in \Delta_{\iota}})$ whose zero superlevel set encodes the prevention of a hole in $\Delta_{\iota}\in \mc T_i(\tau_i^{\mrm k})$ during a time interval $[\tau_i^{\mrm k}, \tau_i^{\mrm k+1})$. 
Then, to prevent a formation of a hole, we develop a distributed algorithm that provides the hybrid forward invariance property of the set $\mc C_i$ expressed as a superlevel set of HNCBF $h_{i}$.
Finally, we show that $h_{i,\Delta_{\iota}}^{\mrm k}$ composing $h_i$ 
is a valid NCBF except for pathological points, hence the proposed algorithm can generically find the control input preventing a hole. 
In each procedure, we utilize the favorable properties of the power diagram incorporated in the designed CBFs.

\subsection{Control Barrier Functions Preventing a Hole} \label{ssec:hole_CBF}

In this subsection, we formalize the condition to be satisfied for preventing a hole during a time interval $[\tau_i^{\mrm k}, \tau_i^{\mrm k+1})$. 
As discussed in Section~\ref{ssec:hole_def}, an appearance of a hole in-between $\Delta_{\iota} = \{i,j,k\} \in \mc T_i(\tau_i^{\mrm k})$ can be prevented by satisfying $v_{ijk} \in \mc F_i$. 
However, as shown in Fig.~\ref{fig:trio_hole}(c), this condition is too conservative because certain movements of quadcopters are restricted, even if such movements do not form a hole.
To alleviate this condition, let us consider the differences between Figs.~\ref{fig:trio_hole}(a) and (c), where a hole only appears in Fig.~\ref{fig:trio_hole}(a) although $v_{ijk}$ is outside of $\mc F_i$ in either case. Then, we realize that, in Fig.~\ref{fig:trio_hole}(c), $v_{ijk}$ is outside of $\triangle{IJK}$ while the deployment in Fig.~\ref{fig:trio_hole}(a) is not. 
This observation leads us to introduce the following condition to prevent a formation of a hole
\begin{align}
\lnot \left( v_{ijk} \in \triangle{IJK} \right) 
\lor \left( v_{ijk} \in \mc F_i \right), \label{eq:iff_condition}
\end{align}
which requires Quadcopter $i$ to satisfy $v_{ijk} \in \mc F_i$ only when $v_{ijk} \in \triangle{IJK}$ holds.
The following theorem in our previous work \cite{Funada20} formalizes the relationship between \eqref{eq:iff_condition} and the existence of a hole.
\begin{theorem} %\cite[Theorem~1]{Funada20}
Let us consider Quadcopters $i, j$, and $k$, where each quadcopter's FOV intersects with that of the other two quadcopters. Then, there is no hole in-between the three quadcopters if and only if the condition \eqref{eq:iff_condition} is satisfied.
\end{theorem}
\begin{proof}
See Theorem~1 in our conference work \cite{Funada20}.
\end{proof}

\begin{figure}[t!]%\label{fig:exp_snap}
\vspace{-2mm}
    \centering
    \subfloat[]
    {\makebox[0.3\hsize][c]{\includegraphics[width=0.29\linewidth]{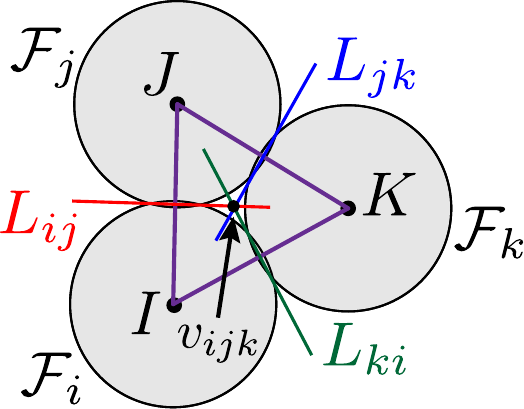}
    \label{fig:trio_1}}} \quad 
    \subfloat[]
    {\makebox[0.3\hsize][c]{\includegraphics[width=0.29\linewidth]{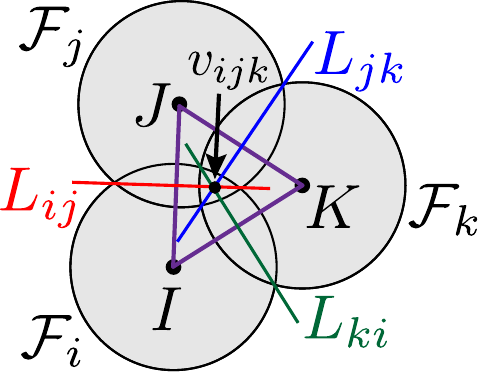}
    \label{fig:trio_2}}} \quad
    \subfloat[]
    {\makebox[0.3\hsize][c]{\includegraphics[width=0.29\linewidth]{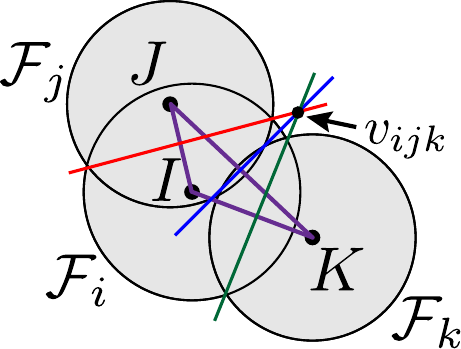}
    \label{fig:trio_3}}}
    \caption{Three cases of FOV configurations with $\triangle{IJK}$ depicted as a purple triangle. (a) illustrates a hole in the sense of Definition~\ref{def:hole}, where $v_{ijk}$ are not contained by any of a trio. Neither (b) nor (c) has a hole, but $v_{ijk}$ lies outside of FOVs in (c). This signifies that if we utilize a constraint $v_{ijk} \in \mc F_i$ only, the configuration (c) is also restricted as well as (a), even if (c) does not have any holes.}
    \label{fig:trio_hole}
\end{figure}

The condition $v_{ijk} \in \triangle{IJK}$ appeared in \eqref{eq:iff_condition} can be described as 
\begin{align}
    (h_{i, \Delta_{\iota}, IJK}^{\mrm k} \!>\! 0) \!\land\! (h_{i, \Delta_{\iota}, JKI}^{\mrm k} \!>\! 0) \!\land\! (h_{i, \Delta_{\iota}, KIJ}^{\mrm k} \!>\! 0), \label{eq:h_intri}
\end{align}
where 
\begin{align} 
h_{i, \Delta_{\iota},IJK}^{\mrm k} &= \frac{ \be_z^\top \left(\overrightarrow{IJ} \times \overrightarrow{Iv_{ijk}}\right)}  
{\be_z^\top \left( \overrightarrow{IJ} \times \overrightarrow{IK} \right)} \label{eq:h_IJK} %\\
%&=\frac{A_{IJv_{ijk}}}{A_{IJK}},  
\end{align}
is a candidate CBF satisfying $h_{i, \Delta_{\iota},IJK}^{\mrm k}>0$ if and only if $v_{ijk}$ and $K$ exist in the same half-plane separated by the line $IJ$.
% Note that denoting the area of the triangle $IJK$ as $A_{IJK}$, then $\| h_{IJK} \| = A_{IJv_{ijk}}/ A_{IJK}$ holds.
Note that the denominator of \eqref{eq:h_IJK} does not become $0$ under Assumption~\ref{asm:non_parallel}.
By definition in \eqref{eq:h_IJK}, 
\begin{align} \label{eq:tri_NCBF_same}
%\begin{split}
h_{i, \Delta_{\iota},IJK}^{\mrm k} \!=\! h_{j, \Delta_{\iota},IJK}^{\mrm k} \!=\! h_{k, \Delta_{\iota},IJK}^{\mrm k}
%\hspace{5mm}\Delta_{i,\iota} = \Delta_{j,\iota} = \Delta_{k,\iota} = \{i,j,k \} \\
%\end{split}
\end{align}
holds for $\Delta_{\iota} = \{i,j,k \} \in \mc T_i(\tau_i^{\mrm k})$, where the same relationship is satisfied for $h_{i, \Delta_{\iota}, JKI}^{\mrm k}$ and $h_{i, \Delta_{\iota}, KIJ}^{\mrm k}$.
%Sinec the 
%with $A_{IJK}$ signifying the area of $\triangle{IJK}$.
The other condition $v_{ijk} \in \mc F_i$ in \eqref{eq:iff_condition} is expressed with another CBF candidate as
\begin{align} 
     h_{i, \Delta_{\iota}, \mc F}^{\mrm k} = \left( r\frac{z_i}{\lam_i} \right)^2 - \left\|[v_{ijk}^x~v_{ijk}^y]^\top-[x_i~y_i]^\top \right\|^2. \label{eq:h_F}
\end{align}

The condition \eqref{eq:iff_condition} preventing a hole in a triangular subgraph $\Delta_{\iota}=\{i,j,k\} \in \mc T_i(\tau_i^{\mrm k})$ can be written as 
\begin{align}
    &\lnot \! \left(  (h_{i, \Delta_{\iota}, IJK}^{\mrm k} \!>\! 0) \land (h_{i, \Delta_{\iota}, JKI}^{\mrm k} \!>\! 0) \land (h_{i, \Delta_{\iota}, KIJ}^{\mrm k} \!>\! 0)  \right) \nonumber \\ 
    &\hspace{5mm}\lor (h_{i, \Delta_{\iota} \mc F}^{\mrm k} \!>\! 0) \label{eq:con_nohole}.
\end{align}
From \eqref{eq:Bool_logic} and \eqref{eq:con_nohole}, the NCBFs of Quadcopter $i$ to be considered for preventing the formation of holes in-between all triangular subgraphs $\Delta_{\iota} \in \mc T_i(t)$ on the $\mrm k$-th interval are described as
\begin{align} \label{eq:h_mi}
    h_{i, \Delta_{\iota}}^{\mrm k}(\bfp_{l \in \Delta_{\iota}}) \!=\! \max_{\ell \in \{1, \cdots, 4 \} }\{h_{i, \Delta_{\iota}, \ell}^{\mrm k}\},
    \forall \Delta_{\iota} \in \mc T_i(\tau_i^{\mrm k}).
    %&\hspace{0.5cm} \ell=1,\cdots, 4, ~m_i \in \mc T_i(\tau_i^{\mrm k}),
\end{align}
%by following \eqref{eq:h_intri}, \eqref{eq:h_F}, and \eqref{eq:NCBF_nohole}.
Here, we define $h_{i, \Delta_{\iota}, \ell}^{\mrm k}$, $\ell \in \{1, \cdots, 4\}$ for each triangular subgraph $\Delta_{\iota} = \{i,j,k \} \in \mc T_i(\tau_i^{\mrm k})$, where they correspond with $-h_{i, \Delta_{\iota},IJK}^{\mrm k}$, $-h_{i, \Delta_{\iota},JKI}^{\mrm k}$, $-h_{i, \Delta_{\iota},KIJ}^{\mrm k}$, and $h_{i, \Delta_{\iota},\mc F}^{\mrm k}$ in order for notational simplicity.

\subsection{Algorithm for Preventing a Hole} \label{ssec:no_hole_alg}

This subsection presents the distributed algorithm that prevents the appearance of a hole in-between a team even under the graph $\mc G(t)$ switches. %by extending the previous results into CHNBF. 
Namely, the set $\mc C_i$ in \eqref{eq:set_CHNBF} defined with the HNCBF \eqref{eq:can_CHNBF} and the NCBF \eqref{eq:h_mi} is rendered hybrid forward invariance by the algorithm. %by control input generated by the algorithm.
To calculate the control input in a distributed manner, we will utilize the symmetric property of the designed CBFs.
%The proposed algorithm utilizes the symmetric property of the designed CBFs to allow each agent to generate its control input in a distributed manner.

% Hybrid forward invariance
The hybrid forward invariance of the set $\mc C_i$ can be guaranteed by satisfying the following constraint during each switching interval $t' \in [\tau_i^{\mrm k}, \tau_i^{\mrm k+1})$
\begin{align} \label{eq:const_HCBF}
    \dot h_{i, \Delta_{\iota}}^{\mrm k}(\bfp_{l\in \Delta_{\iota}}) \!\geq\! - \alpha( h_{i, \Delta_{\iota}}^{\mrm k}(\bfp_{l\in \Delta_{\iota}}) ),
    \forall \Delta_{\iota} \!\in\! \mc T_i(\tau_i^{\mrm k})
\end{align}
for every Carath\'{e}odry solution, 
for some locally Lipschtz extended class-$\mc K$ function $\alpha : \R \to \R$ \cite{Glotfelter19}. 
However, the nonsmoothness exists in \eqref{eq:h_mi} hinders us from obtaining the time derivative $\dot h_{i, \Delta_{\iota}}^{\mrm k}$ via the usual chain rule.

As a remedy for this problem, we reformulate the constraint \eqref{eq:const_HCBF} by incorporating the gradients of all the functions comprising the almost-active set of functions as discussed in Remark~\ref{rem:QP_feas}. 
Then, the appearance of holes can be prevented by satisfying the following constraint on each interval
\begin{align}
 &\frac{\partial h_{i, \Delta_{\iota}, \ell}^{\mrm k}}{\partial \bfp_i}^\top \!u_i \!+\! \frac{\partial h_{i, \Delta_{\iota}, \ell}^{\mrm k}}{\partial \bfp_j}^\top \!u_j \!+\! \frac{\partial h_{i, \Delta_{\iota}, \ell}^{\mrm k}}{\partial \bfp_k}^\top \!u_k \!\geq\! - \alpha(h_{i,\Delta_{\iota}}^{\mrm k}), \nonumber\\
 &\hspace{0.3cm}\forall \ell \in I_{i, \Delta_{\iota}, \epsilon}^{\mrm k}(\bfp_{l\in \Delta_{\iota}}),~\forall \Delta_{\iota} \in \mc T_i(\tau_i^{\mrm k}). \label{eq:CBF_cent}
\end{align}
with the almost-active set of functions
\begin{align}
\begin{split}
    &I_{i, \Delta_{\iota}, \epsilon}^{\mrm k}(\bfp_{l\in \Delta_{\iota}}) = \\
    &\hspace{5mm} \left\{ \ell \;\middle|\;  \left\| h_{i, \Delta_{\iota}, \ell}^{\mrm k}(\bfp_{l\in \Delta_{\iota}}) - h_{i,\Delta_{\iota}}^{\mrm k}(\bfp_{l\in \Delta_{\iota}}) \right\| < \epsilon \right\}.
\end{split}
\end{align}
%\begin{align} \label{eq:CBF_cent}
%    \frac{\partial h_{ijk}}{\partial \bfp_i}^\top u_i + \frac{\partial h_{ijk}}{\partial \bfp_j}^\top u_j + \frac{\partial h_{ijk}}{\partial \bfp_k}^\top u_k \geq - \alpha().
%\end{align}
However, the input $u_i$ satisfying \eqref{eq:CBF_cent} cannot be calculated in a distributed manner because it requires input information of other agents $l \in \Delta_{\iota} \backslash \{i\}$. 
Furthermore, in order to evaluate the second term of the left side of \eqref{eq:CBF_cent}, it requires the information of both $h_{i, \Delta_{\iota}, \ell}^{\mrm k}$ and $u_j$; namely, neither Quadcopter $i$ nor $j$ can calculate it alone. 
The same issue occurs in the third term. The following theorem proves that condition (28) can be divided into three conditions, where each of them can be evaluated in a distributed manner by Quadcopters $i$, $j$, and $k$, respectively.
%In the following theorem, we propose the constraint which is implementable in a distributed fashion.
%The following theorem proposes the constraint which is implementable in a distributed fashion.

%
\begin{theorem} \label{th:CBF_dist_hybrid}
Suppose that holes in the sense of Definition~\ref{def:hole} do not appear in-between agents at the initial time $\tau_i^{\mrm k}$ of an interval $[\tau_i^{\mrm k}, \tau_i^{\mrm k+1})$, $\mrm k \in \{1,\cdots \infty\}$. Then, any measurable and locally bounded controllers $u_i$, $i\in \mc N$ satisfying 
\begin{align} \label{eq:CBF_dist}
\begin{split}
    &\frac{\partial h_{i, \Delta_{\iota}, \ell}^{\mrm k}(\bfp_{l\in \Delta_{\iota}})}{\partial \bfp_i}^\top u_i \geq 
    - \frac{1}{3} \alpha\left( h_{i, \Delta_{\iota}}^{\mrm k}(\bfp_{l\in \Delta_{\iota}})\right), \\
    &\hspace{0.3cm}\forall \ell \in I_{i, \Delta_{\iota}, \epsilon}^{\mrm k}(\bfp_{l\in \Delta_{\iota}}),~\forall \Delta_{\iota} \in \mc T_i(\tau_i^{\mrm k})
\end{split}
\end{align}
will render the safe set $\mc C_i$ hybrid forward invariant. 
\end{theorem}
\begin{proof}
%First, we prove that if we satisfy the condition \eqref{eq:CBF_cent}
Let us consider the three agents $\{i, j, k\}=\Delta_{\iota}$. % belonging to $\Delta_{i,\iota} = \Delta_{j,\iota} = \Delta_{k,\iota} = \{i,j,k \}$.
We show the condition \eqref{eq:CBF_cent} can be distributed in the form of the equation \eqref{eq:CBF_dist} from the symmetric property of the designed NCBF $h_{i,\Delta_{\iota}}^{\mrm k}$. First, from the equation \eqref{eq:power_dist} and \eqref{eq:h_F}, $h_{i,\Delta_{\iota},\mc F}^{\mrm k}$ can be expressed by the power distance as 
\begin{align} \label{eq:h_F2_dp}
    h_{i,\Delta_{\iota},\mc F}^{\mrm k} = -d_{\mc P}(\bfp_i, v_{ijk}).
\end{align}
From the property of the power distance \eqref{eq:power_dist_eq}, 
\begin{align}
%\begin{split}
    h_{i, \Delta_{\iota}, \mc F}^{\mrm k} = h_{j, \Delta_{\iota},\mc F}^{\mrm k} = h_{k, \Delta_{\iota},\mc F}^{\mrm k} 
    %&\hspace{5mm} \Delta_{\iota} = \Delta_{j,\iota} = \Delta_{k,\iota} = \{i,j,k \}
%\end{split}
    %&\Delta_{\iota} = \Delta_{j,\iota} = \Delta_{k,\iota} = \{i,j,k \} \\
    %&\hspace{5mm} \Rightarrow h_{i, \Delta_{\iota}, \mc F} \!=\! h_{j, \Delta_{j,\iota},\mc F} \!=\! h_{k, \Delta_{k,\iota},\mc F}%, \, \forall \Delta_{\iota} \!=\! \{i,j,k\} \!\in\! \mc T_i
\end{align}
holds.
Second, from \eqref{eq:tri_NCBF_same}, the NCBFs evaluating the condition $v_{ijk}\notin \triangle{IJK}$ takes the same value among $\{i,j,k \}$.
%Second, from \eqref{eq:tri_NCBF_same}, all the agent $\{i,j,k\} \in \mc T_i$ utilize the common NCBF $\max\{-h_{IJK}, -h_{JKI}, -h_{KIJ}\}$ to evaluate the condition $v_{ijk}\notin \triangle{IJK}$ from the definition \eqref{eq:h_intri} and \eqref{eq:h_IJK}.
Hence, the synthesized NCBF \eqref{eq:h_mi} has the following symmetric property
\begin{align} \label{eq:h_eq_trio}
%\begin{split}
    h_{i, \Delta_{\iota}}^{\mrm k} = h_{j, \Delta_{\iota}}^{\mrm k} = h_{k, \Delta_{\iota}}^{\mrm k}%,~%~ \forall \{i,j,k\} \in\mc T_i, \\
    %\Delta_{\iota} \!=\! \Delta_{j,\iota} \!=\! \Delta_{k,\iota} \!=\! \{i,j,k \}
%\end{split}
\end{align}
Considering the fact that the derivatives of $h_{i, \Delta_{\iota}}^{\mrm k}$, $h_{j, \Delta_{\iota}}^{\mrm k}$, and $h_{k, \Delta_{\iota}}^{\mrm k}$ take the same value among $\{i,j,k\} \in\mc T_i$ from \eqref{eq:h_eq_trio}, the condition \eqref{eq:CBF_cent} can be decomposed into %three agents composing $\Delta_{i,\iota} = \Delta_{j,\iota} = \Delta_{k,\iota} = \{i,j,k \}$
\begin{subequations} %\label{eq:CBF_dist}
\begin{align}
     &\frac{\partial h_{i, \Delta_{\iota},\ell}^{\mrm k}(\bfp_{l\in \Delta_{\iota}})}{\partial \bfp_i}^\top u_i \geq -\frac{1}{3} \alpha\left( h_{i, \Delta_{\iota}}^{\mrm k}(\bfp_{l\in \Delta_{\iota}})\right), \\
     &\frac{\partial h_{j, \Delta_{\iota},\ell}^{\mrm k}(\bfp_{l\in \Delta_{\iota}})}{\partial \bfp_j}^\top u_j \geq -\frac{1}{3} \alpha\left( h_{j, \Delta_{\iota}}^{\mrm k}(\bfp_{l\in \Delta_{\iota}})\right), \label{eq:dist_cond_j}\\
     &\frac{\partial h_{k, \Delta_{\iota},\ell}^{\mrm k}(\bfp_{l\in \Delta_{\iota}})}{\partial \bfp_k}^\top u_k \geq -\frac{1}{3} \alpha\left( h_{k, \Delta_{\iota}}^{\mrm k}(\bfp_{l\in \Delta_{\iota}})\right). \label{eq:dist_cond_k}
\end{align}
\end{subequations}
This means that if each quadcopter $i\in\Delta_\iota$ satisfies the condition \eqref{eq:CBF_dist} $\forall \Delta_{\iota} \in \mc T_i(\tau_i^{\mrm k})$, %, $\forall i\in \mc N$ 
then the condition \eqref{eq:CBF_cent} is also satisfied. 
Notice that the inequalities \eqref{eq:dist_cond_j} and \eqref{eq:dist_cond_k} no longer contain either $h_{i, \Delta_{\iota}}^{\mrm k}$ or the gradients of its components, hence Quadcopters $j$ and $k$ can evaluate \eqref{eq:dist_cond_j} and \eqref{eq:dist_cond_k} alone, respectively.
This completes the proof.
\end{proof}

%\bfp_{i\in m_i}

\begin{algorithm}[t] 
\caption{Distributed Coverage Hole Prevention}  \label{alg:cov_main}
\begin{algorithmic} 
\STATE{INPUT:} Nominal controller: $u_{i, {\rm nom}}$
\STATE{\hspace{0.5cm}} The state of the neighboring agents of $i$: $\mc \bfp_{l \in {\mc N}_i}$
\STATE{\hspace{0.5cm}} The set of triangular subgraphs containing $i$: $\mc T_i(\tau_i^{\mrm k})$
\STATE{OUTPUT:} Controller preventing holes: $u_i^*$
%\STATE $h_{i, m_i}^{\mrm k}(\bfp_i, \bfp_j, \bfp_k) \gets \max_\ell\{h_{m_i, \ell}^{\mrm k} \},$ 
%\STATE \hspace{3cm} $\ell=1, \cdots, 4, ~\forall m_i \in \mc T_i(\tau_i^{\mrm k})$ 
%\STATE $h_i^{\mrm k}(\bfp_i, \bfp_{j \in {\mc N}_i}) \gets \min_{m_i} \{ h_{i,m_i}^{\mrm k}(\bfp_i, \bfp_j, \bfp_k) \}$
%\STATE $I_{\epsilon, m_i}^{\mrm k} \gets \emptyset$, $I_{\epsilon, \ell}^{\mrm k} \gets \emptyset$
%\STATE $I_{\epsilon, \ell}^{\mrm k} \gets \emptyset$
\FOR{$\iota=1: |\mc T_i(\tau_i^{\mrm k})|$} 
\STATE $h_{i, \Delta_{\iota}}^{\mrm k}(\bfp_{l \in \Delta_{\iota}}) \gets \max_\ell\{h_{i, \Delta_{\iota}, \ell}^{\mrm k} \}$, $\ell \in \{1, \cdots, 4\}$ 
%\STATE \hspace{3cm} $\ell \in \{1, \cdots, 4\}$
\STATE $I_{i, \Delta_{\iota}, \epsilon}^{\mrm k} \gets \emptyset$
%\IF{$\| h_{i, m_i}^{\mrm k} - h_i^{\mrm k} \| \leq \epsilon$}
%\STATE $I_{\epsilon, m_i}^{\mrm k}(\bfp_i, \bfp_{j \in {\mc N}_i}) \cup \{m_i\}$
\FOR{$\ell=1:4$} 
\IF{$\| h_{i, \Delta_{\iota}, \ell}^{\mrm k} - h_{i, \Delta_{\iota}}^{\mrm k} \| \leq \epsilon$}
%\STATE $I_{\epsilon, \ell, m_i}^{\mrm k}(\bfp_i, \bfp_{j \in {\mc N}_i}) \cup \{\ell\}$
\STATE $I_{i, \Delta_{\iota}, \epsilon}^{\mrm k}(\bfp_{l\in \Delta_{\iota}}) \cup \{\ell\}$
\ENDIF
\ENDFOR
%\ENDIF
\ENDFOR
\STATE $u_i^* \gets \argmin_{u_i} (u_i - u_{i, {\rm nom}})^\top W (u_i - u_{i, {\rm nom}})$
\STATE  ${\rm s.t.}~ \frac{\partial h_{i, \Delta_{\iota},\ell}^{\mrm k}(\bfp_{l\in \Delta_{\iota}})}{\partial \bfp_i}^\top u_i \geq -\frac{1}{3} \alpha( h_{i,\Delta_{\iota}}^{\mrm k}(\bfp_{l\in \Delta_{\iota}}))$
\STATE  \hspace{0.9cm} $\forall \ell \in I_{i, \Delta_{\iota}, \epsilon}^{\mrm k}(\bfp_{l\in \Delta_{\iota}}),~\forall \Delta_{\iota} \in \mc T_i(\tau_i^{\mrm k})$
%\STATE  ${\rm s.t.}~ \frac{\partial h_{m_i,\ell}^{\mrm k}(\bfp_i, \bfp_{j\in \mc N_i})}{\partial \bfp_i}^\top u_i \geq -\gamma h_{i, m_i}^{\mrm k}(\bfp_i, \bfp_{j\in \mc N_i})^3$
%\STATE  \hspace{0.7cm} $\forall \ell \in I_{\epsilon, \ell, m_i}^{\mrm k}(\bfp_i, \bfp_{j \in {\mc N}_i})$,~$\forall m_i \in \mc T_i(\tau_i^{\mrm k})$
\end{algorithmic} 
\end{algorithm}

Theorem~\ref{th:CBF_dist_hybrid} signifies that each quadcopter can calculate the NCBF condition \eqref{eq:CBF_dist} in a distributed manner to render the admissible set $\mc C_i$ hybrid forward invariant. More specifically, if the state trajectory does not leave the admissible set $\mc C_i$ at the initial time $\tau_i^{\mrm k}$ of the switching interval $t' \in [\tau_i^{\mrm k}, \tau_i^{\mrm k+1})$, the formation of holes during each switching interval $t' \in [\tau_i^{\mrm k}, \tau_i^{\mrm k+1})$ can be prevented by a control input satisfying the condition \eqref{eq:CBF_dist}. Therefore, given a nominal control input $u_{i,{\rm nom}}$ for maximizing the monitoring performance, which will be presented in Section~\ref{sec:cov_cont}, the QP modifying $u_{i,{\rm nom}}$ minimally invasive way to prevent the appearance of holes during each switching interval $t' \in [\tau_i^{\mrm k}, \tau_i^{\mrm k+1})$ can be expressed as
\begin{equation} \label{eq:QP_nohole}
\begin{split}
& u_i^* = \argmin_{u_i \in \R^4} (u_i - u_{i, {\rm nom}})^\top W (u_i - u_{i, {\rm nom}}) \\
& {\rm s.t.}~ \frac{\partial h_{i, \Delta_{\iota},\ell}^{\mrm k}(\bfp_{l\in \Delta_{\iota}})}{\partial \bfp_i}^\top u_i \geq -\frac{1}{3} \alpha( h_{i,\Delta_{\iota}}^{\mrm k}(\bfp_{l\in \Delta_{\iota}})), \\
& \hspace{0.3cm}\forall \ell \in I_{i, \Delta_{\iota}, \epsilon}^{\mrm k}(\bfp_{l\in \Delta_{\iota}}),~\forall \Delta_{\iota} \in \mc T_i(\tau_i^{\mrm k}),
\end{split}
\end{equation}
%with $\| u_i - u_{i, {\rm nom}} \|_W^2 = (u_i - u_{i, {\rm nom}})^T W (u_i - u_{i, {\rm nom}})$ and 
where $W=\diag(1,1,1,w_\lam)$ is the weighting matrix that adjusts the unit differences between $p_i$ and $\lambda_i$.
%introduced to adjust the unit differences between $p_i$ and $\lambda_i$.
The whole procedure is presented in Algorithm~\ref{alg:cov_main}.

%\blue{The following remarks discuss how to address the possible uncertainties }

\begin{remark} \label{rem:pos_unc}
    The proposed algorithm assumes that each quadcopter can accurately estimate its position and attitude. Although the accuracy of the attitude estimation of the quadcopters generally has sufficient performance, e.g., around a few degrees in \cite{Hoffmann10}, the position estimation could contain an error depending on the estimation method, such as GNSS~\cite{Balamurugan16,Bryce07}. If an upper bound of the estimation error is known a priori, one possible remedy to prevent holes even with estimation uncertainties could be to formulate a guaranteed FOV, the area guaranteed to be monitored by a quadcopter even with localization uncertainties, as presented in \cite{Bousias19}. Because the guaranteed FOV constitutes the circular FOV, we can apply the proposed NCBFs to the power diagram based on the guaranteed FOV.
\end{remark}

\begin{remark} \label{rem:com_unc}
    If imperfect communication causes an error between the communicated and the actual state, one could employ the same approach as Remark~\ref{rem:pos_unc} to deal with the error. If a communication edge is completely disconnected, a trio containing the disconnected edge faces difficulty calculating the NCBF. In such cases, a quadcopter could estimate the other quadcopters' state while fixing the state that is difficult to estimate externally, e.g., a zoom level, to a predefined value.
\end{remark}

\begin{remark}
    The proposed method considers the scenario where all quadcopters have a circular image plane. Considering other shapes of image planes would be worthwhile. 
    %But, alternative image shapes can pose challenges in regard to the formulation of Voronoi partitions amenable to the NCBFs formulation. 
    However, other image shapes can pose challenges in formulating Voronoi partitions amenable to the NCBFs formulation.
    For instance, if we consider an elliptic image plane, a Voronoi partition generally becomes a complicated shape \cite{Emiris07}. Still, as conducted in \cite{Gusrialdi2008}, by synchronizing the yaw-angle of the quadcopters, the Voronoi cell can be formulated as a convex polygon to which our proposed method can be extended.
\end{remark}

\subsection{Validity of the Proposed CBFs} \label{ssec:Valid_CBF}

As discussed in Section~\ref{sec:pre}, to show the validness of the HNCBF, we need to exhibit the validity of the proposed NCBFs.
Since the individual gradients of the component functions in \eqref{eq:con_nohole} at each point in time is smooth under Assumption~\ref{asm:non_parallel}, the designed NCBF is a smoothly composed candidate NCBF consisted by a Boolean operator in \eqref{eq:Bool_logic}.
Therefore, as a sufficient condition to guarantee the validity of the proposed NCBF, we investigate whether $\partial h_{i, \Delta_{\iota},\ell}^{\mrm k}(\bfp_{l\in \Delta_{\iota}})/\partial \bfp_i \neq {\bf 0}_{4\times 1}$ holds for each CBF \eqref{eq:h_intri}, \eqref{eq:h_IJK}, and \eqref{eq:h_F} composing NCBF as discussed in Remark~\ref{rem:QP_feas}.

In our previous works \cite{Funada19, Funada20}, complexities appeared in the derivative of the radical center $v_{ijk}$ with respect to $\bfp_i$ hinder us from analyzing the condition $\partial h_{i, \Delta_{\iota},\ell}^{\mrm k}(\bfp_{l\in \Delta_{\iota}})/\partial \bfp_i \neq {\bf 0}_{4\times 1}$ %$L_g h(x)\neq 0_{4\times 1}$ 
and elucidating the properties of each CBF's derivative.
To overcome this difficulty, we introduce a new coordinate frame $\Sigma_d$ on the environment $\mc E$, without loss of generality, that intends to simplify $\partial v_{ijk}/\partial \bfp_i$.
Since we intend to derive $\partial h_{i, \Delta_{\iota},\ell}^{\mrm k}(\bfp_{l\in \Delta_{\iota}})/\partial \bfp_i$, the circles $\bd(\mc F_j)$ and $\bd(\mc F_k)$, namely $\bfp_j$ and $\bfp_k$, are considered to have some constant values in this subsection. 
This corresponds with the fact that $\bfp_j$ and $\bfp_k$ are held constant when we calculate the partial derivative $\partial h_{i, \Delta_{\iota},\ell}^{\mrm k}(\bfp_{l\in \Delta_{\iota}})/\partial \bfp_i$.

The coordinate frame $\Sigma_d$ on the environment $\mc E$ %for calculating the derivatives of CBFs with respect to $\bf p_i$ 
is illustrated in Fig.~\ref{fig:coordinate}.
The introduced coordinate frame is arranged so that its $x_d$-axis and $y_d$-axis correspond with the line $JK$ and the radical axis $L_{jk}$, respectively. 
In addition, we define the position vector of a point $q$ with respect to $\Sigma_d$ as $q^d$, e.g., %$\bfp_i$ and $\bfp_i^d$ 
$v_{ijk} = [v_{ijk}^x~ v_{ijk}^y~0]^\top$ are denoted as $v_{ijk}^d = [v_{ijk}^{x,d}~ v_{ijk}^{y,d}~0]^\top$ after it is transformed from $\Sigma_w$ to the coordinate frame $\Sigma_d$. 
Note that the value of partial derivatives calculated in $\Sigma_d$ differs from those derived in $\Sigma_w$. 
However, as both $\Sigma_w$ and $\Sigma_d$ are fixed in the environment, the analysis in $\Sigma_d$ is enough to inspect whether $\partial h_{i, \Delta_{\iota},\ell}^{\mrm k}/\partial \bfp_i \neq {\bf 0}_{4\times1}$ holds and derives the condition $\partial h_{i, \Delta_{\iota},\ell}^{\mrm k}/\partial \bfp_i = {\bf 0}_{4\times 1}$ otherwise.

Considering the fact that the radical center $v_{ijk}$ is always on the radical axis $L_{jk}$, the following condition is satisfied on the frame $\Sigma_d$.
\begin{align}
    &v_{ijk}^{x,d} = 0, \label{eq:sig_d_vx} \\
    &\frac{\partial v_{ijk}^{x,d}}{\partial \bfp_i} = {\bf 0}_{4\times 1}. \label{eq:sig_d_vy}
\end{align}
Futhermore, $v_{ijk}^{y,d}$ coincides with the $y$-intercept of the radical axis $L_{ij}$, and derived as
\begin{align} \label{eq:radical_cent_y}
    v_{ijk}^{y,d} = \frac{\left({x_i^d}^2 + {y_i^d}^2 - {R_i^d}^2\right) - \left({x_j^d}^2 - {R_j^d}^2\right)}{2y_i^d},
\end{align}
with $R_i:= rz_i/\lambda_i$ denoting the radius of $\bd(\mc F_i)$. Note that $y_i^d\neq 0$ from Assumption~\ref{asm:non_parallel}.

\begin{figure}[t!]
    \centering% 0.55
    \includegraphics[width=0.55\linewidth]{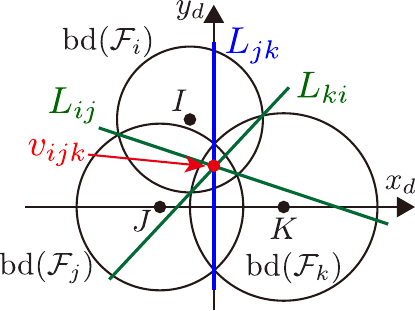}
    \caption{The coordinate frame $\Sigma_d$ introduced for calculation of partial derivatives of designed CBFs.}
    \label{fig:coordinate}
\end{figure}

\begin{figure}[t!]%\label{fig:exp_snap}
    \centering
    \subfloat[]% 0.44
    {\makebox[0.44\hsize][c]{\includegraphics[width=0.44\linewidth]{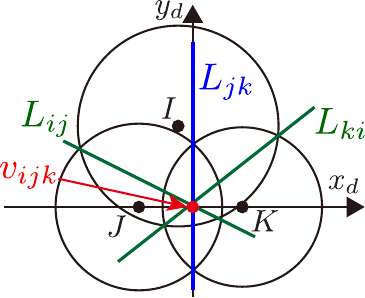}
    \label{fig:diff_h_F}}} \quad 
    \subfloat[]
    {\makebox[0.44\hsize][c]{\includegraphics[width=0.44\linewidth]{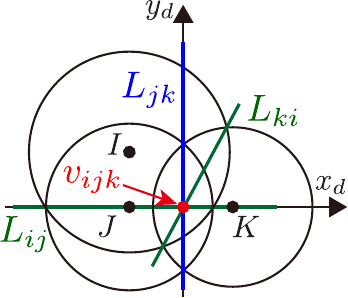}
    \label{fig:diff_h_IJK}}} \quad
    \caption{The condition that makes (a) $\partial h_{i, \Delta_{\iota}, \mc F}^{\mrm k}/\partial \bfp_i = {\bf 0}_{4\times 1}$ and (b) $\partial h_{i, \Delta_{\iota}, IJK}^{\mrm k}/\partial \bfp_i = {\bf 0}_{4 \times 1}$, respectively. (a) $\partial h_{i, \Delta_{\iota}, \mc F}^{\mrm k}/\partial \bfp_i = {\bf 0}_{4\times 1}$ holds if and only if $v_{ijk}$ lies on the line $JK$. (b) $\partial h_{i, \Delta_{\iota}, IJK}^{\mrm k}/\partial \bfp_i = {\bf 0}_{4 \times 1}$ holds if and only if the radical axis $L_{ij}$ is identical to the line $JK$.}
    \label{fig:diff_zero}
\end{figure}

We are now ready to analyze the validness of each CBF. 
The following lemma investigates $\partial h_{i, \Delta_{\iota}, \mc F}^{\mrm k}/\partial \bfp_i$.
\begin{lemma} \label{thm:h_f}
Suppose that Assumption \ref{asm:non_parallel}, $z_i>0$, and $\lam_i > 0$ hold. Then, except for $\bfp_i$ which makes $v_{ijk}$ lie on the line $JK$, the function $h_{i, \Delta_{\iota}, \mc F}^{\mrm k}$ in \eqref{eq:h_F} is a valid CBF when $u_i\in\R^4$.
\end{lemma}
\begin{proof}
See Appendix~\ref{ap:hiF}.
\end{proof}

From \eqref{eq:power_dist_eq} and \eqref{eq:h_F2_dp}, the condition derived in Lemma~\ref{thm:h_f} must make the derivatives of %$d_{\mc P}(\bfp_i, v_{ijk})$, 
$d_{\mc P}(\bfp_j, v_{ijk})$ and $d_{\mc P}(\bfp_k, v_{ijk})$ zero vectors too. Lemma~\ref{thm:h_f} proves such a situation, in which the transitions of $v_{ijk}$ on $L_{jk}$ do not influence the value of $d_{\mc P}(\bfp_j, v_{ijk})$ and $d_{\mc P}(\bfp_k, v_{ijk})$, only happens when both the line $Jv_{ijk}$ and $Kv_{ijk}$ are orthogonal to the radical axis $L_{jk}$ as shown in Fig.~\ref{fig:diff_zero}\subref{fig:diff_h_F}. 
Nevertheless, the condition $v_{ijk}^{d,y} = 0$ is generally non-stationary points, hence the state of Quadcopter~$i$ does not keep satisfying such an ill condition.
%\blue{In practice, the state of Quadcopter $i$ does not keep satisfying the condition $v_{ijk}^{d,y} = 0$ as this condition is generically non-stationary points.} 
Hence, in practice, the condition induced from $h_{i, \Delta_{\iota}, \mc F}^{\mrm k}$ can generically yield a valid controller during a monitoring mission.

We next derive the following lemma that indicates $h_{i, \Delta_{\iota}, JKI}^{\mrm k}$ is a valid CBF.
\begin{lemma} \label{thm:h_JKI}
Suppose that Assumption~\ref{asm:non_parallel}, $z_i > 0$, and $\lam_i > 0$ hold. Then, $h_{i, \Delta_{\iota}, JKI}^{\mrm k}$ is a valid CBF when $u_i \in \R^4$.
\end{lemma}
\begin{proof}
See Appendix~\ref{ap:h_JKI}\blue{.}
\end{proof}
Finally, we investigate $h_{i, \Delta_{\iota}, IJK}^{\mrm k}$ and $h_{i, \Delta_{\iota}, KIJ}^{\mrm k}$.
\begin{lemma} \label{thm:h_IJK}
Suppose that Assumption~\ref{asm:non_parallel} and $z_i>0$ hold. Then, except for $\bfp_i$ which makes the radical axis $L_{ij}$ lie on the line $JK$, the function $h_{i, \Delta_{\iota}, IJK}^{\mrm k}$ is a valid CBF when $u_i \in \R^4$.
\end{lemma}
\begin{proof}
See Appendix~\ref{ap:h_IJK}\blue{.}
\end{proof}

Since $h_{i, \Delta_{\iota}, IJK}^{\mrm k}$ and $h_{i, \Delta_{\iota}, KIJ}^{\mrm k}$ are symmetric with respect to $y_d$-axis, the following lemma immediately holds from Lemma~\ref{thm:h_IJK}.
\begin{lemma} \label{thm:h_KIJ}
Suppose that Assumption~\ref{asm:non_parallel} and $z_i>0$ hold. Then, except for $\bfp_i$ which makes the radical axis $L_{ki}$ lie on the line $JK$, the function $h_{i, \Delta_{\iota}, KIJ}^{\mrm k}$ is a valid CBF when $u_i \in \R^4$.
\end{lemma}

Fig.~\ref{fig:diff_zero}\subref{fig:diff_h_IJK} illustrates the condition, where $\partial h_{i, \Delta_{\iota}, IJK}^{\mrm k}/ \bfp_i = {\bf 0}_{4\times 1}$ holds. 
Lemma~\ref{thm:h_IJK} and \ref{thm:h_KIJ} revealed that the derivatives of the proposed CBFs do not become zero vectors except for limited conditions. 
Similar to the discussion about $h_{i, \Delta_{\iota}, \mc F}^{\mrm k}$, Quadcopter $i$ does not keep satisfying this condition. %in practice.
Therefore, the controller can render an input that satisfies the constraint incurred by $h_{i, \Delta_{\iota}, IJK}^{\mrm k}$ and $h_{i, \Delta_{\iota}, KIJ}^{\mrm k}$ in practice.

From the discussion in this subsection, it is ensured that, except for a few pathological conditions, Algorithm~\ref{alg:cov_main} employing the QP \eqref{eq:QP_nohole} can find the control input preventing a hole in-between neighboring agents determined by triangular subgraphs $\mc T_i(t)$. 
This results in the following theorem about the validity of the proposed NCBF.
\begin{theorem} \label{thm:Valid_HNCBF}
Suppose that the Assumption~\ref{asm:non_parallel}, $z_i > 0$, $\lam_i>0$, and $\mc C_i^{\mrm k} \neq \emptyset$ hold. Then, with $u_i\in \R^4$, $h_{i, \Delta_{\iota}}^{\mrm k}$ in \eqref{eq:h_mi} is a valid NCBF except for $\bfp_{l\in \Delta_{\iota}}$ satisfying either of the following conditions.
\begin{itemize}
    \item $(v_{ijk} \in JK) ~\land~ 
    (\exists \ell \in I_{i, \Delta_{\iota}, \epsilon}^{\mrm k}~{\rm s.t.}~h_{i, \Delta_{\iota},\ell}^{\mrm k} = h_{i, \Delta_{\iota},\mc F}^{\mrm k})$
    \item $(L_{ij} = JK) ~\land~ 
    (\exists \ell \in I_{i, \Delta_{\iota}, \epsilon}^{\mrm k}~{\rm s.t.}~h_{i, \Delta_{\iota},\ell}^{\mrm k} = h_{i, \Delta_{\iota},IJK}^{\mrm k})$
    \item $(L_{ki} = JK) ~\land~ 
    (\exists \ell \in I_{i, \Delta_{\iota}, \epsilon}^{\mrm k}~{\rm s.t.}~h_{i, \Delta_{\iota},\ell}^{\mrm k} = h_{i, \Delta_{\iota},KIJ}^{\mrm k})$
\end{itemize}
\end{theorem}
\begin{proof}
From Lemma\,\ref{thm:h_f}-\ref{thm:h_KIJ}, $\partial h_{i, \Delta_{\iota},\ell}^{\mrm k}(\bfp_{l\in \Delta_{\iota}})/\partial \bfp_i \neq {\bf 0}_{4\times 1}$ holds for all agents composing $\Delta_{\iota}$, in the domain $\mc D_i^{\mrm k}$ except for $\bfp_{l\in \Delta_{\iota}}$ satisfying either of the above three conditions. Namely, if neither of the three conditions holds, there exists an input $u_i\in \R^4$ satisfying constraints in both forms \eqref{eq:CBF_cent} and \eqref{eq:CBF_dist}.
\end{proof}
Theorem~\ref{thm:Valid_HNCBF} guarantees that the QP \eqref{eq:QP_nohole} in each switching interval ${\mrm k}\in \bN$ yields a control input preventing the appearance of a hole in most parts of its domain. Hence, from Theorem~\ref{thm:paul_HNCBF}, the proposed distributed algorithm can render the set $\mc C_i$ hybrid forward invariant except for the pathological conditions stated in Theorem~\ref{thm:Valid_HNCBF}. 
Note that Theorem~\ref{thm:Valid_HNCBF} presumes $\mc C_i^{\mrm k}\neq \emptyset$ holds. This corresponds with the fact that Quadcopter~$i$ always has a trivial solution to eliminate a hole, where Quadcopter~$i$ expands its FOV to cover other Quadcopters FOVs $\mc F_l,~\forall l\in \mc N_i$ entirely.
In summary, the proposed QP \eqref{eq:QP_nohole} has a feasible solution for preventing holes, except for a few limited conditions that do not cause severe issues in practice.

\section{Coverage Control} \label{sec:cov_cont}

Having designed the control method that prevents the appearance of a hole in-between a team, we present the nominal control input $u_{i, {\rm nom}}$ in \eqref{eq:QP_nohole} that maximizes the monitoring quality. 
The proposed control input improves a coverage cost that simultaneously serves two objectives: (a) quantifying the performance of all the cameras, including the appropriate zoom level and distance from the ground and (b) reducing the overlaps of FOVs among a team.
%reducing the overlap between FOVs of quadcopters.

%The quality of surveillance of a point $q\in \mc Q$ can be quantified with the model proposed in \cite{Arslan2018},
The monitoring performance of a point $q\in \mc Q$ can be quantified with the model proposed in \cite{Arslan2018},
\begin{equation}\label{eq:sensing_function}
f(\bfp_i, q) = f_{\pers}(\bfp_i, q)f_{\res}(\bfp_i, q),
\end{equation}
with $f_{\pers}(\bfp_i, q)$ characterizing the perspective quality
\begin{equation*}\label{eq:f_pers}
%f_{\pers}(\bfp_i, q) := \frac{\lambda_i^2+r^2}{r^2}\left(\frac{z_i}{\|p_i-q\|}-\frac{\lambda_i}{\lambda_i^2+r^2}\right)
f_{\pers}(\bfp_i, q) := \frac{\sqrt{\lam_i^2 + r^2}}{\sqrt{\lam_i^2 + r^2}-\lam_i} \left( \frac{z_i}{\| q-p_i \|} - \frac{\lambda_i}{\sqrt{\lam_i^2 + r^2}} \right),
\end{equation*}
and $f_{\res}(\bfp_i, q)$ the loss of resolution
\begin{equation*}\label{eq:f_res}
f_{\res}(\bfp_i, q):= \left( \frac{\lambda_i}{\sqrt{\lambda_i^2+r^2}} \right)^\kappa \text{exp}\left(-\frac{(\|q-p_i\|-M)^2}{2\sigma^2} \right).
\end{equation*}
Note that the parameters $\kappa, \sigma > 0$ signify the spatial resolution property of a camera, and $M>0$ models the desired distance to capture an environment.
Note that $f_\textrm{res}$ renders a higher value as the focal length $\lambda_i$, namely the zoom level, increases, though the size of FOV becomes smaller.
If a point $q\in \mc Q$ is not in $\mc F_i$, then we set
\begin{equation*}\label{eq:sensing_performance_zero}
f(\bfp_i, q) = 0, \quad q\notin \mc F_i.
\end{equation*}

The quality of coverage achieved by the team of quadcopters is characterized as
\begin{equation}\label{eq:cost_coverage_quality}
 \mc H_{\mc M}(\bfp) = \int_{\mc Q} \max_{i\in \N}f(\bfp_i, q) \phi(q) \rd q,
\end{equation}
with $\bfp = [\bfp_1^\top,\cdots, \bfp_n^\top]^\top$ and a density function $\phi(q)$, which achieves a higher value when the mission space $\mc Q$ is well monitored.
By introducing the region allocated to Quadcopter~$i$ according to the conic Voronoi diagram \cite{Arslan2018}
\begin{equation}\label{eq:Voronoi_partition}
    \mc V_i(\bfp) \!=\! \{q\in\mc Q\cap \mc F_i\mid f(\bfp_i, q)\!\geq\!f(\bfp_j, q), j\in\mc N \backslash \{ i\} \},
\end{equation}
the cost in \eqref{eq:cost_coverage_quality} can be rewritten as
\begin{equation}\label{eq:cost_coverage_quality_Voronoi}
    \mc H_{\mc M}(\bfp) = \sum_{i\in\mc N}\int_{\mc V_i(\bfp)} f(\bfp_i, q) \phi(q) \rd q.
\end{equation}

However, the coverage cost \eqref{eq:cost_coverage_quality_Voronoi} does not penalize the overlaps of FOVs, which might lead several quadcopters to monitor a same region.
These overlaps of FOVs can deteriorate the monitoring performance since a team could potentially monitor a wider area by prompting those quadcopters to cover different areas.
Let us define the region monitored by Quadcopter~$i$ where another quadcopter in a team has a better sensing quality as
\begin{equation*}\label{eq:Voronoi_partition_overlap}
    \bar{\mc V}_i(\bfp) = \{q\in\mc Q \cap \mc F_i\mid f(\bfp_i, q)< f(\bfp_j, q), j\in\mc N \backslash \{ i\}\}.
\end{equation*}
Then, the performance loss induced by the overlaps of FOVs is evaluated as
\begin{equation}\label{eq:cost_overlap}
    \mc H_{\mc O}(\bfp) = \sum_{i\in\mc N}\int_{\bar{\mc V}_i(\bfp)} f(\bfp_i, q) \phi(q) \rd q.
\end{equation}

% 目的関数の導入
As we are interested in minimizing the overlap of FOVs evaluated in \eqref{eq:cost_overlap} while maximizing the coverage quality in \eqref{eq:cost_coverage_quality_Voronoi}, the overall objective can be encoded as 
\begin{equation}\label{eq:locational_cost}
 \mc H(\bfp) = \mc H_{\mc M}(\bfp) - w \mc H_{\mc O}(\bfp),
\end{equation}
with the weight $w>0$.

The nominal control input $u_{i, \rm nom}$ to maximize the objective function \eqref{eq:locational_cost} can be obtained in a distributed manner by following a gradient-ascent method as 
\begin{align}\label{eq:pi_gradient_H}
    u_{i, {\rm nom}} =  \frac{\partial\mc H(\bfp)}{\partial \bfp_i}, \quad i\in\mc N, %\blue{\Gamma}
\end{align}
%\blue{where $\Gamma = \diag(\gamma, \gamma, \gamma, \gamma_\lambda)$ is the gain for the coverage control input. Note that $\Gamma$ also adjusts the unit differences between $p_i$ and $\lambda_i$.}
In our conference work \cite{Funada19}, the gradient of $\mc H$ is derived as %\red{(detailed explanation written in ICRA19 might be required...)}
\begin{equation}\label{eq:gradient_H}
 \frac{\partial\mc H(\bfp)}{\partial \bfp_i} = \int_{\mc V_i}\frac{\partial f(\bfp_i, q)}{\partial\bfp_i}^\top\phi(q)dq - w \int_{\bar{\mc V}_i}\frac{\partial f(\bfp_i, q)}{\partial\bfp_i}^\top\phi(q)\rd q.
\end{equation}
The following theorem in our conference paper \cite{Funada19} guarantees the convergence of the proposed coverage input $u_{i, \rm nom}$.
%\red{As written in \cite{Schwager2011}, we could say the convergence to the local optimal solution}
\begin{theorem} %\cite[Theorem~1]{Funada19}
Let Quadcopter $i$, with state $\bfp_i=[p_i^\top~ \lambda_i]^\top$, evolve according to the control law $\dot\bfp_i=u$, with
\begin{equation}\label{eq:control_input}
 u = \int_{\mc V_i}\frac{\partial f(\bfp_i, q)}{\partial\bfp_i}^\top\phi(q)\rd q - \int_{\bar{\mc V}_i}\frac{\partial f(\bfp_i, q)}{\partial\bfp_i}^\top\phi(q)\rd q,
\end{equation}
then, as $t\to\infty$, the team of quadcopters will converge to a critical point of the locational cost in \eqref{eq:locational_cost}.
\end{theorem}
\begin{proof}
See Theorem~1 in our conference work \cite{Funada19}.
\end{proof}

\begin{remark} \label{rem:computation}
    The computational complexity of the proposed control strategy can be calculated by evaluating the three processes involved: the computation of the power diagram, the QP optimization problem \eqref{eq:QP_nohole}, and the coverage control input \eqref{eq:control_input}. First, the computational complexity of finding the Voronoi cell of Quadcopter~$i$ in a distributed manner is known as $\mc O( M_i\log(M_i) )$, where $M_i$ stands for the number of neighbors Quadcopter~$i$ has to consider. In the worst-case scenario, where Quadcopter~$i$'s FOV is overlapped with all others' FOVs, the computational complexity is $\mc O( (n-1)\log(n-1) )$. Second, the QP \eqref{eq:QP_nohole} can be solved by the interior point method, the computational complexity of which is $\mc O (N^{3.5})$, where $N$ stands for the dimension of the optimization variable. Since the optimization variable of the QP~\eqref{eq:QP_nohole} is $\bfp_i~\in~\R^4$, ($N=4$), its computational burden is low. Finally, the coverage control input is calculated by discretizing the integral in~\eqref{eq:control_input}, similar to \cite{Schwager2011}. Its computational complexity %of this coverage control law 
    depends on how many grids we utilize for a discretized approximation. In the current implementation, the computation time of~\eqref{eq:control_input} is around several milliseconds. % with a reasonably accurate approximation. 
    More details on the computational complexity of calculating the Voronoi cell, the QP with CBF, and the coverage control law can be found in \cite{He2018}, \cite{Wang17_swarm}, and \cite{Schwager2011}, respectively.
\end{remark}
\section{Simulations and Experiments}

\subsection{Comparison between NCBF $h_{i, \Delta_{\iota}}^{\mrm k}$ and CBF $h_{i, \Delta_{\iota}, \mc F}^{\mrm k}$} \label{ssec:sim_comp_trio}
We first exemplify how the proposed NCBF $h_{i, \Delta_{\iota}}^{\mrm k}$ prevents a hole while reducing too conservative behavior that appears in the case if we utilize $h_{i, \Delta_{\iota}, \mc F}^{\mrm k}$ alone, as discussed in Section~\ref{ssec:hole_CBF}. 
Minimizing how much the nominal controller, such as the coverage controller \eqref{eq:pi_gradient_H}, is modified is crucial for optimizing the coverage performance because too conservative constraints hinder the controller from maximizing the coverage cost, resulting in a significant deterioration of the monitoring performance. 
%\blue{Minimizing the modification from the nominal controller, such as the coverage controller \eqref{eq:pi_gradient_H}, is crucial for optimizing the coverage performance because too conservative constraints hinder the controller from maximizing the coverage cost, resulting in a significant deterioration of the monitoring performance.}
Let us consider a team of three quadcopters composing $\Delta_{\iota} = \{i, j, k \}$, where the green Quadcopter~$i$ moves upwards while the other two quadcopters stay in the initial positions, as shown in Fig.~\ref{fig:CBF_NCBF_snap}. 
The power diagram is depicted in blue, with $\triangle{IJK}$ in red.
Note that the intersection of three boundaries corresponds with the radical center $v_{ijk}$.

Fig.~\ref{fig:CBF_NCBF_snap}\subref{fig:sim_trio_Initial} depicts the initial configuration of the three quadcopters. 
Figs.~\ref{fig:CBF_NCBF_snap}\subref{fig:sim_trio_NCBF1} and \subref{fig:sim_trio_NCBF2} show the simulation results with the proposed NCBF $h_{i, \Delta_{\iota}}^{\mrm k}$. 
Because $h_{i, \Delta_{\iota}}^{\mrm k}$ allows the radical center $v_{ijk}$ to leave $\mc F_i$ when $v_{ijk} \notin \triangle{IJK}$ as shown in Fig.~\ref{fig:CBF_NCBF_snap}\subref{fig:sim_trio_NCBF1}, Quadcopter $i$ successfully passes through in-between two agents without being disturbed by $h_{i, \Delta_{\iota}}^{\mrm k}$. 
In contrast, since $h_{i, \Delta_{\iota}, \mc F}^{\mrm k}$ does not allow $v_{ijk}$ to leave $\mc F_i$, $h_{i, \Delta_{\iota}, \mc F}^{\mrm k}$ overly modifies the nominal input for Quadcopter~$i$ even if the nominal input does not produce a hole. 
This overly modified input makes $\mc F_i$ smaller so that Quadcopter $i$ can keep moving upwards while $v_{ijk}$ is inside of $\mc F_i$, though such modification is unnecessary to prevent a hole, as shown in Figs.~\ref{fig:CBF_NCBF_snap}\subref{fig:sim_trio_CBF1}-\subref{fig:sim_trio_CBF3}.
%This overly modified input makes $\mc F_i$ smaller so that Quadcopter $i$ can keep moving upwards while the radical center $v_{ijk}$ is inside of $\mc F_i$.

%More specifically, the CBF makes the FOV smaller so that Quadcopter can keep moving upwards while the radical center is inside of ABC. 

%The top row of Fig.~\ref{fig:CBF_NCBF_snap} shows the simulation results with $h_{i, \mc F}$ only.
%Since $h_{i, \mc F}$ does not allow the radical center $v_{ijk}$ to leave $\mc F_i$, Quadcopter $i$ cannot smoothly pass through in-between two other agents to reach the upper region. 

The evolution of $h_{i, \Delta_{\iota}}^{\mrm k}$ and $h_{i, \Delta_{\iota}, \mc F}^{\mrm k}$ in each simulation is shown in Fig.~\ref{fig:CBF_NCBF_comp}. $h_{i, \Delta_{\iota}}^{\mrm k}$ in the proposed method switches its value from $h_{i, \Delta_{\iota}, \mc F}^{\mrm k}$ to $-h_{i, \Delta_{\iota}, JKI}^{\mrm k}$ around 0.1s and back at 3.1s, accompanied by the nonsmooth transitions in Fig.~\ref{fig:CBF_NCBF_comp}. 
In contrast, if we utilize only $h_{i, \Delta_{\iota}, \mc F}^{\mrm k}$, its value continues to keep a small value until $5.2$\,s, resulting in too conservative modification of the nominal input.

Note that $h_{i, \Delta_{\iota}}^{\mrm k}$ in the proposed algorithm takes a significantly large value from $1$\,s to $2.2$\,s. This peak is associated with the configuration of a team that makes $\triangle{IJK}$ a degenerate triangle, namely, the unsatisfaction of Assumption~\ref{asm:non_parallel}, which leads the denominator of \eqref{eq:h_IJK} to a small value. 
Nevertheless, since the area of $\triangle{IJK}$ is close to zero and the radical center $v_{ijk}$ exists at infinity under such a degeneracy, the condition $v_{ijk} \notin \triangle{IJK}$ is satisfied. 
Therefore, we can regard that the team is still in the safe set defined as the zero-superlevel set of \eqref{eq:h_mi}. 
Still, a too large value in the NCBF % and its gradient 
could cause a numerical error when we solve the QP. 
To circumvent this problem in the implementation, if one of $\{-h_{i,\Delta_{\iota},IJK}, -h_{i,\Delta_{\iota},JKI}, -h_{i,\Delta_{\iota},KIJ} \}$ is listed in $I_{i, \Delta_{\iota}, \epsilon}$ and takes a significantly large value, we allow the QP not to evaluate the corresponding constraint because the team is deemed not to break the constraint in the next few steps.

\begin{figure}[t!]%\label{fig:exp_snap}
\vspace{-2mm}
    \centering
    \subfloat[Initial configuration]
    {\makebox[0.31\hsize][c]{\includegraphics[width=0.3\linewidth]{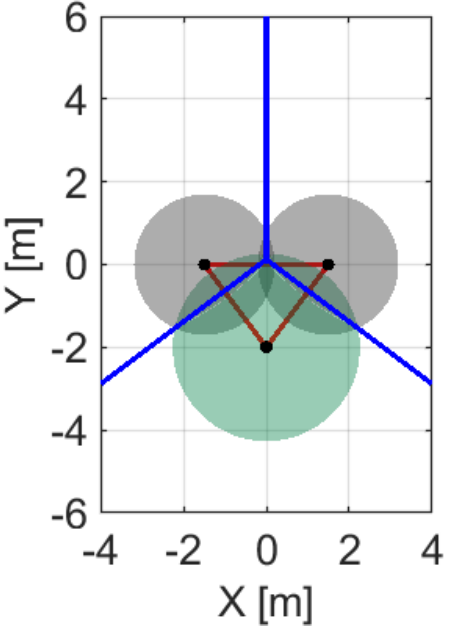}
    \label{fig:sim_trio_Initial}}} \,
    \subfloat[$t=1$\,s, $h_{i, \Delta_{\iota}}^{\mrm k}$]
    {\makebox[0.31\hsize][c]{\includegraphics[width=0.3\linewidth]{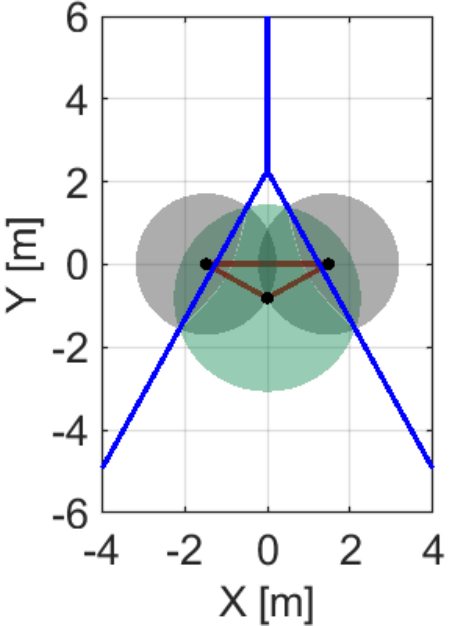}
    \label{fig:sim_trio_NCBF1}}} \,
    \subfloat[$t= 5$\,s, $h_{i, \Delta_{\iota}}^{\mrm k}$]
    {\makebox[0.31\hsize][c]{\includegraphics[width=0.3\linewidth]{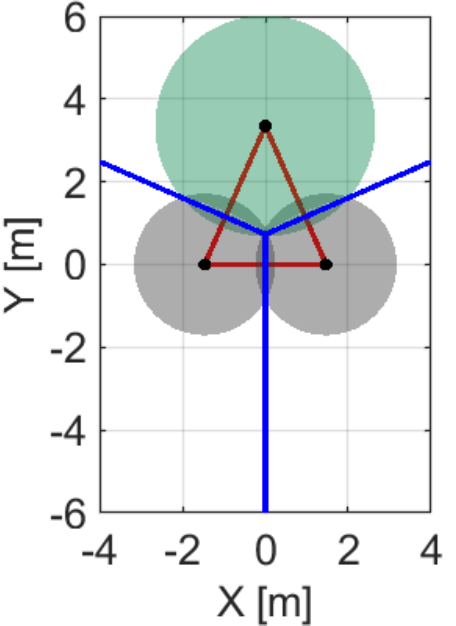}
    \label{fig:sim_trio_NCBF2}}} \\
    \subfloat[$t=2$\,s, $h_{i, \Delta_{\iota}, \mc F}^{\mrm k}$]
    {\makebox[0.31\hsize][c]{\includegraphics[width=0.3\linewidth]{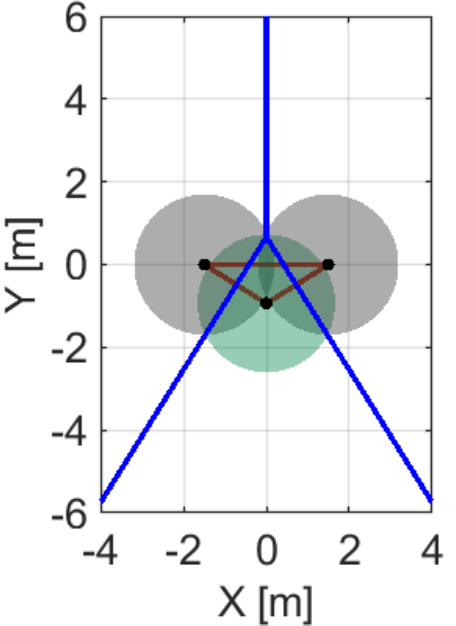}
    \label{fig:sim_trio_CBF1}}} \,
    \subfloat[$t=4$\,s, $h_{i, \Delta_{\iota}, \mc F}^{\mrm k}$]
    {\makebox[0.31\hsize][c]{\includegraphics[width=0.3\linewidth]{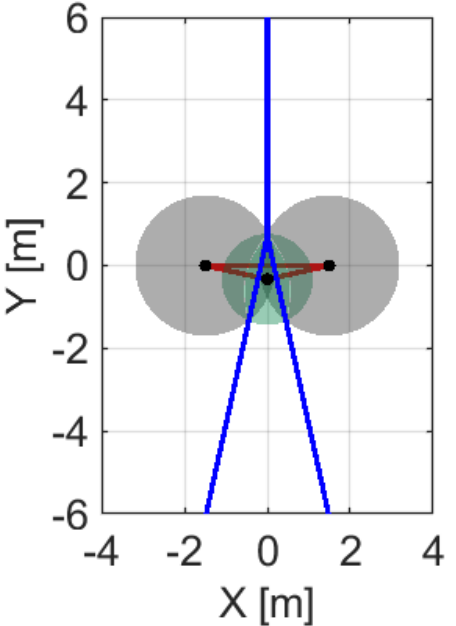}
    \label{fig:sim_trio_CBF2}}} \,
    \subfloat[$t=7$\,s, $h_{i, \Delta_{\iota}, \mc F}^{\mrm k}$]
    {\makebox[0.31\hsize][c]{\includegraphics[width=0.3\linewidth]{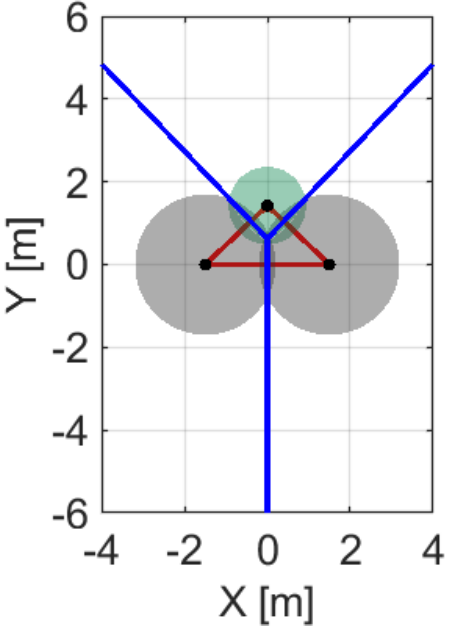}
    \label{fig:sim_trio_CBF3}}}
    \caption{Snapshots of the simulations to compare the proposed algorithm (upper row) and the method utilizing only $h_{i, \Delta_{\iota}, \mc F}^{\mrm k}, \forall i \in \mc N$ (lower row). The proposed algorithm successfully prevents the appearance of a hole while eliminating the conservative behavior shown in the lower row.}
    \label{fig:CBF_NCBF_snap}
\end{figure}

\begin{figure}
    \centering
    \includegraphics[trim = 0cm 0.1cm 0cm 0.4cm, clip=true, width=0.75\linewidth]{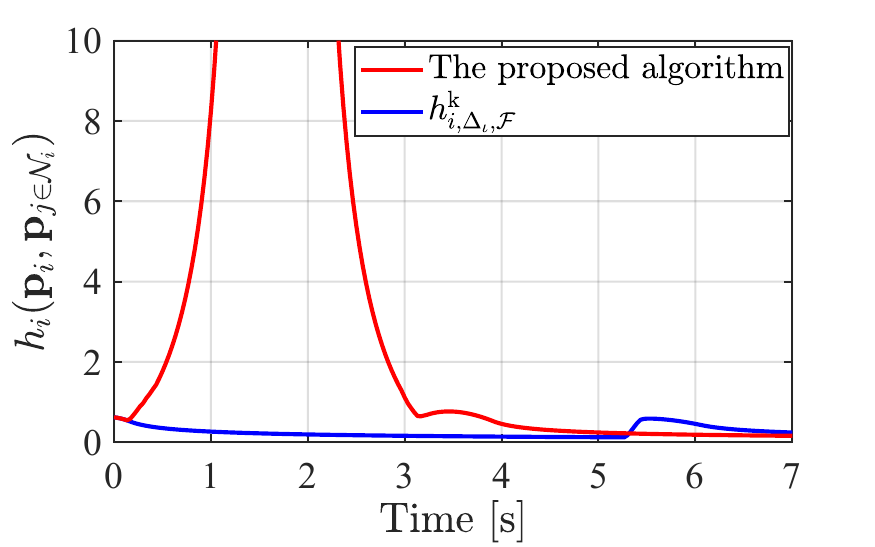}
    \caption{Evolution of $h_{i, \Delta_{\iota}}^{\mrm k}$ and $h_{i, \Delta_{\iota}, \mc F}^{\mrm k}$. There exist nonsmooth points in the transition of $h_{i, \Delta_{\iota}}^{\mrm k}$ in the proposed method at $t=0.1$\,s and $t=3.1$\,s, respectively. These points exhibit the switching from $h_{i, \Delta_{\iota}, \mc F}^{\mrm k}$ to $-h_{i, \Delta_{\iota}, JKI}^{\mrm k}$ and back induced from the max operator in the proposed NCBF $h_{i, \Delta_{\iota}}^{\mrm k}$.}
    \label{fig:CBF_NCBF_comp}
\end{figure}

\subsection{Comparison to the previous work \cite{Funada20}}

% Simulationの動機づけ，parameterに関しての説明

Let us next compare the proposed method with the authors' previous work \cite{Funada20} to demonstrate Algorithm~1 can prevent the appearance of holes even in scenarios where the one in \cite{Funada20} cannot. Fig.~\ref{fig:comp_feas}\subref{fig:initial_comp} shows the initial configuration of the five quadcopters composing two trios. The green Quadcopter~$i$ has the nominal input decreasing its altitude with 3\,m/s, while the other four quadcopters stay in the initial positions. The power diagram is depicted in blue, with each trio shown in red. The parameters are set as $\omega_{\lam} = 3.0\times 10^6$ and $\epsilon = 0.02$, with the extended class-$\mc K$ function $\alpha(h) = h^3$.

% Simulationの結果
The final configuration of the team with the proposed algorithm and the method in \cite{Funada20} are shown in Fig.~\ref{fig:comp_feas}\subref{fig:fin_comp_prop} and Fig.~\ref{fig:comp_feas}\subref{fig:fin_comp_ICRA20}, respectively. 
First, the proposed method successfully prevents the appearance of the holes by modifying the nominal input. The prevention of the holes can also be confirmed in Fig.~\ref{fig:comp_feas_NCBF}, where each trio's NCBF, shown in red and blue lines, keeps positive values. On the other hand, the method in \cite{Funada20} suffers from the oscillation of each NCBF, and both of the trios' NCBF cannot remain positive, as depicted in violet and orange lines. After 3.1\,s, the left trio is disconnected, as shown in Fig.~\ref{fig:comp_feas}\subref{fig:fin_comp_ICRA20}, and only the right trio's NCBF is evaluated.
The simulation study demonstrates the proposed algorithm successfully prevents the appearance of holes even in scenarios where the method in \cite{Funada20} cannot.

\begin{figure}[t!]%\label{fig:exp_snap}
\vspace{-2mm}
    \centering
    \subfloat[Initial configuration]
    {\makebox[0.31\hsize][c]{\includegraphics[trim = 0.1cm 0cm 1.1cm 0.6cm, clip=true, width=0.31\linewidth]{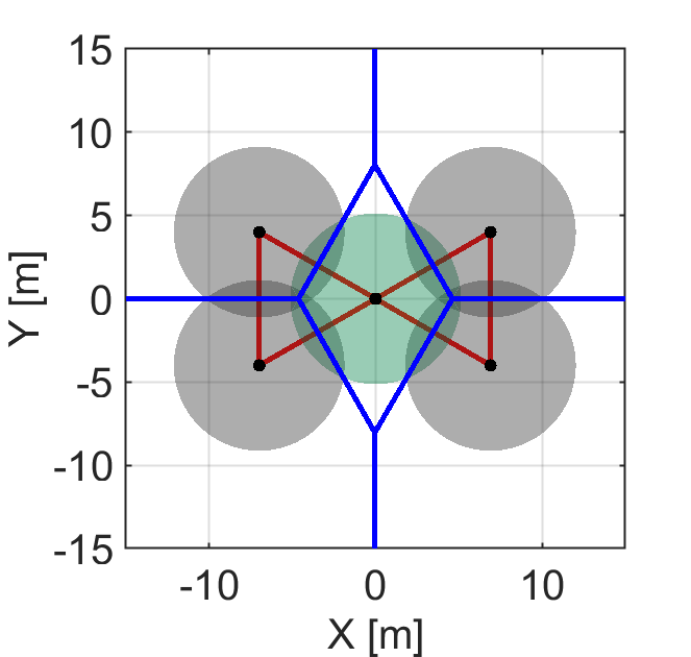}
    \label{fig:initial_comp}}} \,
    \subfloat[$t=4$\,s with the proposed method]
    {\makebox[0.31\hsize][c]{\includegraphics[trim = 0.1cm 0cm 1.1cm 0.6cm, clip=true, width=0.31\linewidth]{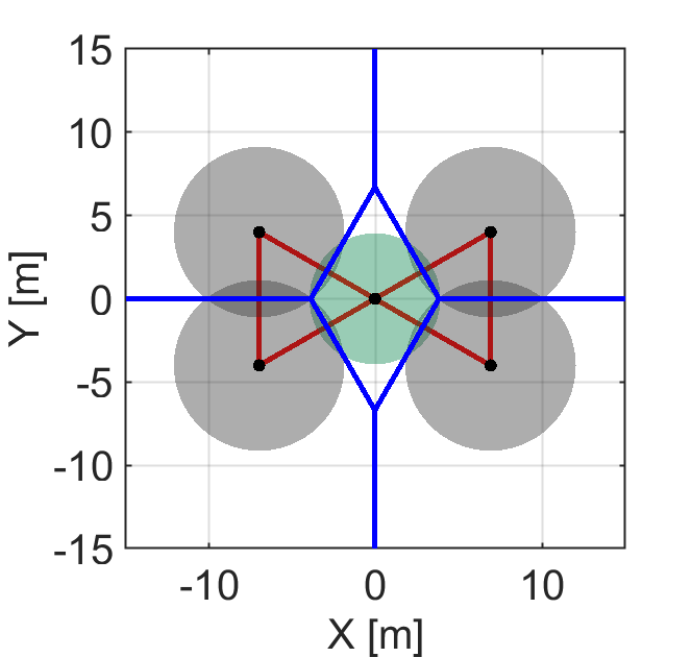}
    \label{fig:fin_comp_prop}}} \,
    \subfloat[$t=4$\,s with the method in \cite{Funada20}.]
    {\makebox[0.31\hsize][c]{\includegraphics[trim = 0.1cm 0cm 1.1cm 0.6cm, clip=true, width=0.31\linewidth]{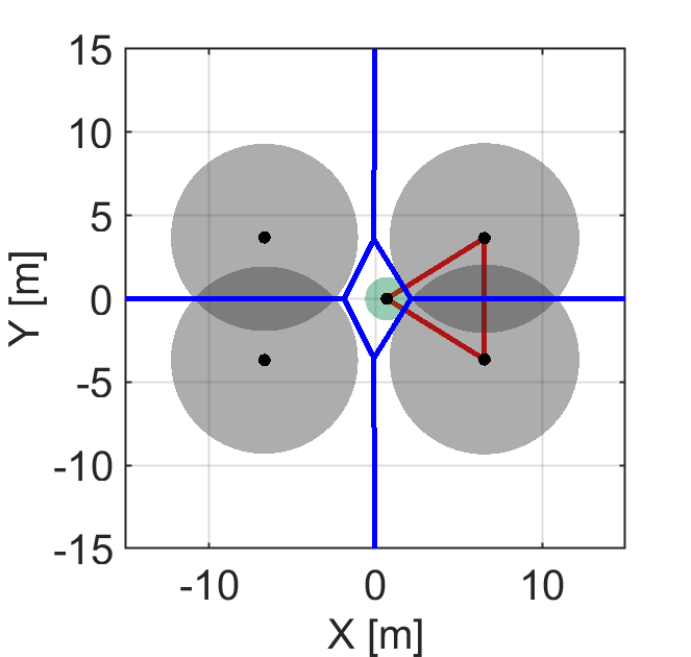}
    \label{fig:fin_comp_ICRA20}}}
    \caption{Snapshots of the simulations to compare the proposed algorithm and the method in \cite{Funada20}.}
    \label{fig:comp_feas}
\end{figure}

\begin{figure}[t!]
    \centering
    \includegraphics[width=0.6\linewidth]{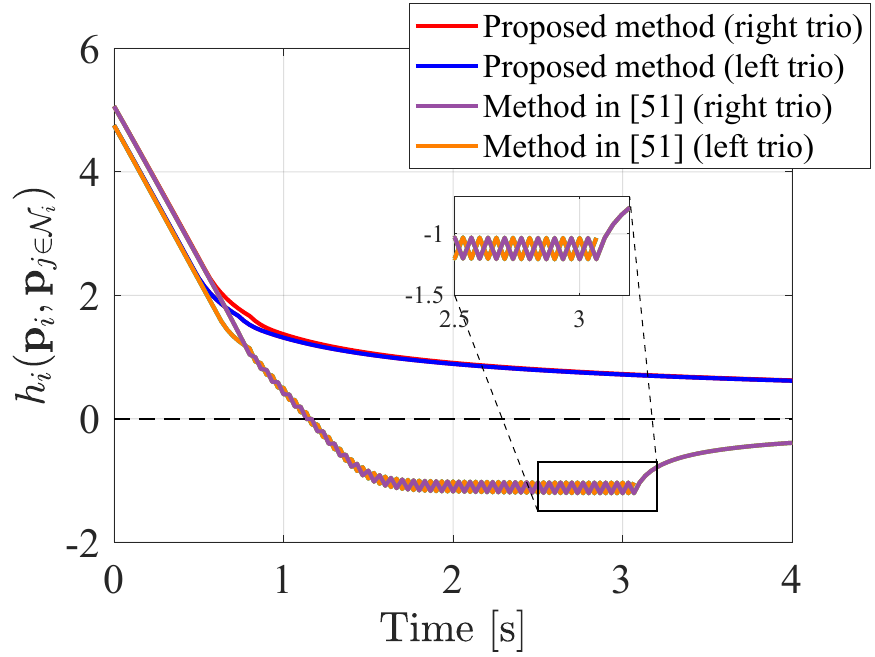}
    \caption{Evolutions of the NCBFs.}
    \label{fig:comp_feas_NCBF}
\end{figure}

% この結果が得られる理由
The main reason the method in \cite{Funada20} violates the constraint is that Quadcopter~$i$ with the method in \cite{Funada20} only evaluates the minimum NCBF among all the trio it belongs to and the NCBFs in the range of $\epsilon$ from the minimum NCBF.
%with the method in \cite{Funada20} only evaluates the NCBFs, which take the minimum value among all the trios it belongs to and the one in the range of $\epsilon$ from the minimum. 
%The main reason for the difference in the results is that Quadcopter~$i$ with the method in \cite{Funada20} only evaluates the NCBFs, which take the minimum NCBF among all the trios Quadcopter~$i$ belongs and the one in the range of $\epsilon$ from the minimum. 
Therefore, if the pre-determined $\epsilon$ is smaller than the quadcopters' speed, the NCBF could be violated by showing chattering, as shown in Fig.~\ref{fig:comp_feas_NCBF}. 
%Note that $\epsilon = 0.02$ utilized in the above simulation study is reasonable value considering the paper first presented NCBF [48] utilizes much smaller value, $\epsilon = 0.007$. 
Algorithm~\ref{alg:cov_main} solves this issue by means of Theorem~4. Theorem~4 utilizes the symmetric property of the proposed NCBF to develop a constraint that each quadcopter can evaluate locally and shows that each quadcopter has to evaluate all the NCBFs of the trio that each quadcopter belongs, as shown in \eqref{eq:QP_nohole}. Because of this theoretical guarantee and improvements in Algorithm~\ref{alg:cov_main}, the proposed method successfully prevents holes.

\subsection{Simulation with Nine Quadcopters} \label{ssec:sim_nine}

In the following simulation, the proposed algorithm is implemented with the coverage control law to the team of nine quadcopters. We then compare the algorithm with the nominal input in \eqref{eq:pi_gradient_H}. The size of the mission space is $60\times 60\,{\rm m}^2$, where the density function is depicted as the contour maps in Fig.~\ref{fig:sim9}.
The parameters are set as $\kappa = 4, \sigma = 3, M = 11$, $w=0.4$, $w_\lam = 3.0\times 10^6$, $\epsilon = 0.2$  %\blue{$\gamma = 4.8\times 10^{-4}$, $\gamma_{\lambda} = 6\times 10^{-11}$} 
with the extended class-$\mc K$ function $\alpha(h) = h^3$. 
%The contour maps in Fig.~\ref{fig:sim9} show the density function having three peaks.
%The density function is shown as the contour maps in Fig.~\ref{fig:sim9}, which has three peaks 

We then run the proposed algorithm from the initial configuration in Fig.~\ref{fig:sim9}\subref{fig:sim9_k1}, where the red and blue lines show the triangular subgraphs $\mc T_i, \forall i \in \mc N$ in the graph $\mc G$ and the power diagram, respectively. 
The snapshots of the simulation are depicted in Fig.~\ref{fig:sim9}. 
In the transition from Fig.~\ref{fig:sim9}\subref{fig:sim9_k1} to Fig.~\ref{fig:sim9}\subref{fig:sim9_k2100}, we can confirm several triangular subgraphs in the team show switching, allowing them to move smoothly to the important region. 
At 4500\,step, a team completely covers the peak of the density function in the lower part and continuously moves to cover two other peaks on the top part, as shown in Fig.~\ref{fig:sim9}\subref{fig:sim9_k4500}. 
At the final configuration in Fig.~\ref{fig:sim9}\subref{fig:sim9_k10000}, a team successfully covers three peaks with no hole.

The evolution of HNCBFs is shown in Fig.~\ref{fig:sim9_HNCBF}, where all of them take the positive value except for the 5642\,step. 
%As explained at the end of Section~\ref{ssec:hole_def}, this is caused by the sudden appearance of a hole produced from the generation of a new triangular subgraph on the left upper quadcopters. 
As discussed at the end of Section~\ref{ssec:hole_def}, this sudden appearance of a hole is caused when a team closes a dent in the upper left part of Fig.~\ref{fig:sim9}\subref{fig:sim9_k4500}.
In other words, this instantaneous hole is caused by the inherent nature of the hybrid forward invariance \cite{Glotfelter19} that does not address whether a system remains within the admissible set at the time switching occurs. Nevertheless, because of the robustness of the proposed algorithm, which stems from a general characteristic of CBFs yielding the admissible set attractive \cite{Glotfelter19,Ames2019_CBF_thapp}, the proposed algorithm successfully eliminates this suddenly generated hole. Note that this formation of a hole does not violate the hybrid forward invariance of the set in Definition~\ref{def:hyb_fw_inv}, as the concept presumes the state does not leave the safe set at each switching time $\tau_k,~\forall k \in {\mathbb N}$.

Fig.~\ref{fig:sim9_cov} shows the snapshots of the simulation results with the nominal input \eqref{eq:pi_gradient_H}, where we set the same initial configuration to Fig.~\ref{fig:sim9}\subref{fig:sim9_k1}. 
As shown in Fig.~\ref{fig:sim9_cov}\subref{fig:sim9_cov_k500}, from an early stage of the simulation, a team exhibits the unsurveilled area in-between quadcopters, which is developed from a hole created in-between a trio. 
In contrast, since a team does not need to consider an HNCBF constraint, a team moves a bit faster than the proposed method. 
As a result, quadcopters reach the configuration shown in Fig.~\ref{fig:sim9_cov}\subref{fig:sim9_cov_k1800} at $1800$\,step, a similar configuration achieved by the proposed method at $2100$\,step. 
In return for this slightly fast transition, the unsurveilled area keeps appearing in a team, which might overlook the region even if the density function at the area takes a large value, roughly 80\% of the highest peak of the density function, as shown in Fig.~\ref{fig:sim9_cov}\subref{fig:sim9_cov_k6500}.
This hole keeps remaining until the end of the simulation shown in Fig.~\ref{fig:sim9_cov}\subref{fig:sim9_cov_k10000}. 
Fig.~\ref{fig:sim9_obj} shows the evolution of the coverage cost in~\eqref{eq:pi_gradient_H}, where the attained coverage cost of the proposed method is lower than that of the coverage controller because the proposed method modifies the nominal input to prevent the appearance of holes. However, as shown in Fig.~\ref{fig:sim9}\subref{fig:sim9_k10000}, the proposed method successfully covers all three peaks of the density function without overlooking the important area in-between the team. This monitoring strategy prevents the overlooking of significant or safety-critical information in the monitoring mission.

\begin{figure*}[t!]%\label{fig:exp_snap}
    \centering
    \subfloat[Initial configuration]
    {\makebox[0.24\hsize][c]{\includegraphics[trim = 0cm 0cm 0cm 0.2cm, clip=true, width=0.24\linewidth]{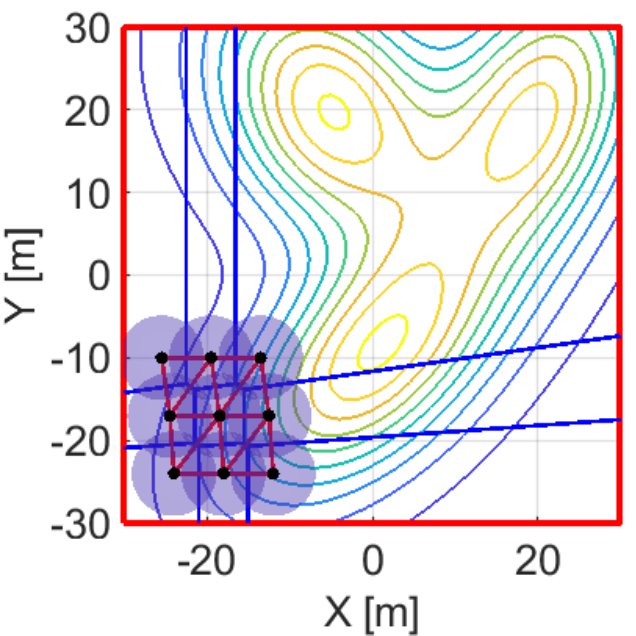}
    \label{fig:sim9_k1}}} \,
    %
    %\subfloat[]
    %{\makebox[0.48\hsize][c]{\includegraphics[width=0.48\linewidth]{figs/sim9/HNCBF_900.pdf}
    %\label{fig:sim9_k900}}} \\
    %
    \subfloat[$2100$ step]
    {\makebox[0.24\hsize][c]{\includegraphics[trim = 0cm 0cm 0cm 0.2cm, clip=true, width=0.24\linewidth]{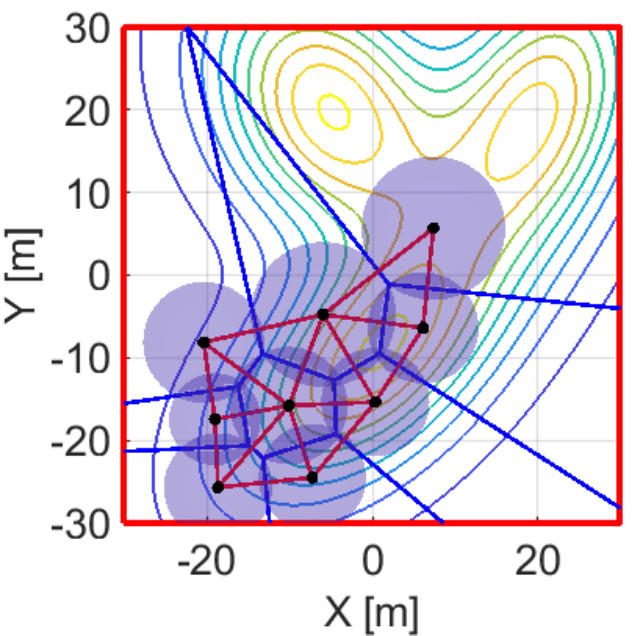}
    \label{fig:sim9_k2100}}} \,
    \subfloat[$4500$ step]
    {\makebox[0.24\hsize][c]{\includegraphics[trim = 0cm 0cm 0cm 0.2cm, clip=true, width=0.24\linewidth]{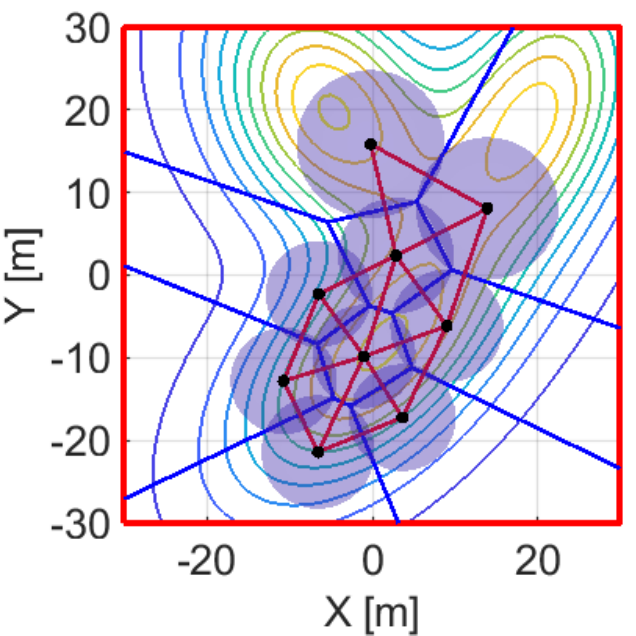}
    \label{fig:sim9_k4500}}} \,
    \subfloat[$10000$ step]
    {\makebox[0.24\hsize][c]{\includegraphics[trim = 0cm 0cm 0cm 0.2cm, clip=true, width=0.24\linewidth]{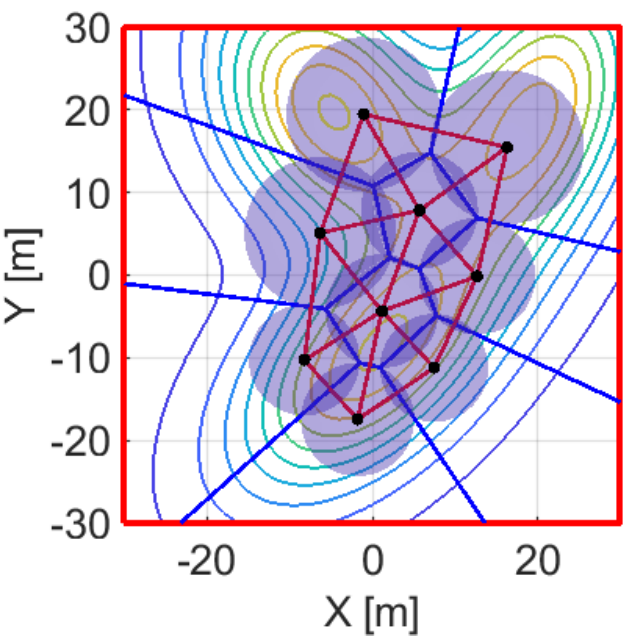}
    \label{fig:sim9_k10000}}}
    \caption{Snapshots of the simulation with the proposed algorithm, where the triangular subgraphs $\mc T_i, i\in \mc N$ in the graph $\mc G$ and the power diagram are depicted in red and blue lines, respectively.}
    \label{fig:sim9} 
\end{figure*}

\begin{figure*}[t!]%\label{fig:exp_snap}
    \centering
    \subfloat[$500$ step]
    {\makebox[0.24\hsize][c]{\includegraphics[trim = 0cm 0cm 0cm 0.2cm, clip=true, width=0.24\linewidth]{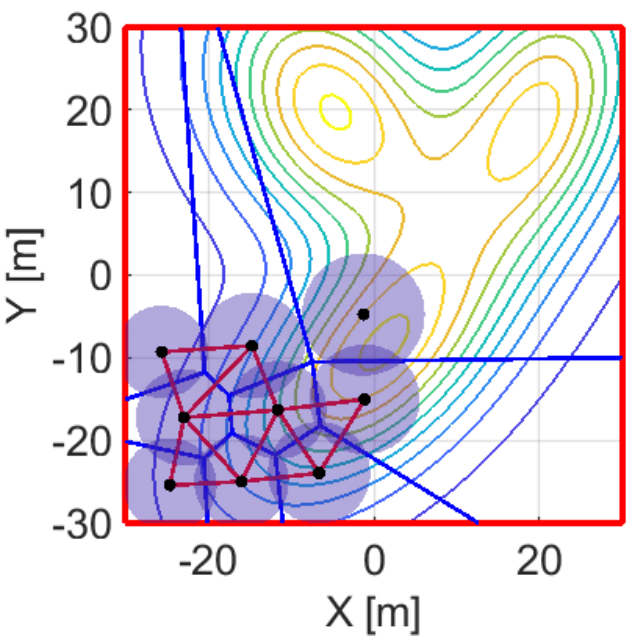}
    \label{fig:sim9_cov_k500}}} \,
    %
    %\subfloat[]
    %{\makebox[0.48\hsize][c]{\includegraphics[width=0.48\linewidth]{figs/sim9/HNCBF_900.pdf}
    %\label{fig:sim9_k900}}} \\
    %
    \subfloat[$1800$ step]
    {\makebox[0.24\hsize][c]{\includegraphics[trim = 0cm 0cm 0cm 0.2cm, clip=true, width=0.24\linewidth]{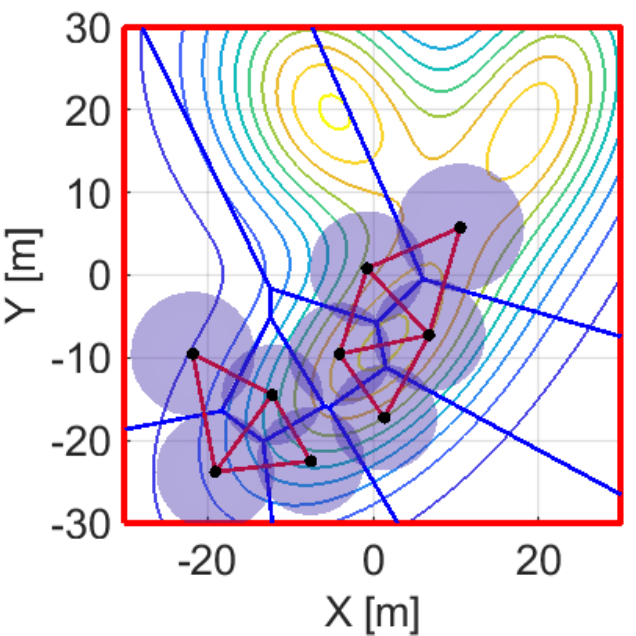}
    \label{fig:sim9_cov_k1800}}} \,
    \subfloat[$6500$ step]
    {\makebox[0.24\hsize][c]{\includegraphics[trim = 0cm 0cm 0cm 0.2cm, clip=true, width=0.24\linewidth]{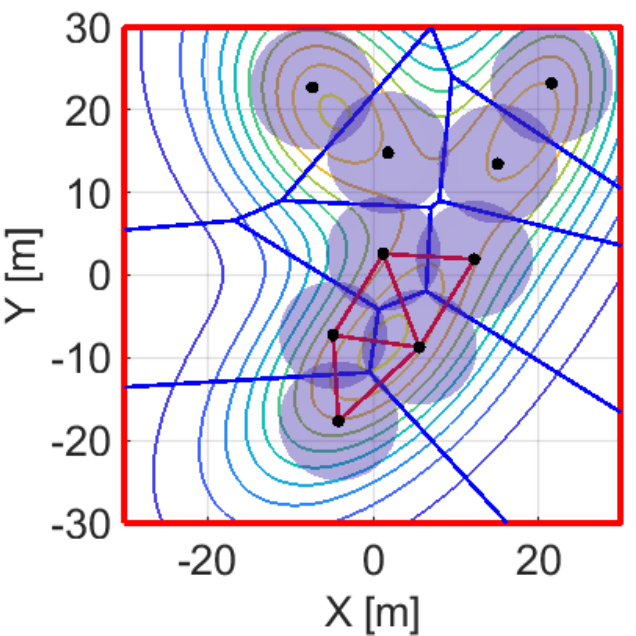}
    \label{fig:sim9_cov_k6500}}} \,
    \subfloat[$10000$ step]
    {\makebox[0.24\hsize][c]{\includegraphics[trim = 0cm 0cm 0cm 0.2cm, clip=true, width=0.24\linewidth]{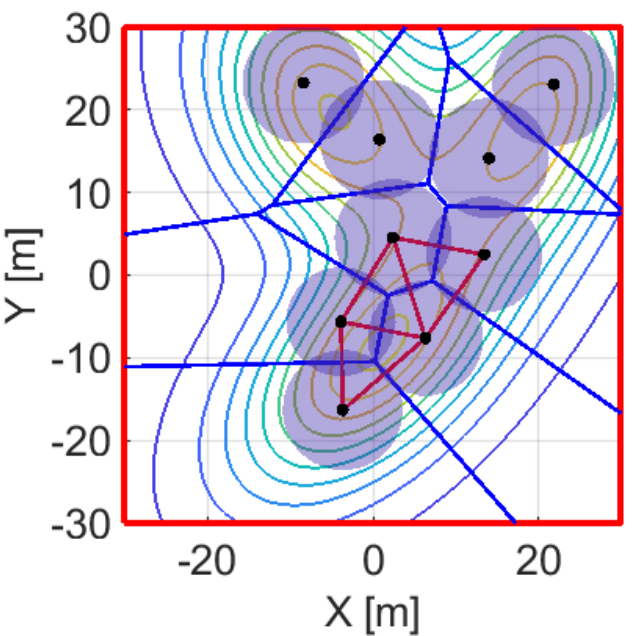}
    \label{fig:sim9_cov_k10000}}}
    \caption{Snapshots of the simulation with the coverage control $u_{i, \rm{nom}}, i\in \mc N$.}
    \label{fig:sim9_cov} 
\end{figure*}

\begin{figure}[t!]
    \centering
    \includegraphics[trim = 0cm 0.1cm 0cm 0.5cm, clip=true, width=0.7\linewidth]{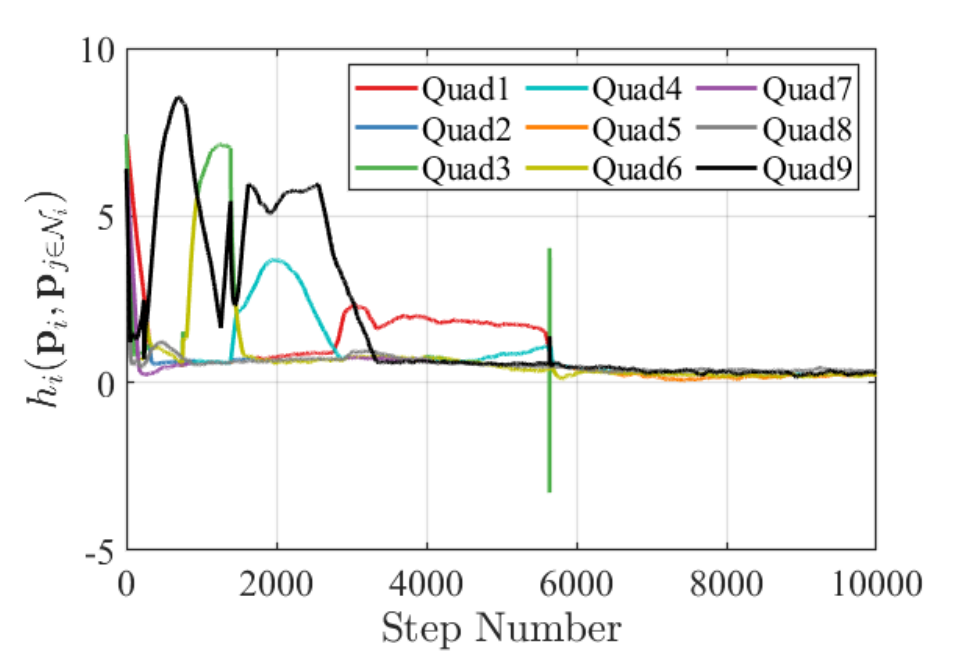}
    \caption{Evolution of the HNCBFs. Note that, to make it easier to understand visually, we only plot the minimum NCBF among triangular subgraphs to which each quadcopter belongs, namely $\min_{\Delta_\iota \in \mc T_i}h_{i,\Delta_\iota}, \forall i \in \mc N$.} %Algorithm~\ref{alg:cov_main} keeps all HNCBFs in the positive value except for step DEF, at which a newly emerged triangular subgraph creates a hole, as mentioned at the end of the section. Note that this instantaneously appeared hole does not violate the hybrid forward invariance of the set in the sense of DEF.
%    }
    \label{fig:sim9_HNCBF}
\end{figure}

\begin{figure}[t!]
    \centering
    \includegraphics[trim = 0cm 0.1cm 0cm 0.1cm, clip=true, width=0.68\linewidth]{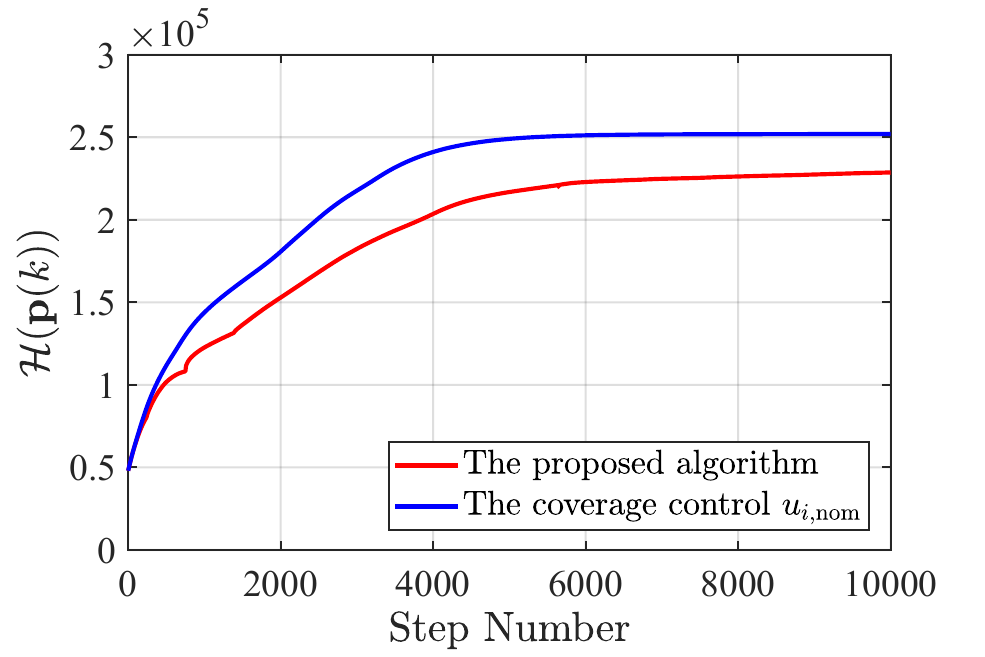}
    \caption{Evolution of the coverage cost in \eqref{eq:locational_cost}. As the price to pay for preventing a hole, the proposed algorithm takes a slightly smaller value than the coverage control $u_{i, {\rm nom}}$.}
    \label{fig:sim9_obj}
\end{figure}

\subsection{Experimental Results with Five Quadcopters}

This section demonstrates the proposed algorithm through experiments while comparing it with the results of the coverage control $u_{i, {\rm nom}},~\forall i \in \mc N$ in \eqref{eq:pi_gradient_H}.
%In this section, we demonstrate the proposed algorithm through experiments with comparing the results of the nominal input $u_{\rm nom}$. 
The size of the experimental field $\mc Q$ is $3\times 4\,{\rm m}^2$. 
We implement the control law to a team of quadcopters composed of five Bitcraze Crazyflie 2.1 ($n=5$), where their FOVs are virtually set. % so that $r=$ holds. 
Each parameter is set as $\kappa = 4$, $\sigma = 1$, $M = 0.7$, $w=0.2$, $w_\lam = 1.0\times 10^6$, $\epsilon = 0.2$ with the extended class-$\mc K$ function $\alpha(h) = 20h^3$.

The schematic of the testbed is illustrated in Fig.~\ref{fig:exp_sys}, which consists of two desktop computers, a motion capture system (OptiTrack), and quadcopters. The motion capture system is composed of eleven PrimeX13 and the Windows PC installed with the tracking system Motive. 
This motion capture system captures the position of drones every $10$\,ms and sends it to the Linux PC (Ubuntu 20.04). 
In the Linux PC, the Crazyswarm platform \cite{Preiss17} allows us to receive the position information from the motion capture system. 
%\blue{Although the control architecture is implementable in a distributed manner, the current experimental system implements the proposed method in a Linux computer because of hardware limitations.}
This information is sent to MATLAB to calculate the control input. 
Although the control architecture is implementable in a distributed manner, the current experimental system implements the proposed method in a Linux computer because of hardware limitations.
Then, the calculated velocity input is sent to the drones through the Crazyswarm which manages a communication between the Linux PC and drones via the Crazyradio.
The computation time of velocity input for five drones in MATLAB is roughly $20$\,ms.
Note that the control of attitude is conducted by the off-the-shelf controller implemented in the Crazyflies~2.1 so that it follows a reference velocity with ensuring a sufficient update rate of the attitude control.
%\blue{Consequently, the update rate of the control input is around 30\,ms, in which the motion capture system occupies 10\,ms. This sampling time is regarded as a time delay and could cause oscillation. To avoid such unfavorable behavior, we set a smaller gain for the coverage control input than in the simulation, $\gamma = 1.5\times 10^{-4}$ and $\gamma_{\lambda} = 2\times 10^{-11}$.}

We run the proposed algorithm from the initial configuration in Fig.~\ref{fig:exp_HNCBF_snap}\subref{fig:exp_HNCBF1}, 
where Figs.~\ref{fig:exp_HNCBF_snap}\subref{fig:exp_HNCBF2}-\subref{fig:exp_HNCBF5} show the snapshots of subsequent results. 
At $t=3$\,s, two quadcopters near the peak of the density function move fast to the important area, as shown in Fig.~\ref{fig:exp_HNCBF_snap}\subref{fig:exp_HNCBF2}. 
Although this transition makes the relative distance between each quadcopter larger, the appearance of a hole is prevented by expanding their FOVs. 
Then, at $t=20$\,s in Fig.~\ref{fig:exp_HNCBF_snap}\subref{fig:exp_HNCBF3}, a team reaches the important area followed by the configuration at $t=60$\,s in Fig.~\ref{fig:exp_HNCBF_snap}\subref{fig:exp_HNCBF4}, which allows a drone to capture the peak of the density at the center of FOV.
Toward the final configuration shown in Fig.~\ref{fig:exp_HNCBF_snap}\subref{fig:exp_HNCBF5}, two quadcopters on the bottom are going to combine their FOVs, resulting in a new edge in the graph.

\begin{figure}[t!]
    \centering
    \includegraphics[width=0.75\linewidth]{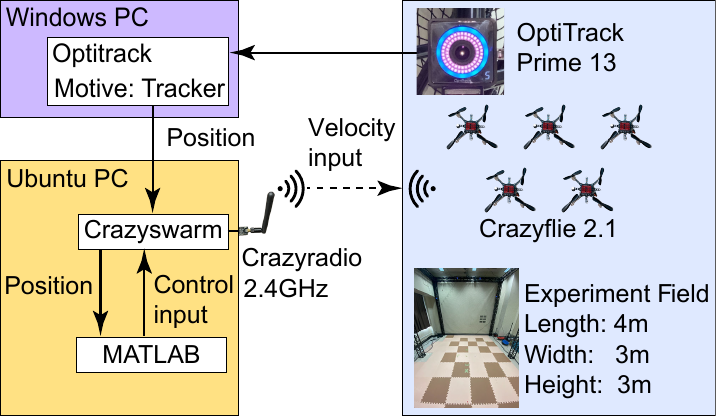}
    \caption{The schematic description of the experimental system.}
    \label{fig:exp_sys}
\end{figure}

\begin{figure}[t!]%\label{fig:exp_snap}
    \centering
    \subfloat[Initial configuration]
    {\makebox[0.98\hsize][c]{\includegraphics[width=0.98\linewidth]{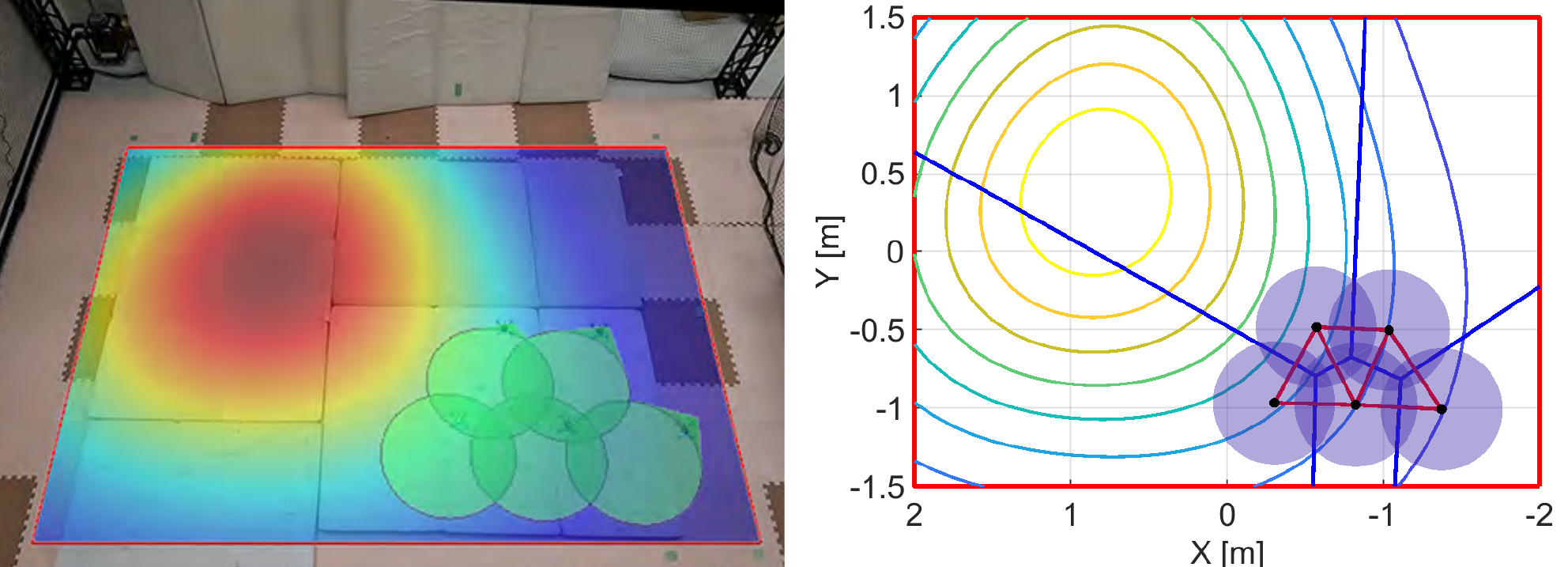}
    \label{fig:exp_HNCBF1}}} \\
    \subfloat[$t= 3$\,s with the proposed algorithm]
    {\makebox[0.98\hsize][c]{\includegraphics[width=0.98\linewidth]{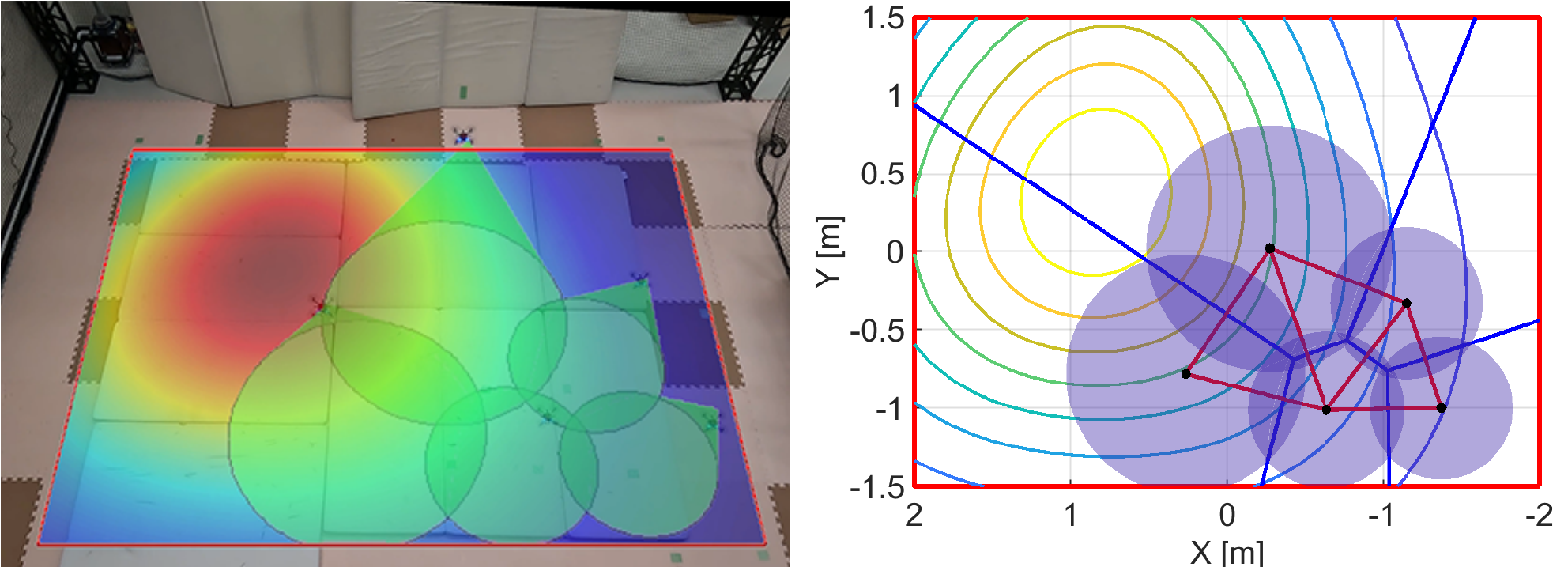}
    \label{fig:exp_HNCBF2}}} \\
    \subfloat[$t= 20$\,s with the proposed algorithm]
    {\makebox[0.98\hsize][c]{\includegraphics[width=0.98\linewidth]{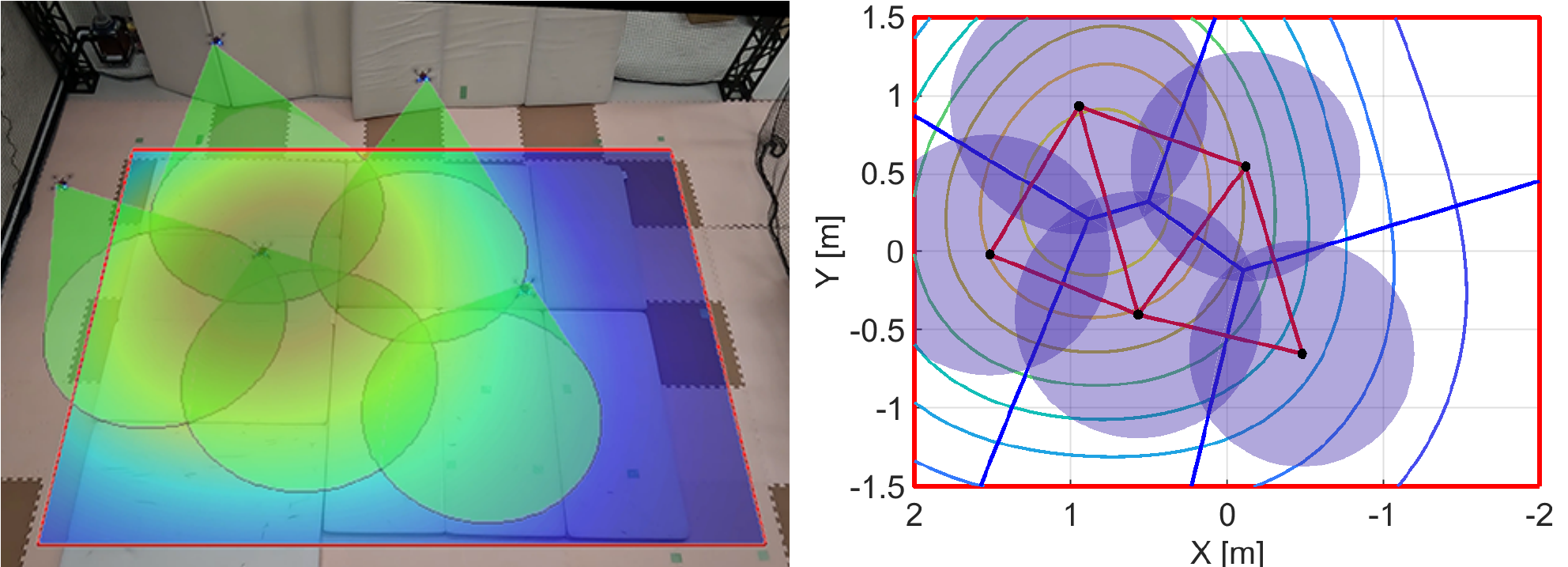}
    \label{fig:exp_HNCBF3}}} \\
    \subfloat[$t=60$\,s with the proposed algorithm]
    {\makebox[0.98\hsize][c]{\includegraphics[width=0.98\linewidth]{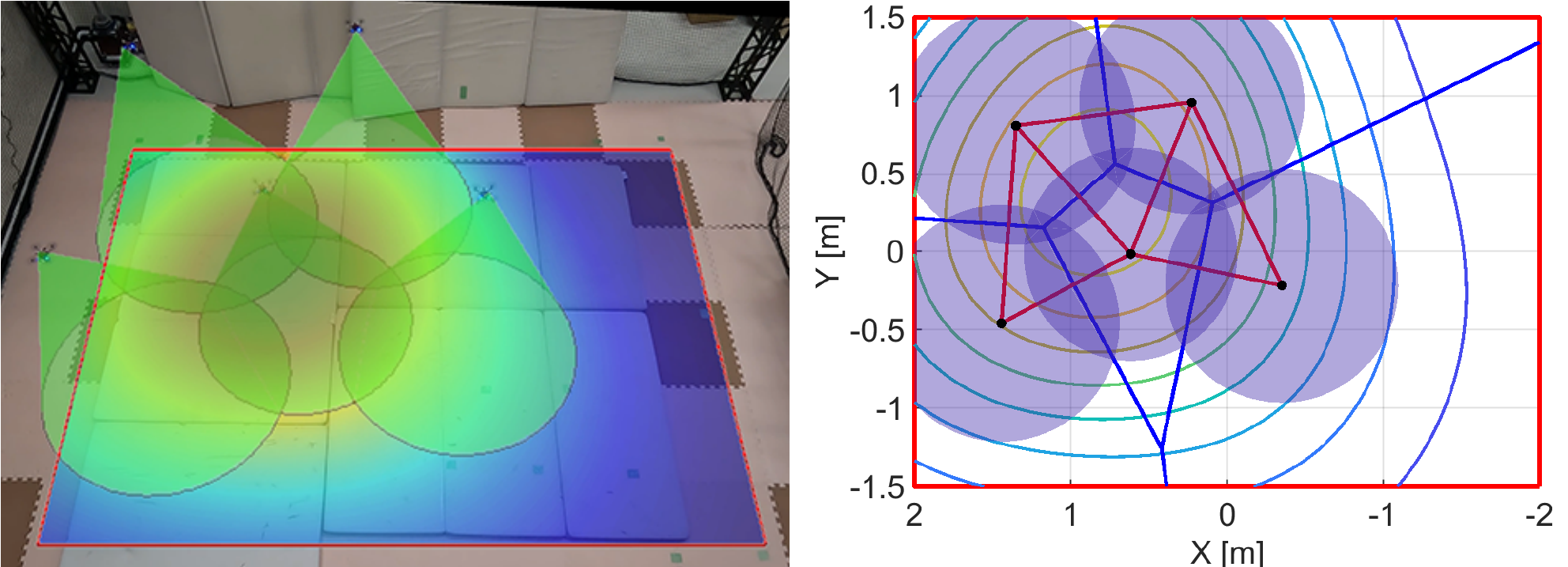}
    \label{fig:exp_HNCBF4}}}\\
    \subfloat[$t=85$\,s with the proposed algorithm]
    {\makebox[0.98\hsize][c]{\includegraphics[width=0.98\linewidth]{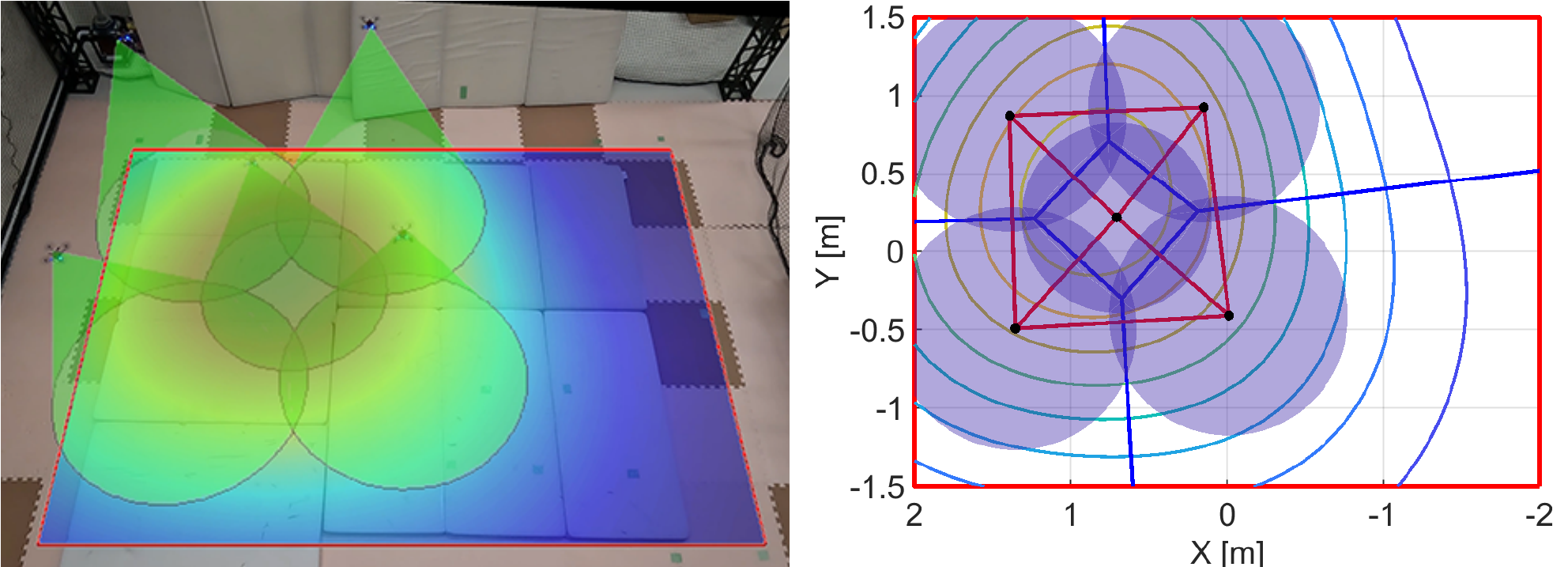}
    \label{fig:exp_HNCBF5}}}
    \caption{Snapshots of the experiments with the proposed algorithm.}
    \label{fig:exp_HNCBF_snap}
\end{figure}

\begin{figure}[t!]%\label{fig:exp_snap}
    \centering
    \subfloat[$t=3$\,s with the coverage control $u_{i, {\rm nom}},~i\in\mc N$ only]
    {\makebox[0.98\hsize][c]{\includegraphics[width=0.98\linewidth]{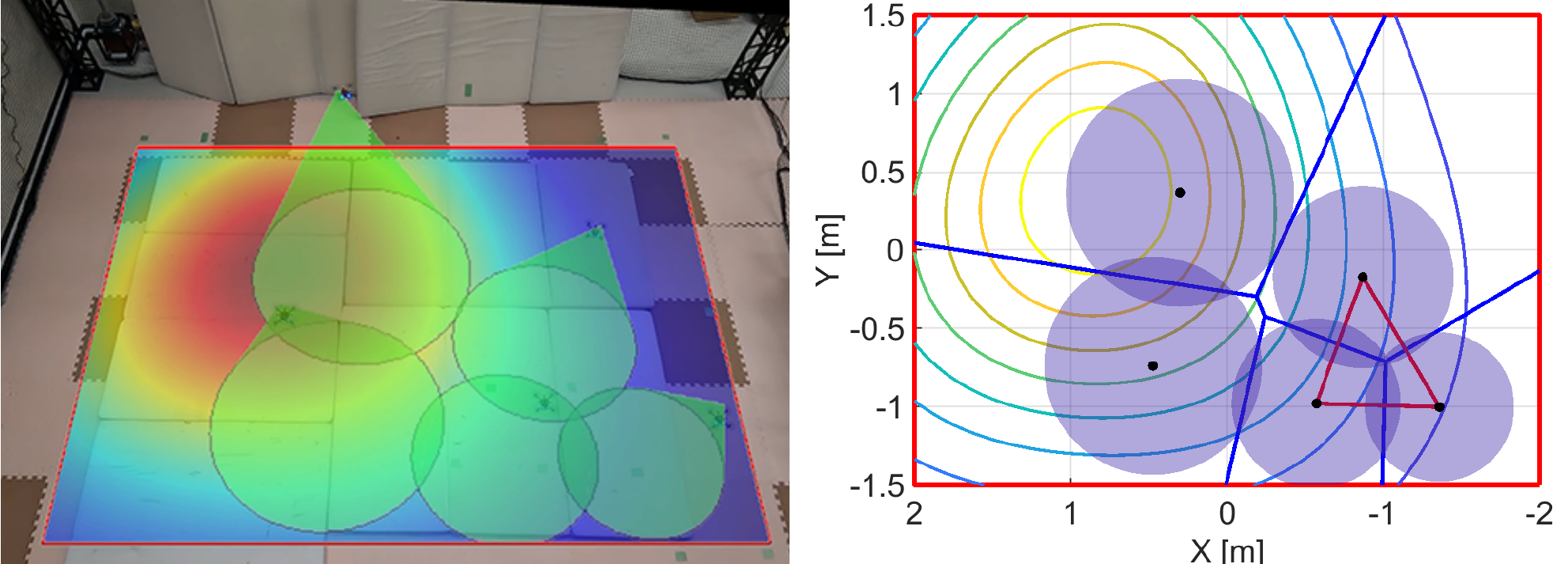}
    \label{fig:exp_nominal1}}} \\
    \subfloat[$t=6$\,s with the coverage control $u_{i, {\rm nom}},~i\in\mc N$ only]
    {\makebox[0.98\hsize][c]{\includegraphics[width=0.98\linewidth]{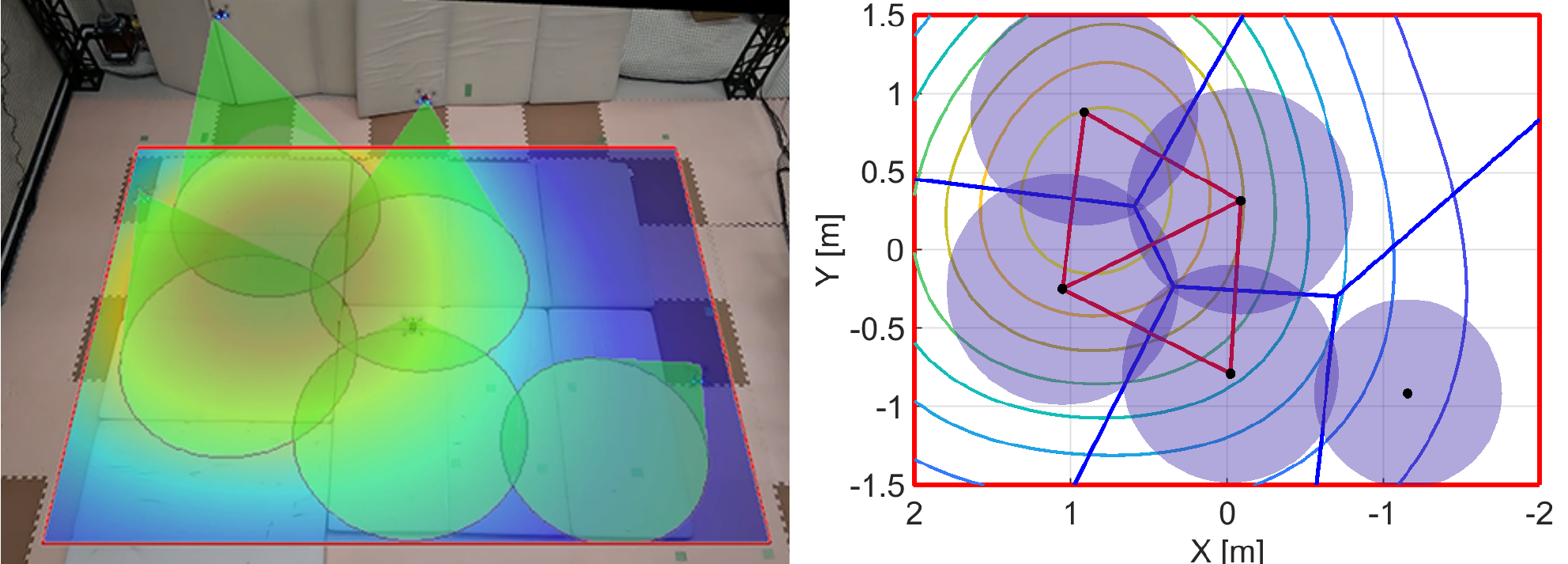}
    \label{fig:exp_nominal2}}} \\
    \subfloat[$t=12$\,s with the coverage control $u_{i, {\rm nom}},~i\in\mc N$ only]
    {\makebox[0.98\hsize][c]{\includegraphics[width=0.98\linewidth]{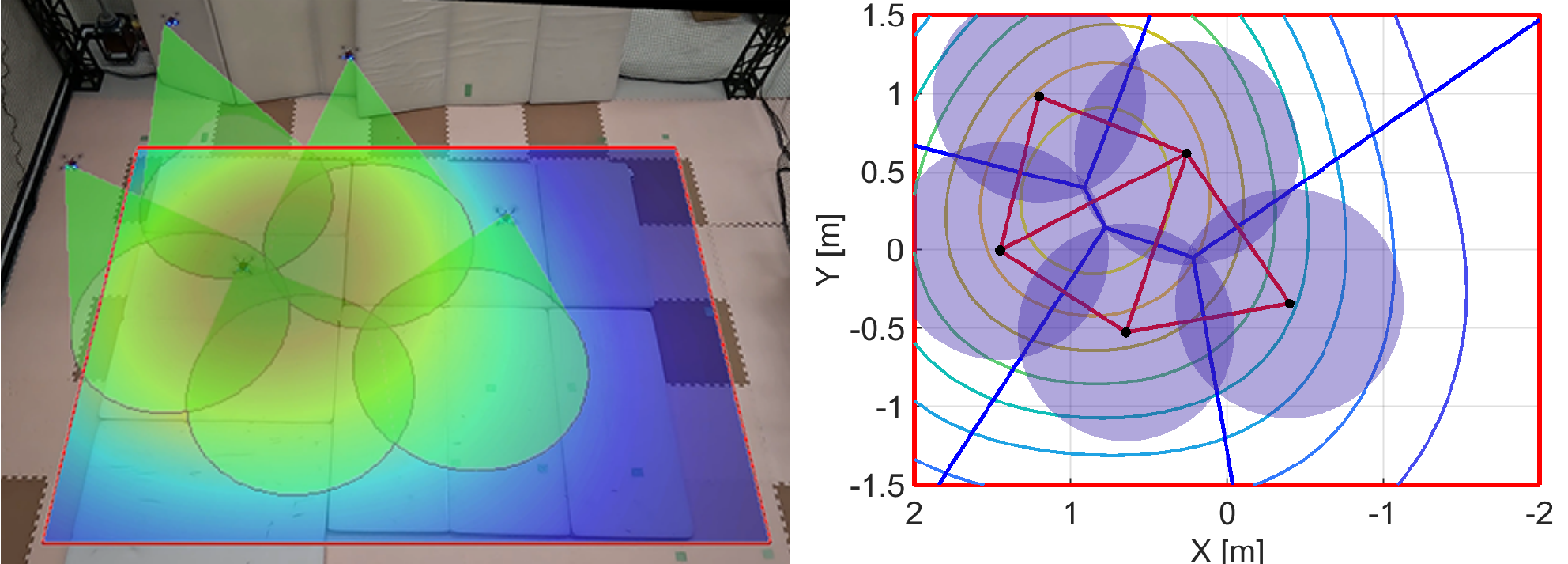}
    \label{fig:exp_nominal3}}} \\
    \subfloat[$t=85$\,s with the coverage control $u_{i, {\rm nom}},~i\in\mc N$ only]
    {\makebox[0.98\hsize][c]{\includegraphics[width=0.98\linewidth]{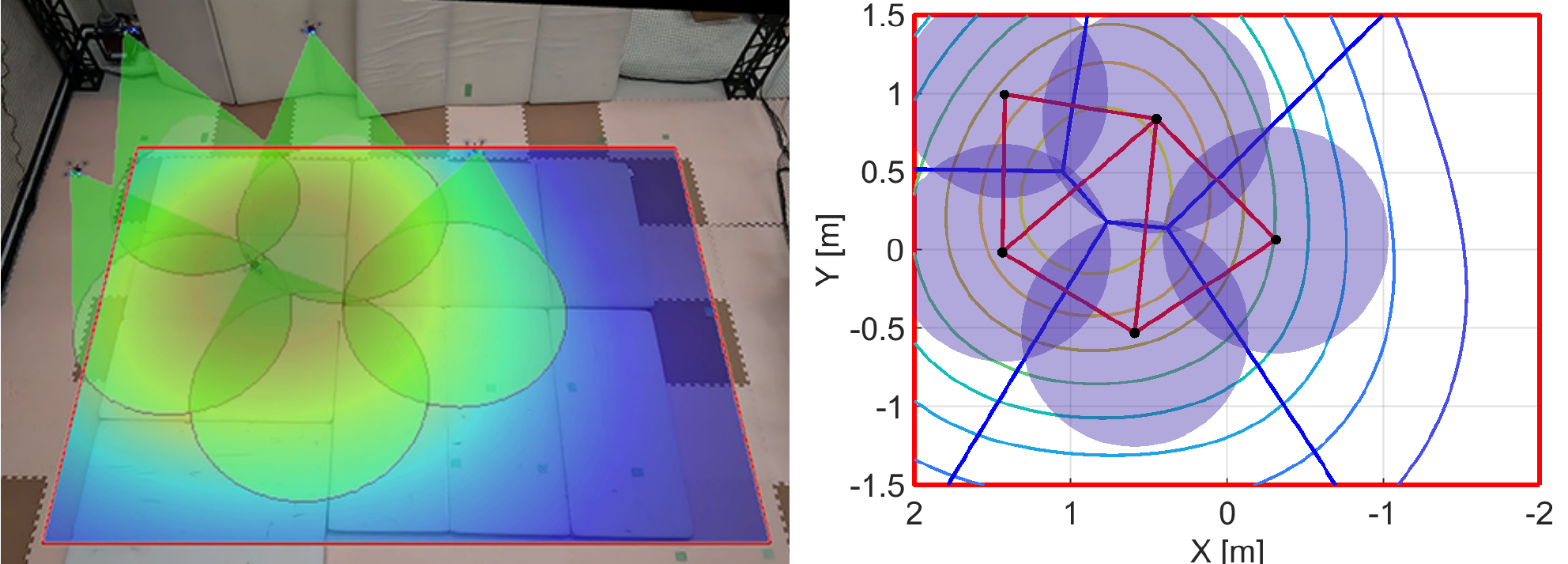}
    \label{fig:exp_nominal4}}}
    \caption{Snapshots of the experiment with the coverage control $u_{i, {\rm nom}}$, $i \in\mc N$.}
    \label{fig:exp_nominal}
\end{figure}

\begin{figure}
    \centering
    \includegraphics[trim = 0cm 0.1cm 0cm 0.1cm, clip=true, width=0.75\linewidth]{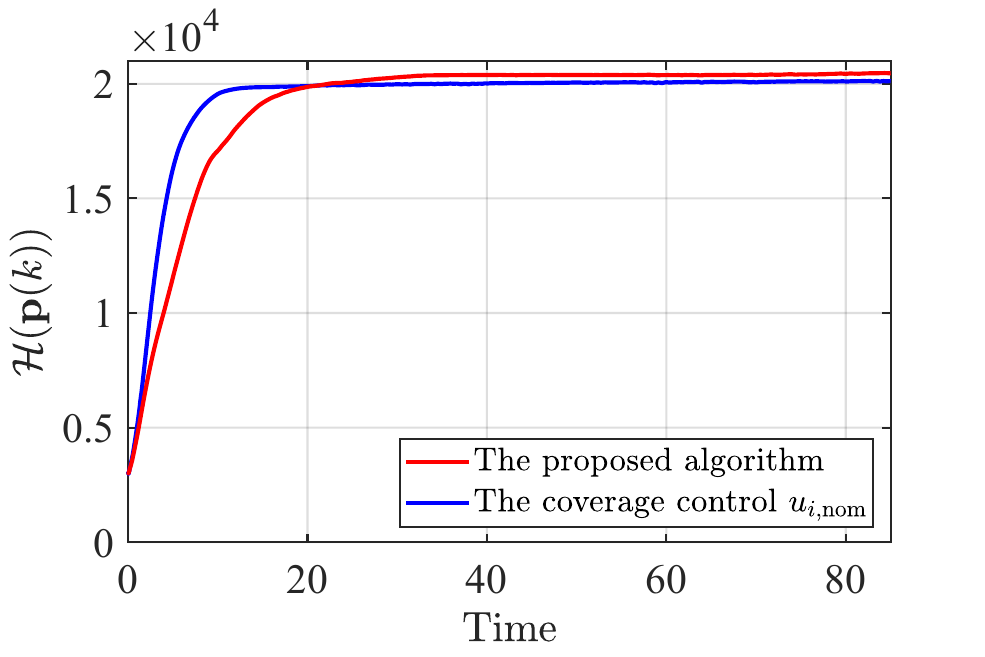}
    \caption{Evolution of the cost function \eqref{eq:locational_cost}. The proposed algorithm achieves a slightly larger coverage quality than that of the coverage control $u_{i, {\rm nom}}$.}
    \label{fig:exp_obj}
\end{figure}

\begin{figure}
    \centering
    \includegraphics[trim = 0cm 0.1cm 0cm 0.6cm, clip=true, width=0.75\linewidth]{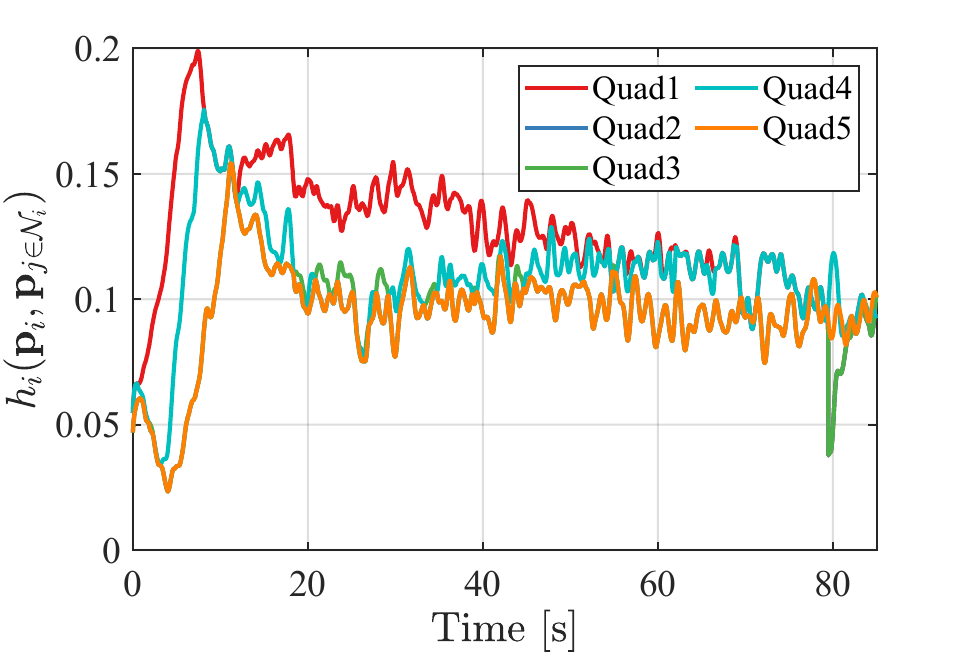}
    \caption{Evolution of the HNCBFs. Note that, to make it easier to understand visually, we only plot the minimum NCBF among triangular subgraphs to which each quadcopter belongs, namely $\min_{\Delta_\iota \in \mc T_i}h_{i,\Delta_\iota}, \forall i \in \mc N$. Since all $\min_{\Delta_\iota \in \mc T_i}h_{i,\Delta_\iota}, \forall i \in \mc N$ take a positive value, we can confirm that Algorithm~\ref{alg:cov_main} can make NCBFs evaluated in \eqref{eq:QP_nohole} positive. %\blue{will be replaced later after fixing up the appearance of the figure}
    }
    \label{fig:exp_HNCBF}
\end{figure}

Fig.~\ref{fig:exp_nominal} shows the results with the nominal input $u_{i,{\rm nom}},~\forall i\in \mc N$ only, where the initial configuration is the same as Fig.~\ref{fig:exp_HNCBF_snap}\subref{fig:exp_HNCBF1}. 
We can confirm from Fig.~\ref{fig:exp_nominal}\subref{fig:exp_nominal1} that a hole appears in the middle of the team at $t=3$\,s. 
Fig.~\ref{fig:exp_nominal}\subref{fig:exp_nominal2} shows the quadcopters located on the right do not preserve the triangular subgraph as opposed to the results with the proposed algorithm. 
The convergence speed is a bit shorter than that of the proposed method, and the quadcopters reach a peak of the density function at $t=12$\,s as shown in Fig.~\ref{fig:exp_nominal}\subref{fig:exp_nominal3}. 
The team keeps the almost same configuration for the rest of the experiment while making small holes. The final configuration is shown in Fig.~\ref{fig:exp_nominal}\subref{fig:exp_nominal4}.

The evolution of the coverage cost \eqref{eq:locational_cost} is shown in Fig.~\ref{fig:exp_obj}. The nominal input allows the team to achieve the maximum at $t=12$\,s, while the proposed algorithm reaches the same value at $t=20$\,s. 
Remarkably, even though the proposed algorithm modifies the nominal controller maximizing the coverage cost to prevent holes, the team attains the same coverage cost as the nominal coverage controller solely aimed at maximizing coverage cost. This result implies that the proposed algorithm does not overly restrict the behavior of the team.
Furthermore, differently from the coverage control only, the proposed method keeps increasing the objective function gradually. 
The proposed method attains a larger maximum because its property to connect each FOV leads to the configuration where one of the drones captures the peak of the density function at the center of its FOV, as illustrated in Fig.~\ref{fig:exp_HNCBF_snap}\subref{fig:exp_HNCBF5}. 
In contrast, a drone with the coverage control in \eqref{eq:pi_gradient_H} repels other agents' FOVs to minimize their overlaps while trying to reach the most important area. 
This property might make a team get stuck into a deadlock before one of them captures the peak of the density function, as shown in Fig.~\ref{fig:exp_nominal}\subref{fig:exp_nominal4}. Fig.~\ref{fig:exp_HNCBF} shows the evolution of HNCBFs in the proposed method, where we observe that they keep positive values throughout the experiment.
%where it can be observed that they remains positive throughout the experiment. 
At $t=80$\,s, there is a jump in the value of HNCBF, which is caused by the addition of the edge in the graph from the configuration in Fig.~\ref{fig:exp_HNCBF_snap}\subref{fig:exp_HNCBF4} to Fig.~\ref{fig:exp_HNCBF_snap}\subref{fig:exp_HNCBF5}.

\section{Conclusion}
In this paper, we proposed a distributed coverage algorithm that prevents the formation of unsurveyed regions in-between FOVs of quadcopters. The necessary and sufficient condition for eliminating holes among trios of a team was introduced from the property of the power diagram. This condition, encoded with a Boolean composition, was then incorporated as NCBFs to design the algorithm guaranteeing the hybrid forward invariance of the set, preventing the appearance of holes. From the symmetric property of the designed NCBFs, the proposed algorithm can be implementable in a distributed fashion. 
We guaranteed that the proposed algorithm can find a control input preventing a hole, %satisfying the constraint induced by the proposed NCBFs, 
except for a few pathological conditions, by analyzing the derivatives of each CBF composing NCBFs.
The proposed algorithm was synthesized with the coverage control law that improves monitoring quality while reducing the overlaps of FOVs. 
The effectiveness of the proposed method was demonstrated via simulations and experiments.

There are several challenges to explore in the future. 
First, the robustness of the proposed method in the presence of uncertainty should be investigated further to be able to deploy the algorithm in more challenging scenarios, including outdoor environments. Foreseeable challenges include uncertainty in the position estimation and in inter-robot communications, with possible remedies entailing considering a conservative monitoring strategy or preparing a redundant system to mitigate possible failures, as discussed in Remarks~\ref{rem:pos_unc} and~\ref{rem:com_unc}. Finally, the extension to the other shapes of FOVs needs to be further investigated.

\appendices

\section{Proof of Lemma~\ref{thm:h_f}} \label{ap:hiF}
\begin{proof}
If $L_g h_{i,\Delta_{\iota},\mc F}^{\mrm k} \neq {\bf 0}_{4\times 1}$ holds, then there exists $u_i \in \R^4$, which satisfies that $L_f h_{i,\Delta_{\iota},\mc F}^{\mrm k} + L_g h_{i,\Delta_{\iota},\mc F}^{\mrm k}u_i + \alpha(h_{i,\Delta_{\iota},\mc F}^{\mrm k}) \geq 0$. %\cite{XU2015}. 
Although the discussion in this appendix presumes the coordinate frame $\Sigma_d$, we will drop the superscript $\circ^d$.

The derivative of $h_{i,\Delta_{\iota},\mc F}^{\mrm k}$ with respect to $\bfp_i$ is given by
\begin{subequations} \label{eqs:def_hiF}
\begin{align}
\frac{\partial h_{i,\Delta_{\iota},\mc F}^{\mrm k}}{\partial x_i} \!&=\! -2\left( 
V_{ijk}^x \left( \frac{\partial v_{ijk}^x}{\partial x_i} \!-\! 1 \right) \!+\!V_{ijk}^y \frac{\partial v_{ijk}^y}{\partial x_i} \right), \\
\frac{\partial h_{i,\Delta_{\iota},\mc F}^{\mrm k}}{\partial y_i} \!&=\! -2\left( 
V_{ijk}^x \frac{\partial v_{ijk}^x}{\partial y_i} 
\!+\!V_{ijk}^y \left( \frac{\partial v_{ijk}^y}{\partial y_i} \!-\! 1 \right) \right), \\
\frac{\partial h_{i,\Delta_{\iota},\mc F}^{\mrm k}}{\partial z_i} \!&=\! 2\left( \frac{r^2 z_i}{\lam_i^2}
\!-\! \left( V_{ijk}^x \frac{\partial v_{ijk}^x}{\partial z_i} \!+\! V_{ijk}^y \frac{\partial v_{ijk}^y}{\partial z_i} \right) \right),\\
\frac{\partial h_{i,\Delta_{\iota},\mc F}^{\mrm k}}{\partial \lam_i} &\!=\! -2 \left( \frac{\left(rz_i\right)^2}{\lam_i^3} \!+\! 
\left( V_{ijk}^x \frac{\partial v_{ijk}^x}{\partial \lam_i} \!+\! V_{ijk}^y \frac{\partial v_{ijk}^y}{\partial \lam_i} \right) \!\right),
\end{align}
\end{subequations}
with $V_{ijk}^x = v_{ijk}^x - x_i$ and $V_{ijk}^y = v_{ijk}^y - y_i$. 
%Since the radical center $v_{ijk}$ is restricted on the radical axis $L_{jk}$, namely on the $Y$-axis, in the introduced coordinate system $\Sigma_d$, $v_{ijk}^x = 0$ and $\partial v_{ijk}^x / \partial x_i = 0$ always hold.
Since the equations \eqref{eq:sig_d_vx} and \eqref{eq:sig_d_vy} hold in the introduced coordinate frame $\Sigma_d$, we can simplify \eqref{eqs:def_hiF} by substituting these conditions as  
%By substituting these two conditions to \eqref{eqs:def_hiF}, we obtain
\begin{subequations} \label{eqs:def_hiF_simp}
\begin{align}
\frac{\partial h_{i,\Delta_{\iota},\mc F}^{\mrm k}}{\partial x_i} &= -2x_i - 2\left(v_{ijk}^y - y_i\right)\frac{\partial v_{ijk}^y}{\partial x_i},\\
\frac{\partial h_{i,\Delta_{\iota},\mc F}^{\mrm k}}{\partial y_i} &= -2\left(v_{ijk}^y - y_i\right)\left(\frac{\partial v_{ijk}^y}{\partial y_i} - 1\right), \label{eq:hify2}\\
\frac{\partial h_{i,\Delta_{\iota},\mc F}^{\mrm k}}{\partial z_i} &= 2r^2\frac{z_i}{\lambda_i^2} - 2\left(v_{ijk}^y - y_i\right)\frac{\partial v_{ijk}^y}{\partial z_i}, \label{eq:hifz2}\\
\frac{\partial h_{i,\Delta_{\iota},\mc F}^{\mrm k}}{\partial \lambda_i} &= -2r^2\frac{z_i^2}{\lambda_i^3} - 2\left(v_{ijk}^y - y_i\right)\frac{\partial v_{ijk}^y}{\partial \lambda_i}.
\end{align}
\end{subequations}
The partial derivatives of $v_{ijk}^y$ contained in \eqref{eqs:def_hiF_simp} can be derived from $\eqref{eq:radical_cent_y}$ as
%From \eqref{eqs:def_hiF_simp}, the partial derivatives of $v_{ijk}^y$ in \eqref{eqs:def_hiF_simp} are calculated as
\begin{subequations} \label{eqs:def_vy}
\begin{align}
    \frac{\partial v_{ijk}^y}{\partial x_i} &= \frac{x_i}{y_i}, \label{eq:dvx}\\
    \frac{\partial v_{ijk}^y}{\partial y_i} &= \frac{1}{2} -\frac{(x_i^2 - R_i^2) - (x_j^2 - R_j^2)}{2y_i^2}, \label{eq:dvy}\\
    \frac{\partial v_{ijk}^y}{\partial z_i} &= -\frac{r^2z_i}{y_i\lambda_i^2}, \label{eq:dvz}\\
    \frac{\partial v_{ijk}^y}{\partial \lambda_i} &= \frac{r^2z_i^2}{y_i\lambda_i^3},  \label{eq:dvl}
\end{align}
\end{subequations}
where $y_i \neq 0$ from Assumption~\ref{asm:non_parallel}. By combining \eqref{eqs:def_hiF_simp} and \eqref{eqs:def_vy} together, we obtain
\begin{subequations} \label{eqs:def_hiF_result}
\begin{align}
    \frac{\partial h_{i,\Delta_{\iota},\mc F}^{\mrm k}}{\partial x_i} &= -2x_i\frac{v_{ijk}^y}{y_i}, \label{eq:dhifx}\\
    \frac{\partial h_{i,\Delta_{\iota},\mc F}^{\mrm k}}{\partial y_i} &= \frac{(x_i^2 - y_i^2 - R_i^2) - (x_j^2 - R_j^2)}{y_i^2}v_{ijk}^y, \label{eq:dhify}\\
    \frac{\partial h_{i,\Delta_{\iota},\mc F}^{\mrm k}}{\partial z_i} &= 2r^2\frac{z_i}{\lambda_i^2}\frac{v_{ijk}^y}{y_i}, \label{eq:dhifz}\\
    \frac{\partial h_{i,\Delta_{\iota},\mc F}^{\mrm k}}{\partial \lambda_i} &=-2r^2\frac{z_i^2}{\lambda_i^3}\frac{v_{ijk}^y}{y_i}. \label{eq:dhifl}
\end{align}
\end{subequations}
The equations \eqref{eqs:def_hiF_result} reveal that $\partial h_{i,\Delta_{\iota},\mc F}^{\mrm k} / \partial \bfp_i = {\bf 0}_{4 \times 1}$ never happens except for $v_{ijk}^y=0$ in $\Sigma_d$, which means the radical center $v_{ijk}$ lies on the line $JK$.
This completes the proof.
%which is a non-stationary point.
%The equations \eqref{eqs:def_hiF_result} reveal that the condition satisfying $\partial h_{i, \mc F} / \partial \bfp_i = 0_{4 \times 1}$ is limited only when $v_{ijk}^y=0$ is satisfied under Assumption~\ref{asm:non_parallel}, $z_i>0$, and $\lam_i > 0$ hold.
\end{proof}

\section{Proof of Lemma~\ref{thm:h_JKI}} \label{ap:h_JKI}
\begin{proof}
%Similar with the definition of $h_{IJK}$ in \eqref{eq:h_IJK}, $h_{JKI} = A_{JKv_{ijk}} / A_{IJK}$.
We will drop the superscript $\circ^d$ in this appendix though the discussion is provided in $\Sigma_d$.
Similar with $h_{i,\Delta_{\iota},IJK}^{\mrm k}$ in \eqref{eq:h_IJK}, $h_{i,\Delta_{\iota},JKI}^{\mrm k}$ takes positive value if and only if $y_i$ and $v_{ijk}$ exist in the same half-plane divided by the line $JK$, which is regarded as $x$-axis in $\Sigma_d$. 
Noticing that $\triangle{JKv_{ijk}}$ and $\triangle{IJK}$ share the side $JK$ as in Fig.~\ref{fig:coord_tri}\subref{fig:JKv}, the following equations hold with denoting the area of $\triangle{IJK}$ as $A_{IJK}$.
\begin{align}
    h_{i,\Delta_{\iota},JKI}^{\mrm k} &= \frac{ \be_z^\top \left(\overrightarrow{IJ} \times \overrightarrow{Iv_{ijk}}\right)}
{\be_z^\top \left( \overrightarrow{IJ} \times \overrightarrow{IK} \right)}, \nonumber\\
&= \sgn(y_i v_{ijk}^y) \frac{A_{JKv_{ijk}}}{A_{IJK}}, \nonumber \\
&= \frac{v_{ijk}^y}{y_i}, \nonumber\\
&= \frac{(x_i^2 + y_i^2 - R_i^2) - (x_j^2 - R_j^2)}{2y_i^2}, \label{eq:hIJK_simp}
\end{align}
where we substitute \eqref{eq:radical_cent_y} to derive \eqref{eq:hIJK_simp}.
Note that $y_i\neq 0$ under Assumption~\ref{asm:non_parallel}. %The partial derivatives of $h_{i,\Delta_{\iota},JKI}^{\mrm k}$ are then calculated as 
Then, we obtain
\begin{subequations} \label{eqs:def_hJKI_result}
\begin{align}
	\frac{\partial h_{i,\Delta_{\iota},JKI}^{\mrm k}}{\partial x_i} &= \frac{x_i}{y_i^2}, \label{eq:hJKIx0}\\
	\frac{\partial h_{i,\Delta_{\iota},JKI}^{\mrm k}}{\partial y_i} &= -\frac{(x_i^2 - R_i^2) - (x_j^2 - R_j^2)}{y_i^2}, \label{eq:hJKIy0}\\
	\frac{\partial h_{i,\Delta_{\iota},JKI}^{\mrm k}}{\partial z_i} &= -\frac{r^2z_i}{y_i^2\lambda_i^2}, \label{eq:hJKIz0}\\
	\frac{\partial h_{i,\Delta_{\iota},JKI}^{\mrm k}}{\partial \lambda_i} &= \frac{r^2z_i^2}{y_i^2\lambda_i^2}.  \label{eq:hJKIl0}
\end{align}
\end{subequations}
From \eqref{eqs:def_hJKI_result}, $\partial h_{i,\Delta_{\iota},JKI}^{\mrm k} / \partial \bfp_i = {\bf 0}_{4 \times 1}$ never occurs if $z_i>0$. This completes the proof.
\end{proof}

\begin{figure}[t!]%\label{fig:exp_snap}
    \centering
    \subfloat[]
    {\makebox[0.48\hsize][c]{\includegraphics[width=0.48\linewidth]{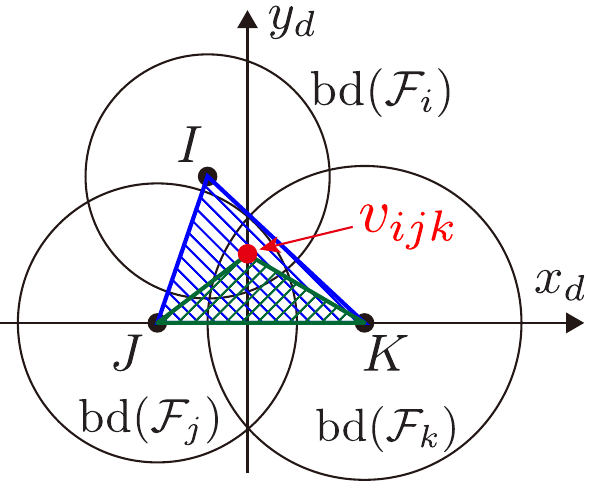}
    \label{fig:JKv}}} \quad 
    \subfloat[]
    {\makebox[0.48\hsize][c]{\includegraphics[width=0.48\linewidth]{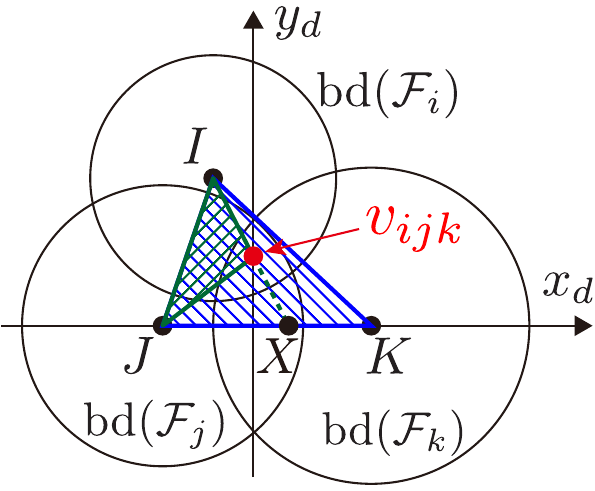}
    \label{fig:IJv}}} \quad
    \caption{The description of area ratio incorporated in (a) $h_{i,\Delta_{\iota},JKI}^{\mrm k}$ and (b) $h_{i,\Delta_{\iota},IJK}^{\mrm k}$, respectively. (a) illustrates the area ratio between $\triangle{IJK}$ and $\triangle{JKv_{ijk}}$, while (b) depicts the area ratio between $\triangle{IJK}$ and $\triangle{IJv_{ijk}}$.}
    \label{fig:coord_tri} 
\end{figure}

\section{Proof of Lemma~\ref{thm:h_IJK}} \label{ap:h_IJK}
\begin{proof}
We will drop the superscript $\circ^d$ in this appendix.
With the introduced coordinated system $\Sigma_d$, let us denote the $x$-intercept of the line $Iv_{ijk}$ as $X$, as illustrated in Fig.~\ref{fig:coord_tri}\subref{fig:IJv}.
Then, from the fact $\|h_{i,\Delta_{\iota},IJK}^{\mrm k}\| = A_{IJv_{ijk}} / A_{IJK}$,
\begin{align}
    h_{i,\Delta_{\iota},IJK}^{\mrm k} %&= \frac{A_{IJv_{ijk}}}{A_{IJK}} \nonumber \\
    = \frac{y_i - v_{ijk}^y}{y_i} \frac{X - x_j}{x_k - x_j}, \label{eq:hIJK1}
\end{align}
with
\begin{align}
    X = \frac{-v_{ijk}^y}{y_i - v_{ijk}^y}x_i. \label{eq:hIJK_X}
\end{align}
By substituting \eqref{eq:hIJK_X} into \eqref{eq:hIJK1}, we obtain
\begin{align}
    h_{i,\Delta_{\iota},IJK}^{\mrm k} \!=\! \frac{x_i + x_j}{2(x_k \!-\! x_j)} \!+\! \frac{(x_i \!-\! x_j)\left((x_i^2 \!-\! R_i^2) \!-\! (x_j^2 \!-\! R_j^2)\right)}{2(x_k - x_j)y_i^2} \label{eq:hIJK2}.
\end{align}
From \eqref{eq:hIJK2}, 
\begin{subequations} \label{eqs:def_hIJK_result}
\begin{align}
	\frac{\partial h_{i,\Delta_{\iota},IJK}^{\mrm k}}{\partial x_i}
    % &= \frac{1}{2(x_k - x_j)} + \frac{3x_i^2 - 2x_j x_i - r_i^2 - x_j^2 + r_j^2}{2(x_k - x_j)y_i^2} \nonumber\\
	&= \frac{3x_i^2 \!-\! 2x_j x_i \!+\!(y_i^2 \!-\! R_i^2) \!-\! (x_j^2 \!-\! R_j^2)}{2(x_k \!-\! x_j)y_i^2},  \label{eq:hIJKx0} \\
	\frac{\partial h_{i,\Delta_{\iota},IJK}^{\mrm k}}{\partial y_i} &= -\frac{(x_i \!-\! x_j)\left((x_i^2 \!-\! R_i^2) \!-\! (x_j^2 \!-\! R_j^2)\right)}{(x_k \!-\! x_j)y_i^3},  \label{eq:hIJKy0} \\
	\frac{\partial h_{i,\Delta_{\iota},IJK}^{\mrm k}}{\partial z_i} &= -\frac{(x_i \!-\! x_j)r^2z_i}{(x_k \!-\! x_j)y_i^2\lambda_i^2},  \label{eq:hIJKz0} \\
    \frac{\partial h_{i,\Delta_{\iota},IJK}^{\mrm k}}{\partial \lambda_i} &= \frac{(x_i \!-\! x_j)r^2z_i^2}{(x_k \!-\! x_j)y_i^2\lambda_i^3},  \label{eq:hIJKl0}
\end{align}
\end{subequations}
hold.
The equations \eqref{eqs:def_hIJK_result} reveal that $\partial h_{i,\Delta_{\iota},IJK}^{\mrm k} / \partial \bfp_i = {\bf 0}_{4 \times 1}$ never happens except for the following condition is satisfied.
\begin{align}
    \left(x_i^2 + y_i^2 - R_i^2 = x_j^2 - R_j^2\right) \quad \land \quad \left(x_i = x_j\right). \label{eq:condition_hIJK1}
\end{align}
From \eqref{eq:radical_cent_y}, the left condition of \eqref{eq:condition_hIJK1} corresponds with $v_{ijk}^y=0$. This signifies the radical axes $L_{ij}$ and $L_{ki}$ both pass through the origin. Since the right condition of \eqref{eq:condition_hIJK1} makes the line $JK$ and the radical axis $L_{ij}$ parallel, \eqref{eq:condition_hIJK1} is satisfied if and only if $L_{ij}$ lies on the line $JK$. This completes the proof.
%$v_{ijk}^y=0$, which means the radical center $v_{ijk}$ is on the line $JK$.
\end{proof}

\bibliographystyle{sty/IEEEtran}
\bibliography{biblio}

\end{document}